\pdfoutput=1
\AtBeginDocument{\date{8 Jul 2020}}

\documentclass[prd,amsmath,twocolumn,nofootinbib, eqsecnum]{revtex4-2}
\usepackage{amsthm}
\usepackage{amssymb}
\usepackage{graphicx}
\usepackage{color}
\usepackage{url}
\usepackage{enumerate}
\usepackage[bookmarksopen,bookmarksnumbered]{hyperref}

\newtheorem{theorem}{Theorem}
\newtheorem{lemma}{Lemma}
\newtheorem{definition}{Definition}

\def\xleft{\mathopen{}\left}

\DeclareRobustCommand*\diff[2][]{%
   \mathop{
     \mathrm{d}^{#1}
     \mskip-0.2\thinmuskip
    #2}\nolimits
}

\newcommand{\eqdef}{\stackrel{\textrm{def}}{=}}

\newcommand{\3}[1]{\boldsymbol{#1}}

\newcommand\Dset{\mathcal{D}}
\newcommand\notM{\widehat{M}}
\newcommand\Lset{\Lambda}

\newif\ifnote
\notetrue
\long\def\??#1{\textbf{\color{red}---~??{}#1~---}}
\long\def\note??#1{\ifnote\??{#1}\fi}

\begin{document}
\title{A new and complete proof of the Landau condition for pinch
  singularities of Feynman graphs and other integrals}
\author{John Collins}
\email{jcc8@psu.edu}
\affiliation{%
  Department of Physics, Penn State University, University Park PA 16802, USA}

\begin{abstract}
  The Landau equations give a physically useful criterion for how
  singularities arise in Feynman amplitudes.  Furthermore, they are
  fundamental to the uses of perturbative QCD, by determining the important
  regions of momentum space in asymptotic problems.  Generalizations are
  also useful. We will show that in existing treatments there are
  significant gaps in derivations, and in some cases implicit assumptions
  that will be shown here to be false in important cases like the massless
  Feynman graphs ubiquitous in QCD applications.  In this paper is given a
  new proof that the Landau condition is both necessary and sufficient for
  physical-region pinches in the kinds of integral typified by Feynman
  graphs.  The proof's range is broad enough to include the modified
  Feynman graphs that are used in QCD applications.  Unlike many existing
  derivations, there is no need to use the Feynman parameter method.  Some
  possible further applications of the new proof and its subsidiary results
  are proposed.
\end{abstract}

\maketitle

\section{Introduction}

The subject of this paper is a set of related topics centered around the
Landau analysis \cite{Landau:1959fi}\footnote{See also many textbooks on
  QFT.} of singularities of Feynman graphs.  What makes this subject
currently important is not merely the classic application to the locating
of singularities of Feynman amplitudes, but its application by Libby and
Sterman \cite{Libby:1978bx} to determine and analyze regions of low
virtuality for amplitudes and cross sections in various asymptotic
high-momentum limits.  Their analysis shows that, in the loop-momentum
space of Feynman graphs, these regions are determined by the locations of
pinches in the massless limit (without any requirement that the full itself
theory is massless).\footnote{The exact wording of this sentence might
  appear to be somewhat at odds with what Libby and Sterman actually wrote.
  See Ch.\ 5 of Ref.\ \cite{Collins:2011qcdbook} for my attempt to explain
  the logic.}  In addition, they formulate a power-counting analysis to
determine which regions contribute at leading power in a given theory.  To
locate the pinches, they use the Landau criterion applied in a massless
theory, in a form given by Coleman and Norton \cite{Coleman:1965xm}.  That
form is that the pinches correspond to classically allowed processes; this
is rather easy to apply in the massless limit.  Libby and Sterman's
analysis results, among other things, in the well-known classification of
momenta into hard, collinear and soft.  It then underlies all results in
factorization, which is an essential tool in most current QCD
phenomenology.  Moreover, as will be explained in more detail below, a
number of extensions are needed to the Landau results for current and
future work.

The primary outcome of the Landau analysis is a criterion for the existence
of a pinch of the contour of integration in terms of what are called the
Landau equations\footnote{In this paper we will solely be concerned with
  physical-region singularities and pinches (or their equivalent for more
  general integrals).} when the objects of study are standard
momentum-space Feynman graphs. To cover more general situations, I will use
the term ``Landau condition'' (or criterion, depending on the shade of
meaning needed).

However, given that the Landau criterion and the work of Coleman and Norton
are foundational to most work on perturbative QCD (pQCD), it is very
disconcerting that there are notable deficiencies in existing treatments of
the Landau criterion, and that these become particularly noticeable in
massless theories.  I will review some of the problems in the next
paragraphs, and then in more detail in Sec.\ \ref{sec:literature}.  One of
the problems is that the actual proof by Coleman and Norton, of the Landau
criterion for pinches, fails completely in the massless case.  This is not
simply a matter of a subtle issue in a high-order graph, but something that
happens in a one-loop self-energy graph.  It turns out that an implicit and
apparently obvious and uncontroversial assumption is false.  None of the
deficiencies necessarily entail that the Landau criterion is incorrect.
Indeed, a primary result of the present paper is a proof that does work and
that is valid for the massless case, as well as for other cases needed in
work on QCD, among others.

Nevertheless the problems indicate areas where some conceptual
understanding has been missing.  This can seriously impact efforts to use
the methods in other situations.  In addition, loopholes in the original
arguments suggest the possibility of interesting new results.

The aim of Coleman and Norton's derivation was to show that the Landau
condition is both necessary and sufficient for physical-region pinches in
the space of loop momenta.  From it they then derived the well-known result
that the location of a pinch corresponds to a classically allowed
process.\footnote{In a theory with no massless particles, this part of the
  proof is correct.  But as pointed out by Ma \cite{Ma:2019hjq}, the proof
  needs extensions to make it work in the massless case.} They apply a
Feynman parameter representation and then perform the momentum integrals.
Their proof is applied to the integral over Feynman parameters, with its
single denominator.  There is an unstated assumption that there is a pinch
in the original momentum-space integral if and only if there is a pinch in
the parameter integral.  But in the massless case that implication simply
fails, and it fails in the simplest graph, a one-loop self energy.  As
shown in the present paper in App.\ \ref{sec:F.param.massless}, the
parameter integral for this graph has no pinch at all corresponding to the
well-known collinear pinch in momentum space; this is quite unlike the
situation for the normal threshold singularity in a massive theory.

Therefore, the first aim of this paper is to prove the necessity and
sufficiency of the Landau criterion for a pinch directly in momentum space.
The proof given here applies to a whole class of integrals, of which
standard Feynman graphs are only one example. Unfortunately some
restrictions are applied to make the proof work, but these are obeyed both
for standard Feynman graphs and for various other kinds of graph that are
commonly used in QCD.  More work is needed to investigate more general
cases.

It should be noted that there are important shifts of emphasis between the
Landau analysis and the QCD applications.  The Landau analysis was
concerned mainly with where actual singularities of graphs occur as
functions of their external parameters, and was almost entirely confined to
the massive case.  But the QCD applications are fundamentally concerned
with locating regions where an integration contour is trapped by propagator
singularities to be in a region of low virtuality compared with some large
scale $Q^2$.  These regions correspond to manifolds of exact pinches in the
massless theory; the regions in a possibly massive theory are neighborhoods
of the pinch singular surfaces in the massless theory.

Moreover, once one has a trap of the contour, one is also interested in
what contour deformations are allowed and hence which subset of propagator
singularities are involved in trapping the contour and which can be
avoided.

Symptoms of problems in the available treatments are found in the classic
book by Eden, Landshoff, Olive, and Polkinghorne \cite{ELOP} (ELOP).  In
their Sec.\ 2.1, they give a general treatment of singularities of
integrals over an arbitrary number of variables.  They present derivations
of the Landau condition for a singularity, in a form appropriate for a
general integral, and not merely for those integrals that arise from
momentum-space Feynman graphs.  They are very explicit (p.\ 48) that a
``proper proof needs the use of topology'' and that they will ``be content
with plausibility arguments''.  However, they do not give any real
indications of the deficiencies of their arguments.  At the end of the
section, they write: ``A rigorous treatment requires homology theory and
for this we refer to the paper by Fotiadi, Froissart, Lascoux, and Pham
(1964).'' They do not explain the need for homology theory, and the paper
by Fotiadi et al.\ is listed as unpublished in the bibliography, with
a statement to see also a published paper by the same authors
\cite{Fotiadi:1965}.  The published paper does contain relevant material,
but not what is needed.  It is clear that there is second paper by the same
authors that contains the missing proof.  However, as far as I can
determine, only one paper by these authors can now be found, and that is
the one that ELOP listed as published.  The then-unpublished paper that
contains the referred-to proof appears to remain both unpublished and
inaccessible.

As regards the application of homology theory to this problem, there is a
book by Hwa and Teplitz \cite{Hwa.Teplitz.1966} on the subject.  They give
much relevant material, including a reprint of the published Fotiadi et
al.\ paper.  But on p.\ 51, they write that an extension of the treatment
is needed for Feynman graphs of more than one loop, and that ``at present
no such extension has been made''.  (Their book is dated 1966.) A much
later book by Pham \cite{Pham:2011} gives many relevant results, but, as
far as I can see, not the ones that are otherwise missing.  In a sense, one
basic problem with both references is that they try to be too general in
the integrals they work with.  For the results in which we are interested
for Feynman graphs, the denominators are real for real values of their
arguments and there is an $i\epsilon$ prescription, and we are interested in
physical-region pinches. These properties are what enable the proof in the
present paper to work.

Primary new results of the present paper are as follows:
\begin{enumerate}

\item A full proof is given that the Landau condition is both necessary and
  sufficient to determine the locations of physical-region pinch
  configurations in a class of integrals that includes momentum-space
  integrals for Feynman graphs.  The applicability to Feynman graphs
  includes not only standard relativistic Feynman graphs, but also the
  various modified graphs that appear in factorization (notably including
  Wilson lines and the approximated graphs that arise in a treatment of the
  Glauber region,\footnote{See Ref.\ \cite{Collins:2011qcdbook} for
    details, including in its Sec.\ 5.6 an analysis that uses the Landau
    criterion.}  as well as those containing Wilson lines).

  The proof applies directly to the momentum-space integrals for Feynman
  graphs without any need to invoke the Feynman parameter method.

\item The proof is in two parts.  One part is a detailed analysis of the
  conditions for a trapped contour in terms of constraints on the direction
  of contour deformation.  Certain restrictions apply to this part of the
  proof --- see item \ref{item:restrictions}.  The other part of the proof
  is embodied in a purely geometric theorem in arbitrarily high dimension
  on whether or not the constraints can be satisfied.\footnote{Undoubtedly
    the second part of the proof is closely related to the mathematical
    subjects of which an account is given in books by Gallier
    \cite{Gallier:2008,Gallier:2011}.  But I have not yet found the result
    that is needed for the applications treated in this paper.}

  The presentation of the two parts is in the reverse order to the
  description just given.  The geometric part comes first, since its
  results are used in analyzing properties of contour deformations. 

\item 
  \label{item:first.order}
  A simple example is given, in App.\ \ref{sec:2D.first.order}, to
  illustrate a difficulty that has to be overcome in the part of the proof
  analyzing contour deformations.  

  Based on experience in visualizable examples in integrals over one
  complex dimension, it is natural to assume that if a contour deformation
  avoids the singularity due to a zero of a denominator, then there must be
  a (non-zero) positive first-order shift in the imaginary part of the
  denominator, given the usual $i\epsilon$ prescription.  The example in App.\
  \ref{sec:2D.first.order} shows that this supposition is false;
  singularities can be avoided with a contour deformation that gives a zero
  first-order shift; I term this an ``anomalous deformation''.  Two or more
  complex dimensions are needed for an anomalous deformation to exist.

  In the particular example given, the contour is not pinched and one can
  equally avoid the singularities by a non-anomalous deformation.  But to
  make an satisfactory determination of the condition(s) for a pinch, it is
  essential to exclude the possibility of an anomalous deformation that
  avoids the singularities of the integrand when the Landau criterion is
  satisfied. This leads to considerable complications in the proof in this
  paper.

\item
  \label{item:restrictions}
  Overcoming the difficulties just mentioned, leads to a need to impose
  certain non-trivial restrictions on the denominators in the integral, in
  order for the methods of proof used here to succeed.  The restrictions
  are that the denominators are at most quadratic in their arguments and
  that any quadratic terms obey a certain sign constraint --- see the
  statement of Thm.\ \ref{thm:main.contour}.

  Luckily the restrictions apply to the pure momentum-space form of
  standard Feynman graphs, including many of the modified graphs used in
  QCD factorization.  However, it would be obviously be useful to find
  better proofs that would eliminate the restrictions as much as possible.
  The difficulties suggest areas for further investigation that appear not
  to have been properly considered in the original proofs.

\item A simple explanation is given that the Landau condition is necessary
  but not sufficient for a singularity as a function of external parameters
  (contrary to the situation concerning pinches of the contour of
  integration).  A supplementary analysis is needed to determine whether or
  not there is an actual singularity given that the contour is trapped.

\item The proofs apply to more general situations than standard Feynman
  graphs.  The range of applicability includes the modified and
  approximated Feynman graphs ubiquitous in QCD factorization.  Others
  include systematic treatments of properties of Feynman graphs in
  coordinate space. Some illustrations are provided.

  In coordinate space, only some restrictions on contour deformations arise
  from singularities of the integrand.  Other restrictions arise when one
  has rapidly oscillating factors like $e^{ik\cdot x}$.  These give strong
  cancellations in the integral, and to get a good analysis it is needed,
  if possible, to deform the contour in a direction that gives an
  exponential suppression.  The geometrical part of this paper's proofs
  applies directly to such cases to determine where such a deformation is
  not possible.

\end{enumerate}
Many or all of the issues are elementary or even trivial in integrals over
one complex variable.  In that situation, issues about contour deformation
are readily visualizable. But this is no longer the case in higher
dimensions, as illustrated by the simple example in App.\
\ref{sec:2D.first.order}.

Some areas for future extension and use of the methods and results in this
paper are:
\begin{enumerate}

\item For a number of purposes, it would be useful to have a systematic and
  general determination for a given process of which regions of space-time
  for vertices dominate.  For example, Brodsky et al.\
  \cite{Brodsky:2019jla} have given an argument that the momentum sum rule
  is violated in deep inelastic scattering, in contradiction with standard
  results from the operator product expansion (OPE) and factorization.
  Their argument depends on properties of the regions of space-time
  involved.  To assess their work completely, it is necessary to have
  systematic and fully deductive derivations of the space-time regions
  involved in processes such as those to which the OPE and standard
  factorization are applied.

  Here we see examples of the situations mentioned above where one needs to
  determine where there is a lack of suppression in an integral containing
  multiple oscillating exponentials of the form $e^{ik\cdot x}$, and to be able
  to do this systematically to all orders of perturbation theory (at
  least).

\item Given that there is a pinch at some point in an integration for a
  Feynman graph, it is common that the pinch is restricted to a subset of
  the propagators.  Then the contour of integration can be deformed to take
  the other propagators off shell, while not crossing the poles for the
  pinched propagators.  How in general is one to characterize the allowed
  directions for such deformations and determine unambiguously which
  propagators are trapped and which not?
  
\item The methods in this paper could be very useful in calculations of
  hard scattering coefficients and other quantities, as is needed for much
  Standard Model phenomenology.  Loop integrals are encountered that are
  not readily amenable to analytic calculations, so that numerical
  calculations are needed.  In the numerical implementation of integrals,
  it is very desirable to have \emph{algorithmic} methods to deform
  integration contours away from non-pinch singularities of the integrand.
  In addition, where there is a pinch, it is important to deform the contour
  away from singularities that do not participate in the pinch.

  Some important work in this area is by Gong, Nagy and Soper
  \cite{Gong:2008ww}, and by Becker and Weinzierl
  \cite{Becker:2012nk,Becker:2012bi}.  The geometric methods obtained in
  the present paper should be able to contribute to more general methods.

\end{enumerate}

Although the results in this paper are in principle purely mathematical,
the motivations and the situations considered arise from certain kinds of
physics problem.  Thus the presentation, the examples, and the terminology
are strongly influenced by the physics applications.

A guide to the statements of the main results is as follows:
\begin{itemize}
\item The statement of the main result on pinches is Thm.\
  \ref{thm:main.contour}.

\item It applies to an integral of the form (\ref{eq:integral}).

\item It uses Definition \ref{def:Landau.point} for a Landau point.

\item The proof of the theorem uses a corresponding geometrical result,
  Thm.\ \ref{thm:main.geom}, and the relation to the notation for the integral
  is specified in Eq.\ (\ref{eq:denom.series}).

\end{itemize}

\section{Pinches of singularities of integrand in momentum-space
  Feynman graphs, etc}
\label{sec:pinch}

In this section, I present the classic problem of determining where the
contour of integration is trapped in the kind of situation exemplified by
momentum-space Feynman graphs.  The integrals are restricted to those such
as occur in momentum-space Feynman graphs in the physical
region\footnote{For the purposes of this paper, saying that a graph is in
  the physical region means that the external momenta are real and that
  before contour deformation the integral is calculated with internal loop
  momenta all real \cite{Coleman:1965xm}.  There is no requirement that the
  external momenta be on-shell.}.  These restrictions are: (a) the external
parameters are real; (b) the integration variables before contour
deformation are real; (c) the denominators giving the singularities of the
integrand are real for real values of their arguments; (d) the integral is
defined by an $i\epsilon$ prescription.

Although much of the material is basically standard, the presentation here
is needed to emphasize particular issues that are important in the sequel,
and to define the notation to be used.

Motivation can be made, both for the general problem and for the
geometrical formulation, from an elementary example.  To this end, App.\
\ref{sec:self-energy} gives the well-known example of a one-loop
self-energy.

The general case is in an arbitrarily high dimension with arbitrarily many
denominators whose zeros give singularities of the integrand.  Considerable
subtleties occur, as we will see.  Hence to provide fully water-tight
derivations, it is important to have precise operational definitions of the
relevant concepts about a given contour deformation, about its
compatibility or not with the integrand's singularities, and about its
avoidance or non-avoidance of the singularities.  It is important that the
definitions can be applied mechanically and essentially computationally,
without the need for creativity or special insights.

The work in this section will motivate the geometric theorem to be proved
in Secs.\ \ref{sec:overall}--\ref{sec:landau.proof}.  Only after that will
be able to find a full proof that a necessary and sufficient condition for
a pinch of the integration contour is that a particular Landau condition is
obeyed.

This is the canonical application of the more abstract geometrical theorem,
and it will influence the terminology used.  Some further applications are
summarized in Sec.\ \ref{sec:extra.cases}.

\subsection{Formulation of problem}
\label{sec:formulation}

\subsubsection{Momentum-space Feynman graphs}
The value of a momentum-space Feynman graph has the form
\begin{equation}
\label{eq:F.graph}
  I(p,m) = \lim_{\epsilon \to 0+} \int_{\Gamma_0} \diff[d]{k}
              \frac{ X(k;p,m) }
                   { \prod_{j=1}^N \left[ f_j(k;p,m) + i\epsilon \right]^{n_j} }.
\end{equation}
Here $p$ is the multi-dimensional array of variables for the external
momenta, and $m$ is the array of masses of the theory.\footnote{No
  restriction is placed on whether the masses are zero or non-zero.}  The
integration variable $k$ is the array of all loop momenta, and has
dimension $d$, which may be arbitrarily high.  We call each $f_j+i\epsilon$ a
denominator factor, and we call $X$ the numerator factor.  Each is a
function of the integration variable $k$ and the external parameters $p$
and $m$.

We restrict from now to situations where:
\begin{itemize}
\item The denominator factors $f_j$ are real-valued when their arguments
  are real.
\item The values of the external parameters $p$ and $m$ are real, and the
  initial contour of integration, denoted by $\Gamma_0$, gives an integral over
  all real values of $k$.  This we will call the restriction to the
  physical region, and it implies a restriction solely to physical-region
  pinches and singularities.
\item There is an $i\epsilon$ with each denominator, and the end result is for the
  boundary value as $\epsilon$ goes to zero from positive values.
\item All the functions $f_j$ and $X$ are analytic functions of their
  arguments.  In particular, $X$ has no singularities for any finite value
  of $X$.  Then all singularities of the integrand are due to zeros in one
  or more of the denominator factors.
\end{itemize}
These properties evidently apply to a much wider class of integrands than
those for standard relativistic Feynman graphs, but they are motivated by
that situation.  The methods we use could be easily applied to certain more
general classes of integrand, but we will not do so.  In the standard case,
each $f_j$ is a quadratic function of its arguments and the numerator $X$
is polynomial in momenta and masses.  But other possibilities can and do
arise.  Notably, denominators from straight Wilson lines (or eikonal lines)
have linear dependence on momentum instead of quadratic.

Since the numerator factor $X$ is non-singular as a function of $k$, it
does not affect the determination of where pinch singularities occur.  In
contrast, the numerator factor often affects power-counting analyses
\cite{Libby:1978bx} for quantifying the size of the contribution associated
with a pinch, but that is not the concern of the present paper.  

The exponent $n_j$ of a denominator factor is typically unity; however, one
regularly meets other cases.  Our concern is with singularities of the
integrand caused by zeros of one or more $f_j$, and with whether the
singularities obstruct contour deformations.  The most general case is that
each $n_j$ is not zero or a negative integer, and this is what we will
assume henceforth.  (In the remaining cases, where an $n_j$ is a negative
integer or zero, a zero in $f_j$ does not cause a singularity of the
integrand, and then the factor $1/(f_j+i\epsilon)^{n_j}$ can be incorporated in
the numerator factor $X$.)

Furthermore, it is possible that when a particular $f_j$ is zero, the
numerator factor is also zero in such a way as to remove the singularity of
the integrand due to a factor $1/f_j^{n_j}$.  In such cases the denominator
factor can be removed and compensated by a corresponding change in the
numerator factor.  So we remove such cases from consideration.

For most values of its arguments, $I(p,m)$ is an analytic function; this is
shown by differentiating the integrand with respect to $p$ and $m$.
However, that argument fails if one or more $f_j$ is zero somewhere on the
initial contour $\Gamma_0$.  But if the contour can be deformed away from the
singularity/ies of the integrand, then the differentiation argument can be
applied on the deformed contour, and gives analyticity of $I(p,m)$ at the
values of $p$ and $m$ under consideration.

Hence our primary aim is to determine situations where such a deformation
away from a singularity of the integrand is not possible.  We term such a
situation a pinch (by analogical generalization from corresponding
situations in one-dimensional contour integrals).

A contour of integration is a surface that in terms of real variables has
dimension $d$.  It is embedded in a space of $d$ complex dimensions, i.e.,
of $2d$ real dimensions.  

The only singularities of the integrand are at zeros of one or more
denominators.  When $\epsilon$ is nonzero and when all of $k$, $p$, and $m$ are
real, the integrand is non-singular on all of the initial contour $\Gamma_0$,
because the imaginary part of each denominator is non-zero.  When $\epsilon\to0+$,
i.e., $\epsilon$ approaches zero from positive values, singularities may appear on
$\Gamma_0$ at positions where one or more $f_j$ is zero.  In analyzing a
candidate deformation to a contour $\Gamma$, we wish first to know whether or
not a singularity is encountered for positive $\epsilon$ during the deformation
and before $\epsilon$ is finally taken to zero.  Such a deformation is disallowed.
Finally, for an allowed deformation we wish to know whether the deformed
contour avoids a particular given singularity after $\epsilon\to0+$

\subsubsection{Feynman parameters}

One technique for evaluating Feynman graphs is to use Feynman parameters.
This converts the original formula (\ref{eq:F.graph}) to an integral with a
single denominator factor.  It has more integration variables, but the
integrals over momenta can be performed analytically.

After Feynman parameterization, but before the integral over momenta, the
integral becomes
\begin{multline}
\label{eq:F.graph.param.mom}
  I(p,m) = \lim_{\epsilon \to 0+} \int \diff[d]{k} \prod_j \left( \int_0^1 \diff{\alpha_j} \right)
           ~ \delta\biggl( \sum_j \alpha_j - 1 \biggr)
\\   
    \frac{ \Gamma\xleft( \sum_jn_j \right) }{ \prod_j \Gamma(n_j) }
           \frac{ X(k;p,m) \prod_j \alpha_j^{n_j-1} }
                { \left( \sum_j \alpha_j f_j(k;p,m) + i\epsilon \right)^{\sum_jn_j} }.
\end{multline}
The single denominator has an exponent that is the sum of the exponents in
the original problem.  The extra normalization factor with the
Gamma-functions can be absorbed into a redefinition of the numerator
factor. The same applies to the factor of powers of $\alpha_j$, which can be
singular only at endpoints of the integration, and then only if some $n_j$s
are not positive integers.

After exchanging the order of the integrations and performing the momentum
integrals, one gets an integral of the form
\begin{multline}
\label{eq:F.graph.param.only}
  I(p,m) = \lim_{\epsilon \to 0+} \prod_j \left( \int_0^1 \diff{\alpha_j} \right)
           ~ \delta\biggl( \sum_j \alpha_j - 1 \biggr)
\\
           \frac{ C(\alpha;p,m) }
                { \bigl( D(\alpha;p,m) + i\epsilon \bigr)^{\sum_jn_j} },
\end{multline}
where $\alpha$ without a subscript denotes the array $j\mapsto\alpha_j$.  The rules for
obtaining the functions $C$ and $D$ can be found in textbooks, e.g.,
\cite{ELOP}.

\subsubsection{General situation}

We can treat all of these integrals as special cases of the following form:
\begin{equation}
\label{eq:integral}
  I(z) = \lim_{\epsilon \to 0+} \int_{\Gamma_0} \diff[d]{w}
              \frac{ B(w;z) }
                   { \prod_{j=1}^N \bigl( A_j(w;z) + i\epsilon \bigr)^{n_j} },
\end{equation}
with each $n_j$ not equal to zero or a negative integer.  This is the most
general form we will consider for our work.  It is just like
(\ref{eq:F.graph}) except that we allow the contour of integration to have
boundaries, and the values of $d$, $N$, and $n_j$ may not be the same as
before.  To indicate the more general situation, the notation has been
changed: all external parameters are folded into a single multidimensional
variable $z$, and the symbols are changed for the integration variables,
the denominators, and the numerator.

The numerator $B$ and the denominators $A_j$ are analytic for all values of
their arguments, and we restrict attention to the case that every $A_j$ is
real when its arguments are real.\footnote{This restriction and the $i\epsilon$
  prescription do not appear in the mathematical work in Refs.\
  \cite{Fotiadi:1965,Pham:2011}.  The restrictions lead here to more
  powerful results of physical relevance \cite{Coleman:1965xm,Libby:1978bx}
  for Feynman graphs in the physical region. }  The numerator factor $B$
need not be real when its arguments are real, and in fact $B$ will play no
role in our work.

\subsubsection{Singularities of integral}
Landau's original problem was to determine where the integral $I(p,m)$ is
singular as a function of $p$ and/or $m$.  Now we change to the more
general notation of Eq.\ (\ref{eq:integral}).  By definition, $I(z)$ is
analytic at some point when it is complex-differentiable in a neighborhood
of the point.  If it is analytic, then all derivatives exist, by standard
theorems.  Hence a function is singular at some point if and only if one or
more derivatives fails to exist at that point, or arbitrarily close to it.

As already observed, we can apply derivatives with respect to $z$ inside
the integration.  Then a singularity can only arise in the dependence on
the external variable $z$ if the initial contour of integration is pinched
somewhere by a singularity caused by a zero of one or more $A_j$s.  The
integral might actually diverge if integrand is singular enough, or an
actual divergence might only occur in a derivative or multiple derivative.

So in a general integral, like (\ref{eq:integral}), we can only get a
singularity as a function of external parameters if one of the following
occurs (see \cite{ELOP} and other references):
\begin{enumerate}
\item The integration contour is trapped, i.e., pinched, by singularities
  of the integrand.  That is, there is no contour deformation that avoids
  these singularities.
\item Singularities of the integrand occur on a boundary of the
  integration, and cannot be avoided by a deformation that preserves the
  boundary of the contour.\footnote{The condition of preserving the
    boundary of the contour is actually not quite what we need.  In a
    multidimensional case, Cauchy's theorem can allow deformations that
    preserve the value of an integral while moving the boundary.  Since our
    concern in this paper is non-boundary singularities, we will not
    consider the ramifications of this remark.  It can matter when a
    momentum-space integral is decomposed into sectors and contour
    deformations determined separately on each sector, as in applications
    of the numerical methods of Ref.\ \cite{Gong:2008ww}.}  This case does
  not occur for pure momentum-space Feynman integrals of a standard kind,
  but can occur in more general situations (including Feynman graphs with
  the use of Feynman parameters).
\end{enumerate}

In a general integral of the form of (\ref{eq:integral}), it is also
possible that for particular values of $z$, the integral acquires
divergences from where some integration variables go to infinity; this
requires that the integration range is infinite in those variables.  This
situation does not arise for Feynman graphs in a pure momentum-space
formulation, since differentiation with respect to $p$ or $m$ in
(\ref{eq:F.graph}) always improves ultra-violet convergence.  But it can
occur when $I(p,m)$ is expressed as an integral over space-time coordinates
of vertices, and divergences occur for large positions.  Such cases do in
fact appear to be able to be treated by elementary generalizations of the
methods considered here, but we leave that for further work.

Notice that the above argument says that the existence of a singularity of
$I(z)$ at a particular value of $z$ implies that the contour of integration
is trapped at a singularity of the integrand considered as a function of
the integration variable $w$.  The converse is definitely not always valid,
as shown in App.\ \ref{sec:trap.no.sing} with the aid of a trivial
counterexample.  Given that a pinch has been found, a separate calculation
of the contribution to the integral from a neighborhood of the pinch point
to determine the existence or non-existence of an actual singularity.

But, as already observed in the introduction, what is important to many
modern applications is not the actual existence of a singularity of $I(z)$
as a function of the external parameter(s) $z$, but whether or not the
integration is trapped, and where.

\subsection{Deformations of contour}
\label{sec:deformation}

In the integral (\ref{eq:integral}), the contour deformations to be
considered are replacements of the real values of the integration
variable $w$ by
\begin{equation}
\label{eq:deform}
  w = w_R + i \lambda v(w_R).
\end{equation}
Here $w_R$ is a real variable that ranges over all values on the original
real contour $\Gamma_0$.  The real variable $\lambda$ is in the range 0 to 1, and
parameterizes the amount of deformation.  For each $\lambda$ we have a particular
contour $\Gamma_\lambda$, parameterized by $w_R$.  The original contour is at $\lambda=0$,
and the final deformed contour is at $\lambda=1$.  The function $w_R \mapsto v(w_R)$ is
from real values $w_R$ to real $d$-dimensional values, so that $i\lambda v(w_R)$
gives the imaginary part of $w$ on the deformed contour.  The function $v$
must be continuous and piecewise differentiable (but only in the sense of
differentiation with respect to \emph{real} variables, not necessarily with
respect to complex variables).  We take it to be zero on any boundary of
$\Gamma_0$, so that the boundaries of the contour are unchanged.

The value of the integral on the deformed contour is
\begin{multline}
\label{eq:integral.deformed}
  \hspace*{-3mm}I(z,\lambda,\epsilon) 
\\
 \eqdef  \int_{\text{real},\Gamma_0} \diff[d]{w_R}  J(w_R,\lambda)
              \frac{ B(w;z) }
                   { \prod_{j=1}^N \bigl( A_j(w;z) + i\epsilon \bigr)^{n_j} }.
\end{multline}
Here, the integral symbol is still equipped with the symbol for the
undeformed contour $\Gamma_0$, but now it is the real variable of integration
$w_R$ that is on $\Gamma_0$, and not the argument $w$ of the integrand.  We have
chosen not to take the limit $\epsilon\to0+$ yet.  Here $w$ is given by
(\ref{eq:deform}), and $J$ is the Jacobian of the transformation from $w_R$
to $w$.  According to Cauchy's theorem\footnote{An accessible and
  elementary proof of the Cauchy theorem beyond the one-dimensional case is
  given by Soper \cite{Soper:1999xk}.  Note that it is possible to
  generalize the theorem to certain cases where the boundary changes, but
  we will not deal with that issue here.}, the value of the integral is
independent of $\lambda$ provided that the integrand has no singularities on the
contour when $0\leq\lambda\leq1$, and hence that every $A_j(z,w_R+i\lambda v(w_R))+i\epsilon$ is
non-zero for all $w_R$ on $\Gamma_0$, for $0\leq\lambda\leq1$.  Thus $I(z,1,\epsilon)=I(z,0,\epsilon)$.

The target value of the integral is the limit as $\epsilon$ decreases to zero of
the integral on the undeformed contour, i.e., of $I(z,0,0+)$.  On the
undeformed contour, $\lambda=0$, there are typically singularities of the
integrand for some values of $w_R$; these are where there are zeros of one
or more denominators.  To get the same value for the integral on the
deformed contour, we must apply the condition of all $A_j$s being nonzero
for all $\lambda$ in the range $0\leq\lambda\leq1$ and for all $\epsilon$ in a range $0<\epsilon\leq\epsilon_0$, for
some positive number $\epsilon_0$. Notice that the range of $\epsilon$ considered in the
condition \emph{excludes} zero.  Then taking the limit $\epsilon\to0+$ gives
$I(z,1,0+)=I(z,0,0+)$.

Even if the contour does not avoid all singularities, it is useful to
deform a contour to avoid as many singularities as is possible.  For such a
deformation to be useful, we require that no singularities of the integrand
appear on the contour until the very last step of taking $\epsilon$ to zero with
the contour fully deformed. We call such a deformation allowed, while any
deformation that encounters a singularity before that step is not an
allowed deformation.

Given some of the complications that arise in a complete analysis, it is
useful to have very precise operational specifications of what is meant by
an allowed deformation, and of what is meant by a singularity-avoiding
deformation.  In analyzing problems concerning possible deformations, it is
also useful to define such concepts relative to particular subsets of
denominators. In particular, we may be interested not only in whether a
contour deformation avoids all singularities of the integrand, but also in
whether it avoids singularities at particular values of $w_R$, and possibly
only for those singularities caused by zeros in particular subsets of the
$A_j$.  One physical motivation is that in many applications in QCD, sets
of singular propagators correspond to factors in a factorization theorem,
each of which can be considered separately; given a kinematic configuration
of momenta, only in a part of a graph that corresponds to a hard scattering
factor can we deform away from propagator singularities.

In all of the following definitions, we will assume a particular value of
the external parameter $z$.

\begin{definition}
  \label{def:compat.deform.point}
  Given a particular subset $\mathcal{A}$ of the $A_j$s, we define that at
  a particular value $z$ and $w_S$ for the external parameter and
  integration variable, a deformation is defined to be \emph{compatible}
  with the denominators $\mathcal{A}$ if there exists a positive non-zero
  $\epsilon_0$ such that $A_j(w_R+i\lambda v(w_R);z)+i\epsilon$ is non-zero when $w_R$ is in a
  neighborhood of $w_S$ for all $A_j$ in $\mathcal{A}$, and for $0\leq\lambda\leq1$
  and $0<\epsilon\leq\epsilon_0$.
\end{definition}
Note that the case $\epsilon=0$ is specifically \emph{not} included. The lack of
singularities in the given range indicates that the zeros of the specified
denominators do not give singularities that obstruct the use of Cauchy's
theorem. We specify that there is a lack of zeros not only at $w_R$ but in
a neighborhood.  The reason is that if there is a zero of $A_j(w)$ at
$w=w_S$, then as $\lambda$ is increased from zero, the position of the zero
usually migrates to nearby values of $w_R$; such a zero is equally
effective at obstructing a contour deformation.

\begin{definition}
  \label{def:allowed.deform}
  A \emph{locally allowed deformation} at $z$ and $w_S$ is one that is
  compatible with all the denominators at $w_S$. An \emph{allowed
    deformation} is one that is an allowed deformation at all $w_S$ in
  the range of integration.
\end{definition}
As already observed, for an allowed deformation there is no obstruction to
the contour deformation, so that $I(z,0,\epsilon)=I(z,1,\epsilon)$ when $0<\epsilon\leq\epsilon_0$, and
hence $I(z,0,0+)=I(z,1,0+)$.  (It is sufficient to take a common nonzero
maximum $\epsilon_0$ for all denominators.)

It should be noted that a trivial special case of an allowed deformation is
when there is no change in the contour at all, i.e., when $v(w_R)$ is zero
for all $w_R$.

We next have the definition of a deformation that avoids singularities:
\begin{definition}
  \label{def:sing.avoid,j.kS}
  Given a particular subset $\mathcal{A}$ of the $A_j$s, we define that at
  $z$ and $w_S$ a deformation \emph{avoids an $\mathcal{A}$-associated
    singularity} if the deformation is allowed at $z$ and $w_S$, and if
  the non-zero condition on $A_j+i\epsilon$ also applies for the given
  $A_j\in\mathcal{A}$ at $\epsilon=0$ and $0<\lambda\leq1$, for $w_R$ in some neighborhood of
  $w_S$.  By bringing in the definition of an allowed deformation, the
  condition for a singularity-avoiding deformation is that the only place
  where $A_j+i\epsilon$ is zero (with $A_j\in\mathcal{A}$) is where $\epsilon$ and $\lambda$ are
  both zero.  (Here it is taken for granted that the relevant ranges for
  $\epsilon$ and $\lambda$ are for non-negative values below $\epsilon_0$ and 1, respectively.)
\end{definition}
Of course, no requirement is placed when $\lambda=\epsilon=0$, because the interesting
case is when one or more $A_j(w_S)$ is zero, and then we ask whether a
particular deformation avoids the resulting singularity in the integrand.

We can define more global kinds of singularity avoidance:
\begin{definition}
  \label{def:sing.avoid,kS}
  A deformation \emph{avoids any singularity} at $w_S$ if it is allowed
  and avoids $A_j$-related singularities for all $A_j$.
\end{definition}
\begin{definition}
  \label{def:sing.avoid}
  We define that at $z$ a \emph{deformation (globally) avoids any
    singularity} if it is allowed and avoids $A_j$-related singularities
  for \emph{all} $A_j$ and for \emph{all} $w_S$ on $\Gamma_0$.
\end{definition}

\begin{figure}
  \centering
  \begin{tabular}{c@{\hspace*{1cm}}c@{\hspace*{1cm}}c}
      \includegraphics[scale=0.7]{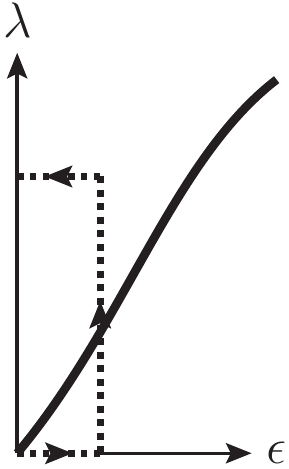}
   &
      \includegraphics[scale=0.7]{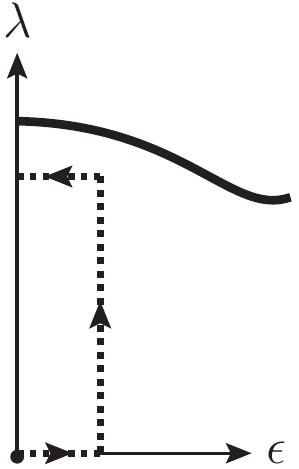}
   &
      \includegraphics[scale=0.7]{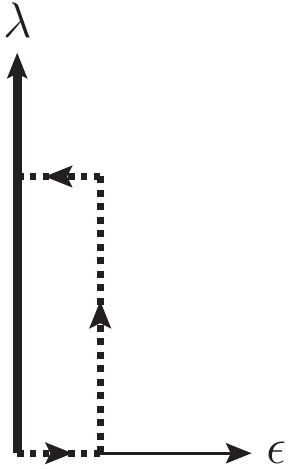}
  \\
     (a) & (b) & (c)
  \end{tabular}
  \caption{(a) To illustrate where in the $(\epsilon,\lambda)$ plane singularities can
    be encountered that prevent a deformation from being allowed.  The
    thick solid line indicates where the contour deformation first hits a
    singularity as $\lambda$ is increased from zero.  The dashed line is a path
    in $(\epsilon,\lambda)$ to get from the initial situation with $\lambda=\epsilon=0$ to the
    deformed contour $(\epsilon,\lambda)=(0,1)$.  The thick solid line extends to the
    origin.  (b) The case where no singularity is encountered except at
    $\epsilon=\lambda=0$, i.e., that deformation avoids the singularity. (c) A situation
    with an allowed deformation that does not avoid a singularity: The
    thick solid line on the vertical axis indicates that the integrand is
    singular when $\epsilon=0$ for all values of $\lambda$.}
  \label{fig:eps.lam.sing}
\end{figure}

An illustration of how these definitions are used is given in Fig.\
\ref{fig:eps.lam.sing}.  In the left-hand diagram of the $\lambda$-$\epsilon$ plane is
illustrated the situation for a disallowed contour deformation. The dashed
line indicates the sequence of values we wish to use to get from the target
value of the integral, i.e., $I(z,0,0+)$, to the value on the deformed
contour. The solid line is where the contour deformation first hits a zero
of $A_j+i\epsilon$ somewhere on the contour of integration as $\lambda$ is increased
from zero.  On the right-hand horizontal axis, where $\lambda$ is zero and $\epsilon$ is
positive, there is no singularity.  As $\lambda$ is increased eventually a zero
of an $A_j$ hits the contour, and that is indicated by where the dashed
line intersects the solid line.  For larger $\lambda$ Cauchy's theorem
fails.\footnote{It may be that for yet larger $\lambda$ there remains a zero of
  an $A_j$ on the contour. That is irrelevant to our considerations.}  The
solid line goes all the way to the origin; otherwise, simply by restricting
$\lambda$ to a smaller range, the zero(s) of $A_j+i\epsilon$ are avoided, and we can
convert the deformation to the standard form by rescaling the deformation.

In contrast, for a singularity-avoiding deformation, any line or region of
zeros of $A_j+i\epsilon$ does not come all the way to the origin, as in the middle
diagram.  There may be a singularity when both $\lambda$ and $\epsilon$ are zero; that
is the case of interest, i.e., of a singularity on the undeformed contour.
It is possible that there are singularities when large enough deformations
are considered, shown in the middle diagram above $\lambda=1$.  Thus $v(w_R)$ has
been scaled down enough to avoid encountering the singularity/ies.

If the deformation is allowed but doesn't avoid singularities, then we
have a line of singularity-encounters on the vertical
axis, i.e., where $\epsilon=0$ and $\lambda>0$, and this line goes all the way to
$\lambda=0$, as in the right-hand diagram.  

Note that in the diagrams, we are concerned with singularities of the
integrand anywhere on the integration over $w_R$.  The solid lines
correspond to the existence of an singularity somewhere in the integration
range. A singularity that is at some point $w_R=w_S$ when $\epsilon=\lambda=0$ often
migrates to other values of $w_R$ as $\lambda$ is increased.

\subsection{Conversion to a geometrical problem}
\label{sec:pinch.to.geom}

Based on experience with simple examples, it is natural to suppose that one
can determine whether or not a singularity due to a zero of $A_j(w_R)$ is
avoided or collided with, by examining the sign of the (imaginary)
first-order shift of $A_j$ in $\lambda$.  Suppose there is a zero of $A_j$ at
$w=w_S$.  Then a Taylor expansion in powers of $\lambda$ gives
\begin{equation}
  A_j(w_S+i\lambda v(w_S)) +i\epsilon = i\lambda v(w_S) \cdot \partial A_j(w_S) + O(\lambda^2) +i\epsilon,
\end{equation}
where $\partial_\mu A_j(w_R)= \partial A_j(w_R)/\partial w_R^\mu$. Then one would normally expect
that the singularity is avoided if and only if $v(w_S) \cdot \partial A_j(w_S)$ is
strictly positive, which is a geometric condition on the deformation vector
$v(w_S)$ at the zero of $A_j$.  In contrast $v(w_S) \cdot \partial A_j(w_S)$ would be
zero if the deformation is allowed but doesn't avoid the singularity, while
if $v(w_S) \cdot \partial A_j(w_S)$ were negative, then the deformation would not be
allowed.  If these statements were all exactly correct, then applying the
positivity condition on $v(w_S) \cdot \partial A_j(w_S)$ to all $A_j$ for which
$A_j(w_S)=0$ would give the condition that the deformation avoids any
singularity at $w_S$.  

As we will see, the Landau condition gives a necessary and sufficient
criterion that these positivity conditions are incompatible and hence that
the contour is trapped at $w_S$.

However, it is possible to arrange a contour deformation that avoids a
singularity by use of second-order or higher-order terms in $\lambda$, as shown
in App.\ \ref{sec:2D.first.order}.  Now our interest is in the exact
conditions under which contours are trapped or not trapped, i.e., we need a
condition for a trap that is both necessary and sufficient. Therefore the
result in App.\ \ref{sec:2D.first.order} suggests that there could be an
interesting loophole in the Landau analysis.

To exclude the loophole, we need a more detailed analysis, which will be
given in Sec.\ \ref{sec:deal.with.zero.first.order}.  The precise
definitions given above will assist that analysis.  In addition, the
following observations concerning the $\lambda$ dependence of $A_j$ will also be
useful.  They are used in treating the zeros of $A_j$ as a function of $\lambda$
at a fixed value of the real part $w_R$ of the integration variable.

Define $f_{j,w_S}(\lambda)=A_j(w_S+i\lambda v(w_S))$.  Since $A_j(w)$ is analytic as a
function of $w$, $f_{j,w_S}(\lambda)$ is analytic as a function of the
one-dimensional variable $\lambda$ with $w_S$ fixed.\footnote{Note that it is not
  necessary that the function $v(w_R)$ specifying the contour deformation
  be analytic as a function of $w_R$.  Hence $f_{j,w_R}(\lambda)$ is not
  necessarily analytic as a function of $w_R$.}  Suppose that $A_j(w_S)=0$.
Then $f_{j,w_S}(0)=0$.  By standard properties of analytic functions,
either this is an isolated zero of $f_{j,w_S}(\lambda)$ or $f_{j,w_S}(\lambda)=0$ for
all $\lambda$. In the first case, $f_{j,w_S}(\lambda)$ is non-zero for all sufficiently
small non-zero $\lambda$.  In the second case, we haven't avoided the singularity
by the contour deformation under consideration.

\subsection{Primary theorems}
\label{sec:primary.theorems}

In order to provide context for later sections, I state here the main
theorems to be proved.

First come a couple of convenient terminological definitions, of a Landau
point and a Landau condition:
\begin{definition}
  \label{def:Landau.point}
  Let $(D_1,\dots,D_N)$ be a list of dual vectors on a real vector space
  $V$.  We define a \emph{Landau point} for $(D_1,\dots,D_N)$ to be a list
  of real numbers $\lambda_j$ ($1\leq j \leq N$) such that
  \begin{itemize}
  \item All the $\lambda_j$ are non-negative, and at least one is strictly
    positive,
  \item $\sum_{j=1}^N \lambda_j D_j = 0$.
  \end{itemize}
\end{definition}
\begin{definition}
  \label{def:Landau.condition}
  We define the \emph{Landau condition} for $(D_1,\dots,D_N)$ to be obeyed
  if and only if there exists a Landau point for them.
\end{definition}
\noindent
Recall that a dual vector on a vector space $V$ is a linear map from $V$ to
the scalars --- e.g., real numbers --- of a vector space.  A standard
example is the derivative of a function on $V$.  Thus in the preceding
subsection, the derivative of a function $A_j(w_R)$ is $\partial A_j$.  It can be
considered a dual vector $D_j$ by the mapping of vectors to scalars that is
given by $D_j(v)= v\cdot\partial A_J = \sum_\mu v^\mu \partial A_j(w_R)/\partial w_R^\mu$.

Observe that in the case that there is only a single denominator $A$, i.e.,
$N=1$, a Landau point is one where $A=0$ and $D=0$.

The main theorem to be proved concerning pinches is:
\begin{theorem}
   \label{thm:main.contour}
   Given an integral of the form (\ref{eq:integral}), but subject to the
   extra restrictions stated below, consider a real (vector) valued point
   $w_S$ where a nonempty set of denominators is zero.  Then the
   integration is trapped at $w_S$ if and only if a Landau point exists for
   the first derivatives of those denominators that are zero.

   The extra restrictions are that (a) the denominators $A_j(w)$ are at
   most quadratic in $w$, and (b) the signs of the nonzero quadratic terms
   obey a condition which is stated below in (\ref{eq:restrict}) and the
   following paragraphs. This condition is obeyed for the denominators
   encountered in Feynman graphs, including cases with Wilson lines.  It
   also applies to the modified Feynman graphs obtained by applying the
   typical approximations used in deriving factorization.
\end{theorem}
\noindent
The theorem may well be true without these extra restrictions or with
weaker restrictions.  But the proof that we will give only applies when the
restrictions are valid.  Either better methods or much more work would be
needed to give a proof of a less restricted theorem.

However, we do in all cases impose the reality conditions etc that were
listed below Eq.\ (\ref{eq:integral}).

Our proof of the main theorem is made by combining two subsidiary theorems.

The first subsidiary theorem relates the trapping or non-trapping of a
contour to the positivity of the first-order (imaginary) shift in
denominators:
\begin{theorem}
  \label{thm:pinch.to.1st.order.shift}
  With the same hypotheses as in Thm.\ \ref{thm:main.contour}, the
  integration is not trapped at $w_S$ if and only if there is a direction $v$
  such that $v\cdot\partial A_j(w_S)$ is strictly positive for every $A_j$ which is
  zero at $w_S$.  
\end{theorem}
\noindent
Notice that the theorem does \emph{not} say that a contour deformation
given by a function $v \mapsto v(w_R)$ avoids a singularity at $w_S$ if and only
if $v(w_S)\cdot\partial A_j(w_S)$ is strictly positive for every $A_j$ which is zero
at $w_S$.  That property appears to be universally assumed in textbook
proofs, but it is fact false, as shown by the example in App.\
\ref{sec:2D.first.order}.  That is, it is possible to avoid singularities
with a anomalous contour deformation, i.e., one for which $v\cdot\partial A_j(w_S)$ is
zero instead of positive for one or more of the relevant denominators.
Hence some trouble is needed to prove a correct theorem, as we will do
later.  What the theorem does enable one to say is that if there exists an
anomalous deformation there is also a non-anomalous deformation that avoids
the singularity.

The second subsidiary theorem is a purely geometrical result:
\begin{theorem}
  \label{thm:main.geom}
  Let $V$ be a real vector space, and let $(D_1,\dots,D_N)$ be dual vectors
  on $V$.  Then, there is a direction $v$ for which $D_j(v)>0$ for all $j$,
  if and only if there is no Landau point for the $D_j$.
\end{theorem}
\noindent 
When $v$ is the deformation direction of a contour, this theorem gives the
condition under which the first order imaginary-direction shifts in
denominators can be made all positive.

It will be convenient to prove these theorems in the opposite order to
which they are stated, since the proof of the contour deformation theorem
uses the geometrical theorem.  But the motivation and relevance of the
geometric theorem arises from considerations about contour deformations, so
it was convenient to state the contour deformation theorems first.

\subsection{Elementary parts of proofs}
\label{sec:elem.parts}

Certain directions of implication in Thms.\
\ref{thm:pinch.to.1st.order.shift} and \ref{thm:main.geom} are almost trivial to
prove, as follows:

Suppose at some point $w_S$ in the integration range, a set of denominators
is zero and that a contour deformation gives positive first order shifts in
these denominators.  Without loss of generality, list the denominators that
are zero as $(A_1,\ldots,A_N)$.  Let the derivatives be $D_j=\partial A_j(w_S)$.  The
positivity of first-order shifts means that $D_j(v(w_S))>0$ for $1\leq j \leq N$.
Then we have already seen that the contour deformation avoids the
integrand's singularity at $w_S$.

Now, under the same conditions, consider any array of real numbers $\lambda_j$
($1\leq j \leq N$) which are non-negative and for which at least one is positive.
Then $\sum_{j=1}^N\lambda_j D_j(v(w_S)) >0$ and hence the dual vector
$\sum_{j=1}^N\lambda_jD_j$ is nonzero.  Therefore there is no Landau point.

In order to get the desired necessary and sufficient conditions, we also
need to prove the reverse implications, which is quite non-trivial.  It is
interesting to observe, that even to get one of the directions of
implication in the main theorem requires that we use of a non-trivial
direction of implication in one or other of the subsidiary theorems.

\section{Useful generalizations}
\label{sec:extra.cases}

In this section, I gather a couple of illustrations of applications of the
derived theorems to situations beyond the standard analyses of
singularities of ordinary Feynman amplitudes.  The standard application to
momentum-space Feynman graphs is illustrated in App.\
\ref{sec:self-energy}.

\subsection{Glauber region}
\label{sec:Glauber}

One example of the need for a more general derivation of the Landau
criterion is the general analysis of the Glauber region given in Sec.\
5.6.3 of Ref.\ \cite{Collins:2011qcdbook}.  In that situation,
approximations have been made for a Feynman graph that are valid in a
certain region of its loop momenta, with the momenta being classified into
soft, collinear, and hard categories.  It is desired to determine when
there is a trap in the Glauber region; this is important because the
approximations used for soft momenta fail when a soft momentum is of the
kind called Glauber.  The importance of this issue is that in some
situations, there are uncanceled Glauber contributions, and these break
standard formulations of factorization in interesting cases, e.g.,
\cite{Collins:2007nk}.  

An appropriate method \cite[Sec.\ 5.6.3]{Collins:2011qcdbook} to locate
Glauber contributions uses a version of the Libby-Sterman argument, but
applied to the approximated graph in which standard soft and collinear
approximations have been made.  If a contour deformation cannot be made to
avoid the Glauber regions, then there is a corresponding exact pinch in the
approximated graph.  The use of the relevant Landau condition gives a
necessary and sufficient condition for the Glauber pinch.

The importance of this analysis, with its systematic use of an improved
Landau analysis, is that it can be used to locate in full generality where
extra regions and scalings in momentum space are important beyond the usual
classification into soft, collinear, and hard, with associated scalings of
momentum components.

\subsection{Coordinate space}
\label{sec:coord}

Another example is the extraction of coordinate-space properties of
amplitudes.  For example, the Fourier transform of a free propagator is
\begin{equation}
\label{eq:coord}
  S_F(x) = \int \frac{ \diff[n]{k} }{ (2\pi)^n }
           e^{-i k \cdot x}
           \frac{ i }{ k^2-m^2 +i\epsilon },
\end{equation}
where $n$ is the number of space-time dimensions, and the limit
$\epsilon \to 0$ from positive values is implicit, as usual.  

Suppose we are interested in how this integral behaves when $x$ is scaled
to large values: $x \mapsto \kappa x$ with $\kappa \to \infty$.  Of course, in the particular case
given, a solution can be found analytically, since the free propagator is a
kind of Bessel function with known asymptotics.  But it is important to
have a method that can be applied much more generally without appealing to
properties of known special functions.  To do this, we observe that over
much of the space of real $k$, one can deform the contour of $k$ so as to
give $k\cdot x$ a negative imaginary part.  But near the pole at $k^2=m^2$, we
need to have the deformation compatible with the $i\epsilon$ prescription in the
denominator.  If these two conditions on contour deformation are
incompatible, then we must leave the contour on the real ``axis'' and get
an unsuppressed contribution to the large $\kappa$ asymptotics.

Let us specify the deformed contour as
\begin{equation}
  k = k_R + iv(k_R).
\end{equation}
Then the condition for an exponential suppression is
\begin{equation}
\label{eq:exp.supp}
  -v(k_R)\cdot x > 0,
\end{equation}
while the condition for avoiding the propagator pole is\footnote{In this
  statement, we are assuming that avoiding the pole can always be done by a
  contour deformation that gives a positive first-order shift to the
  imaginary part of the denominator.  The complications hidden in
  justifying this assumption have already been mentioned.  Nevertheless,
  use of the methods of Sec.\ \ref{sec:deal.with.zero.first.order} will
  show that an exponential suppression with a singularity-avoiding contour
  occurs if and only if there is a contour obeying Eqs.\
  (\ref{eq:exp.supp}) and (\ref{eq:pole.avoid}).}
\begin{equation}
\label{eq:pole.avoid}
  v(k_R) \cdot k_R > 0 \quad \mbox{when $k_R^2=m^2$}.
\end{equation}

These conditions are incompatible when $x$ is proportional to $k_R$ with a
\emph{positive} coefficient and $k_R$ is on-shell.  If $k_R$ has positive
energy, then the relevant values of $x$ are future pointing in the same
direction as $k_R$, while if $k_R$ has negative energy, $x$ is past
pointing.

Given a value of $x$, this observation determines which values (if any) of
$k_R$ give unsuppressed contributions to $S_F(x)$.  Here ``unsuppressed''
means ``not exponentially suppressed''; this use of ``unsuppressed'' allows
it to include merely ``power suppressed''.

Now an on-shell value of $k_R$ is time-like. Hence, when $x$ is space-like,
there is no value of $k_R$ giving an unsuppressed contribution. Then there
is no obstruction to deforming the contour, and an exponential suppression
of $S_F(x)$ is a consequence.

In contrast, when $x$ is time-like, the deformation cannot be made, and
that gives power-law behavior as $x$ is scaled.  The dominant contribution
comes from near the pole in momentum-space, and the asymptote can be
extracted by suitable approximation methods.  These methods continue to
apply if the free momentum-space propagator is replaced by the full
propagator in an interacting theory, which has a more general dependence on
momentum, but with its strongest singularity still being a pole at the
physical particle mass.

It is worth noting that similar methods can also be applied to get from the
behavior of a coordinate-space Green function to particular properties in
momentum space. Thus one can determine for the vertices of a graph the
dominant regions in coordinate space that contribute to a particular
process.  We leave the systematic codification of such results to future
work.

Some relevant recent work is by Erdo\u{g}an and Sterman
\cite{Erdogan:2014gha,Erdogan:2016ylj,Erdogan:2017gyf}.

\section{Literature review}
\label{sec:literature}

In this section I assess some of the classic literature about the Landau
analysis.  Since many of these works continue to be cited regularly as the
primary sources for results on singularities and pinches of contours, it is
useful to examine their arguments in detail.  The review in this section
extends observations already made in the introduction.

It should be observed that typical treatments rely on the use of Feynman
parameters to combine the denominators into a single denominator.  Then
they examine the conditions for a pinch of the integration contour, rather
than trying to create a more detailed geometrical argument that applies to
the multiple-denominator situation.  This rules out any easy application of
the methods to more general situations, e.g., examining properties of
integrals involving coordinate space properties, as in Sec.\
\ref{sec:coord}, or the issues of algorithmic deformation of a contour for
numerical integration of a Feynman graph in the pure momentum-space
formulation.

\paragraph{Landau \cite{Landau:1959fi}}

Landau's paper \cite{Landau:1959fi} gives the original treatment of the his
criterion for singularities of a Feynman graph as a function of its
external parameters.

The analysis solely uses the Feynman parameter representation in the form
(\ref{eq:F.graph.param.mom}).  It relies on the denominators being those of
standard Feynman graphs.  Then in Landau's Eq.\ (4) the single denominator
is written as $\phi+K(k',l',\ldots)$, where $\phi$ is a function only of the external
parameters and $K$ is a homogeneous quadratic form in a set of variables
that are formed by a (parameter-dependent) linear transformation from the
original loop momenta.  This by itself rules out the case that some
denominators have linear dependence on some (or all) loop momenta.  Such
cases arise in practice.  For example, in QCD applications we have cases
with Wilson denominators.  In such a situation, the equivalent of $K$ is not
a homogeneous quadratic function.

The argument then continues to determine that a singularity of the integral
(as a function of external parameters) occurs when there is a point in
integration space where the denominator and its first derivative vanish.
No detailed argument is given, the core parts of the argument being treated
as ``easy to verify''.  However, a detailed derivation, in Sec.\
\ref{sec:one-denom} of the present paper, is not at all easy.  In fact the
proof fails whenever the matrix of second derivatives of the denominator
has an eigenvector with zero eigenvalue.  This situation does in fact
sometimes arise in practice, as mentioned in a later paper by Coleman and
Norton \cite{Coleman:1965xm}.

Moreover, Landau's argument is rather difficult to apply as written if
there are massless particles, as is essential in applications to QCD
factorization.  In contrast, the methods of the present paper do apply
unchanged to such cases.  They are also applied directly to the momentum
space integral without an appeal to Feynman parameters.

A minor problem is that the $i\epsilon$ prescription is not mentioned explicitly
even though that is critical in determining whether or not there is a
pinch.

\paragraph{Coleman and Norton \cite{Coleman:1965xm}}

Coleman and Norton \cite{Coleman:1965xm} again use a parametric
representation.  In the first part of the paper, they discuss the version
with both momentum and parameter integrations.  They state, rather like
Landau, that to get pinch there needs to be either a coalescing pair of
singularities or an end-point singularity.  This immediately gives the
Landau equations.  However, given this first part of the derivation, the
Landau condition is clearly necessary but not sufficient, since it has not
yet been determined whether or not coalescing singularities actually pinch
the contour.  It is also not really obvious what the term ``coalescing
singularities'' means except in one dimension.  In addition, it is not
clear why attention is restricted to pairs of singularities,

To provide an actual proof of necessity and sufficiency, Coleman and Norton
perform the momentum integrals analytically, and work with an integral
solely over the parameters, i.e., an integral of the form
(\ref{eq:F.graph.param.only}), and restore the $i\epsilon$.  It is not actually
clear why they switch to this kind of integral.  The rest of their argument
appears to apply to a general multidimensional integral (subject to certain
conditions on the quadratic terms, as we will see).  Thus their arguments
appear to apply equally to the integral with both momentum and parameter
integrations.  But they clearly think that this approach would fail.

Then they examine the denominator in the neighborhood of a point where both
the denominator and its first derivative are zero.  This is a place where
the Landau condition is satisfied, because of the zero first derivative.
They expand the denominator to quadratic order in small deviations from the
candidate pinch location, which gives a formula for the denominator of the
form
\begin{equation}
  A = \frac12 \sum_{ij} E_{ij} \eta_i\eta_j.
\end{equation}
The authors then state that it is easy to show that the contour is trapped,
but only if none of eigenvalues of $E_{ij}$ is zero.  However, as will be
seen later in the present paper, in Sec.\ \ref{sec:one-denom}, an adequate
proof is not entirely trivial.  The proof does indeed fail when zero
eigenvalues exist.  It is not at all clear whether the failure can be
remedied, or how that can be done.

That cases of zero eigenvalues arise in massive theories in reality is
mentioned; they occur only at ``very exceptional points''.  The reader is
referred to Ref.\ \cite{Eden:1961} for more details.  But that paper
appears not to contain a clear statement of whether such singularities
can occur in the physical region.  Considerable further work is apparently
needed to resolve the issue.

In contrast, in a massless theory, a much simpler failure happens, as will
be explained in this paper in App.\ \ref{sec:F.param.massless} for the case
of a one-loop self-energy with massless particles.  This graph has a
well-known collinear pinch when the external momentum is light-like.  But
it is found that in the parameter integral there is no pinch that
corresponds to the collinear pinch in momentum space.  

A further complication is found in App.\ \ref{sec:F.param.linear} in an
example graph where propagators are linear in a momentum component.  For
that graph the pure parameter integral has a pinch independently of whether
there is a pinch in the momentum integral.

Evidently Coleman and Norton have assumed that a pinch in momentum space
occurs if and only if a corresponding pinch occurs in parameter space, and
that this is so obvious as to need neither mention nor proof.  The examples
just mentioned show that the implication is not even correct, in general,
even if it works in the case of standard massive Feynman graphs.

After giving their derivation of the Landau condition, Coleman and Norton
derive their well-known result that a pinch configuration corresponds to a
situation with classical particles propagating and scattering in space-time
with momenta corresponding to the on-shell momenta of the lines
participating in the pinch.

It is important to remember that it is not the result that breaks down, but
the proof. But the proof's breakdown is a symptom of things that were not
understood.  For example, in Apps.\ \ref{sec:F.param.massless} and
\ref{sec:F.param.linear} are given counterexamples that imply a failure of
Coleman and Norton's proof.  But in both cases the Landau condition
correctly locates pinch(es) in the momentum-space integral.  The general
proof in the present paper applies perfectly well to those cases.  Of
course the new proof is much longer than those in the old papers.

\paragraph{ELOP \cite{ELOP}}

The venerable book by Eden et al.\ \cite{ELOP} remains a standard reference
for analyticity properties of Feynman graphs.  Therefore it is worth
carefully assessing its treatment.  As was remarked in the introduction,
the authors do say that their treatment lacks rigor, but do not make
explicit what is not rigorous.

After a clear discussion of the one-dimensional case, they come to the
multidimensional case on p.\ 47.  Their subject matter is a general
integral over multiple complex variables, but without the further
``physical region'' restrictions inherent in our (\ref{eq:integral}); these
are a reality property of the denominators and an $i\epsilon$ prescription.
Theirs is therefore in principle a more general treatment.  Their equations
for singularity surfaces $S_r=0$ correspond to the equations $A_j=0$ for
the zeros of our denominator factors.

The first problem is that they say that when a singularity surface advances
on the contour of integration, they say that if the singularity is to be
avoided, the contour should be distorted ``in the direction of the normal''
to the singularity surface.  This appears to say that there is a unique
direction in which to distort the contour.  But we have seen that in fact
there is a whole half-space of possible directions, and it is absolutely
necessary to take this into account.  In addition, the concept of an
unambiguous ``normal'' to a surface only makes sense in a Euclidean space,
which is not the case for multidimensional complex variables with which we
are concerned.

In addition, they appear to assume as so obvious as not to need a proof
that for a contour deformation to avoid a singularity surface it must give
a nonzero first order shift in the denominator factors (or the equivalent
in their more general integral).  But this definitely not the case ---
see App.\ \ref{sec:2D.first.order}.

Then in Eq.\ (2.1.19) they assert the conditions for singularity surfaces
to trap the integration contour.  These are a form of the Landau condition.
But no proof and no reference to a proof is given.  It is as if they think
the equation is obvious.  But as we will see in Secs.\
\ref{sec:overall}--\ref{sec:deal.with.zero.first.order}, the condition is
rather non-trivial to derive.  They continue to refer to normals to
surfaces, but have evidently confused the concept with the relevant one of
dual vectors, so that there is considerable conceptual confusion not
conducive to adequate reasoning.  It is not at all obvious whether they
consider the conditions to be both necessary and sufficient, and why.

Finally, their statement (2.1.19b) of the condition for a version of a
Landau point lacks the positivity constraint needed for the kind of
``physical region'' pinch we consider.  Recall that the positivity
constraint is that the $\lambda_j$ parameters in Defn.\ \ref{def:Landau.point}
are non-negative, and that at least one is positive.  While the more
general version is appropriate for pinches outside the physical region,
further conditions are needed to determine whether or not there is a pinch.
This can be seen from the fact that their version of the condition is
trivially satisfied whenever the number of singularity surfaces is larger
than the dimension of the integration space, as the authors do indeed
observe.  Hence some stronger condition than (2.1.19b) is needed to provide
sufficient conditions to determine that there is a pinch.  

In stark contrast, for a physical region pinch, the Landau condition (with
the positivity constraint) is both necessary and sufficient.  Of course
this only applies given both the reality conditions on our denominator
factors $A_j$ and the $i\epsilon$ prescription; the relevant theorem is Thm.\
\ref{thm:main.contour}, and its very non-trivial proof appears in later
sections.  (Our proof also has some further restrictions, given in the
statement of the theorem; these are obeyed by standard and by important
non-standard Feynman graphs.)  It is worth re-emphasizing that it is solely
the physical-region pinches that are relevant to QCD applications, and the
positivity constraints on the $\lambda_j$ parameters in the Landau point
definition are very important in delimiting collinear configurations of
partons.

The positivity conditions do appear in the ELOP treatment for physical
region pinches/singularities, but only when they consider Feynman graphs in
a Feynman parametric representation.  Then the positivity conditions arise
from the range of the Feynman parameters.  But they do not derive the same
constraints when the derive the conditions for a pinch from the pure
momentum-space formula for a Feynman graph.  Moreover, working in parameter
space leads to the issues explained in the analysis of the Coleman-Norton
treatment.

\section{The geometrical theorem: Set up}
\label{sec:overall}

In this section and the next two sections, we will prove the last of the
theorems listed in Sec.\ \ref{sec:primary.theorems}, i.e., the purely
geometric Thm.\ \ref{thm:main.geom}.  It can be regarded as giving a
compatibility condition for linear constraints on directions in a vector
space.

Throughout the treatment of this theorem, we work with a
finite-dimensional\footnote{The assumption of finite dimensionality can be
  relaxed, but we will not need to do so.}  real vector space $V$ of
dimension $d$, and we suppose given a list $\Dset$ of dual vectors $D_j$ on
$V$ ($1\leq j \leq N$).  By definition, each $D_j$ is a real-valued linear
function from $V$ to the space of real numbers.  The constraints on vectors
with which we are concerned are written $D_j(v) > 0$.  In component
notation, we write
\begin{equation}
  \label{eq:Dj.compt}
  D_j(v) = \sum_\alpha D_{j\alpha}v^\alpha,
\end{equation}
where $v^\alpha$ denotes the components of $v$ with respect to some basis.  But
we will use coordinate-independent notation much of the time.  The space of
all dual vectors is a vector space $V^*$ of the same dimension as $V$ (if
$V$ is finite dimensional).  We do not assume that there is any metric
given on $V$ or $V^*$.

Observe that although our original subject was integration in a complex
space, the manipulations involved in analyzing possible directions of
deformation, and hence of the constraints $D_j(v)>0$, only concern a real
vector space.

In the integration problem, we were concerned with whether or not a contour
deformation exists that avoids a singularity of the integrand.  In the
geometric problem that we are addressing at the moment, a concept
corresponding to singularity avoidance in integrations is what we call a
``good direction'', defined by
\begin{definition}
  \label{def:all.dir}
  A \emph{good direction} for $(D_1,\dots,D_N)$ is defined to be a $v \in V$
  such that $D_j(v)>0$ for all $j$.
\end{definition}

Throughout this and the next two sections, we use the terminology of Landau
points and Landau conditions that was defined in Defns.\
\ref{def:Landau.point} and \ref{def:Landau.condition}, names motivated by
the application to integrals. The theorem to be proved is that a good
direction exists for $(D_1,\dots,D_N)$ if and only if there is no Landau
point.  Alternatively, there is no good direction if and only if there is
at least one Landau point.

We have already observed, in Sec.\ \ref{sec:elem.parts}, that if there is a
good direction then there is no Landau point and hence the Landau condition
holds.  Equivalently, if the Landau condition holds, then there is no good
direction.  

To complete the proof of Thm.\ \ref{thm:main.geom}, we need to prove the
converse, i.e., that if there is no good direction then there is a Landau
point.  What is needed is to exclude with full generality the possibility
that there might fail to exist both a Landau point and a good direction.

In simple examples, it is not too hard to see that the theorem is valid,
with both directions of implication; such examples can often be visualized.
But in general the vector space $V$ can be of arbitrarily high dimension,
and the number of $D_j$ can be arbitrarily large.  Then visualizing the
details of the proof is hard.  Thus careful abstract arguments are needed.
In making the detailed analysis, we will encounter methods and results that
should be useful in algorithmic determining good directions for contour
deformations in numerical integration over loop momenta in Feynman graphs.

We will start in Sec.\ \ref{sec:pos.sets} by characterizing properties of
the set of good directions, and especially the boundaries of this set.
Then in Sec.\ \ref{sec:landau.proof}, we will use these properties to
complete the proof of the geometric theorem.  A reader may find it unclear
what the motivation is for deriving some of the earlier properties, i.e.,
those in Sec.\ \ref{sec:pos.sets}.  So it may be useful to skip ahead to
Sec.\ \ref{sec:landau.proof} to see what use is made of the results of
Sec.\ \ref{sec:pos.sets}.

\section{Geometry of positive regions of sets of dual vectors}
\label{sec:pos.sets}

\subsection{Setting up the problem}

We use the notation of the previous section, and define the positive region
of a list $\Dset=(D_1,\ldots,D_N)$ of dual vectors by
\begin{definition}
  We define $P_{\Dset}$ to be the region of $V$ in which all the
  $D_j$s in $\Dset$ are strictly positive:
  \begin{equation}
    P_{\Dset} \eqdef \left\{ v \in V : \forall D_j \in \Dset, D_j(v) > 0 \right\}.
  \end{equation}
  We call this the \emph{``positive region''} of $\Dset$.
\end{definition}
The overall issue we are addressing is the determination of whether or not
$P_{\Dset}$ is empty. 

In this section, we will examine the case that $P_{\Dset}$ is non-empty,
and determine properties of its boundary that we will need later.  Observe
that if $P_{\Dset}$ is non-empty, then all the $D_j\in\Dset$ are necessarily
non-zero.

\begin{definition}
  The complement of $P_{\Dset}$ is notated as:
  \begin{equation}
    \widehat{P}_{\Dset} \eqdef V \setminus P_{\Dset}
                 = \left\{ v \in V : \exists D_j \in \Dset : D_j(v) \leq 0 \right\}.
  \end{equation}

\end{definition}

We make a lot of use of the intersection of the kernels of $D_j$.
So we define
\begin{definition}
  \begin{equation}
    K_{\Dset} \eqdef \left\{ v \in V : \forall D_j \in \Dset, D_j(v) = 0 \right\}.
  \end{equation}
\end{definition}
\begin{definition}
  Define $n_{\Dset}$ to be the codimension of $K_{\Dset}$ in $V$, i.e.,
  $n_{\Dset} = d - \dim(K_{\Dset})$.
\end{definition}
It is well-known that $K_{\Dset}$ is a vector subspace of $V$.  When
$P_{\Dset}$ is non-empty, $K_{\Dset}$ cannot be the whole of $V$, so
that in this case its codimension obeys $n_{\Dset} \geq 1$.

We can decompose $V$ as a direct sum of the form
\begin{equation}
\label{eq:V.decomp}
  V = V_{\perp \Dset} \oplus K_{\Dset}.
\end{equation}
The dimension of $V_{\perp \Dset}$ is $n_{\Dset}$.  Note that $V_{\perp
  \Dset}$ is non-unique, since its basis vectors can be changed by the
addition of elements of $K_{\Dset}$.  If we are given that $P_{\Dset}$
is non-empty, then there must be a region of $V_{\perp \Dset}$ where the
$D_j$ are positive.

The critical result that we are working towards in this section is Thm.\
\ref{thm:PD.decomp} below, where we find a set of non-zero ``edge vectors''
$e_L$ for $P_{\Dset}$ such that every element of $v$ of $P_{\Dset}$ has the
form $v = \sum_L C_L e_L + v_K$, where all the $C_L$ are positive real
numbers, $C_L>0$, and $v_K \in K_{\Dset}$.

To derive Thm.\ \ref{thm:PD.decomp}, we will need a series of subsidiary
results, many of which are very elementary, and are obvious in
low-dimensional examples.  But these results need to be explicitly stated
in order to ensure that the main theorem is properly proved in a space of
arbitrarily high dimension; their cumulative effect is quite non-trivial.
Many of the subsidiary results are likely to be useful in themselves for
applications, e.g., for searching for good directions to deform a contour
when there is no pinch.

\subsection{Elementary properties of \texorpdfstring{$P_{\Dset}$}{PD}}

\begin{theorem}
\label{thm:P.basic}
  $P_{\Dset}$ obeys
  \begin{enumerate}[(a)]
  \item It is convex, i.e., if $v_1,v_2 \in P_{\Dset}$ and $\kappa$ is any
    real number between 0 and 1 inclusive (i.e., $0\leq\kappa\leq1$), then 
    $\kappa v_1 + (1-\kappa)v_2 \in P_{\Dset}$.
  \item If $v\in P_{\Dset}$ then $\lambda v\in P_{\Dset}$ for any positive real
    $\lambda$. 
  \item $P_{\Dset}$ is connected.
  \item It is an open set.
  \end{enumerate}
\end{theorem}
\begin{proof}
  Suppose that $v_1,v_2 \in P_{\Dset}$, that $\lambda_1$ and $\lambda_2$ are real
  numbers, that both $\lambda_1,\lambda_2 \geq 0$, and that at least one is strictly
  positive.  Then each $D_j(\lambda_1v_1 + \lambda_2v_2)=\lambda_1D_j(v_1) + \lambda_2D_j(v_2)$ is
  strictly positive, and hence $\lambda_1v_1 + \lambda_2v_2 \in P_{\Dset}$.  (This
  demonstrates that $P_{\Dset}$ is an example of a convex cone in
  mathematical terminology.)

  Properties (a) and (b) immediately follow, and then so does (c) from
  (a). 

  To derive part (d), let $v \in P_{\Dset}$, and let $l = \min_{D_j \in \Dset}
  D_j(v) > 0$. Now let $\delta v$ be another element of $V$. Then
    \begin{equation}
      D_j(v+\delta v) = D_j(v) + D_j(\delta v) \geq l + D_j(\delta v).
    \end{equation}
    For all small enough $\delta v$, we have $|D_j(\delta v)|<l$ for every $D_j\in
    \Dset$, and then $v+\delta v \in P_{\Dset}$.  Hence $P_{\Dset}$ is open.
\end{proof}

Since $P_{\Dset}$ is open, it is a manifold of the same dimension as
$V$, i.e., $d$, provided only that it is non-empty.

\begin{quote}
  \textbf{From now on, we will assume that $P_{\Dset}$ is non-empty,
    unless explicitly stated, and will only reiterate this assumption
    when it seems particularly important.}
\end{quote}

\subsection{Basic properties of the boundary of
  \texorpdfstring{$P_{\Dset}$}{PD}}
\label{sec:bdy.props}

We now consider the boundary $\partial P_{\Dset}$ of $P_{\Dset}$, i.e., the
set of points of $V$ that are limit points both of $P_{\Dset}$ and its
complement $\widehat{P}_{\Dset}$.

\begin{theorem}
\label{thm:bdy.char}
  If $P_{\Dset}$ is non-empty, the boundary of $P_{\Dset}$ is
  characterized by 
  \begin{align}
  \partial P_{\Dset}
  = \bigl\{ v \in V : {}& \forall D_j \in \Dset, D_j(v) \geq 0 
\nonumber\\
           & \mbox{ \rm and } \exists D_j \in \Dset : D_j(v) = 0
    \bigr\}.
  \end{align}
\end{theorem}
It follows that the boundary is contained in $\widehat{P}_{\Dset}$.

\begin{proof}
Suppose we have a point $v \in \partial P_{\Dset}$.  Then there is a sequence
$v_a$ in $P_{\Dset}$ whose limit is $v$.  So for all $D_j\in \Dset$
\begin{equation}
  D_j(v) = \lim_{a\to\infty} D_j(v_a) \geq 0. 
\end{equation}
If $D_j(v)$ were also nonzero for all $D_j$, then it would be in
$P_{\Dset}$.  Since $P_{\Dset}$ is open, this would imply that $v$ is not
in its boundary.  Hence we must have $D_j(v)=0$ for at least one $D_j$.

Conversely, suppose we have a point $v \in V$ for which all the $D_j(v)$ are
positive or zero, and at least one of which is zero, i.e.,
\begin{equation}
\label{eq:P.bc}
    \forall D_j \in \Dset, D_j(v) \geq 0,
  \mbox{ and }
    \exists D_j \in \Dset : D_j(v) = 0.
\end{equation}
Then choose $\delta v \in P_{\Dset}$.  For every positive real number $\lambda$,
$D_j(v+\lambda\delta v) = D_j(v) + \lambda D_j(\delta v) > 0$, so that $v+\lambda\delta v \in P_{\Dset}$.
Thus $v$ is a limit point of $P_{\Dset}$.  But it is not in $P_{\Dset}$, so
it must be in the complement $\widehat{P}_{\Dset}$. It follows that $v$ is
trivially a limit point of $\widehat{P}_{\Dset}$.
\end{proof}

\begin{theorem}
   (a) The subspace where all the $D_j$s are zero is inside the boundary
   of $P_{\Dset}$.  I.e., $K_{\Dset} \subseteq \partial
   P_{\Dset}$.
   \\
   (b) $\partial P_{\Dset}$ is connected.
\end{theorem}

\begin{proof}
Every element $k$ of $K_{\Dset}$ obeys $D_j(k)=0$, for all $j$, and is
thus in $\partial P_{\Dset}$, by Thm.\ \ref{thm:bdy.char}.  Hence
$K_{\Dset} \subseteq \partial P_{\Dset}$.

Since the zero vector is in $K_{\Dset}$ it is also in $\partial
P_{\Dset}$.  For any $v$ in $\partial P_{\Dset}$, $\lambda v$ is also
in $\partial P_{\Dset}$ whenever $\lambda \geq 0$.  This gives a line
connecting an arbitrary element of $\partial P_{\Dset}$ to one
particular element, i.e., the zero vector.  Hence $\partial P_{\Dset}$
is connected.
\end{proof}

For our purposes, the interesting parts of $\partial P_{\Dset}$ are those that are
not in $K_{\Dset}$, i.e., where at least one $D_j(v)$ is strictly positive.
Therefore we define
\begin{definition}
  The non-trivial part of the boundary of $P_{\Dset}$ is
  \begin{equation}
    \widetilde{\partial} P_{\Dset} \eqdef \partial P_{\Dset} \setminus K_{\Dset}.
  \end{equation}
\end{definition}
The set $\widetilde{\partial} P_{\Dset}$ may be empty; our later work
shows that this happens if and only if $n_{\Dset}=1$ (or, of course if
$P_{\Dset}$ itself is empty).

From Thm.\ \ref{thm:bdy.char} it follows that the non-trivial part of the
boundary obeys
\begin{align}
  \widetilde{\partial} P_{\Dset}
  = \bigl\{ v \in V : {}& \forall D_j \in \Dset, D_j(v) \geq 0 
\nonumber\\
           & \mbox{ \rm and } \exists D_j \in \Dset : D_j(v) = 0
\nonumber\\
           & \mbox{ \rm and } \exists D_j \in \Dset : D_j(v) > 0
    \bigr\},
\end{align}
i.e., all the $D_j(v)$ are non-negative, at least one is zero, and at
least one is positive. 

\begin{definition}
  Here we define some auxiliary objects at a point $w$ that is in the
  non-trivial part of the boundary, $w \in \widetilde{\partial} P_{\Dset}$.
\begin{enumerate}[(a)]

\item The sets of $D_j$ with zero and non-zero values are:
\begin{subequations}
\label{ref:Zw.hatZw}
\begin{align}
  Z(w) & \eqdef \left\{ D_j \in \Dset : D_j(w) = 0 \right\},
\\
  \widehat{Z}(w) & \eqdef \left\{ D_j \in \Dset : D_j(w) > 0 \right\}.
\end{align}
\end{subequations}
Given $w$, each $D_j$ is in exactly one of these sets, of course.  Both
sets are non-empty when $w$ is in the non-trivial part of the boundary.

\item The minimum non-zero value of the $D_j(w)$s is:
\begin{equation}
  m(w) \eqdef \min_{D_j \in \widehat{Z}(w)} D_j(w) > 0.
\end{equation}

\item Let $K(w)$ be the intersection of the kernels of those $D_j$
that are in $Z(w)$:
\begin{align}
  K(w) & \eqdef \left\{ v \in V : \forall D_j \in Z(w) : D_j(v) = 0 \right\}
\nonumber\\
       & = \cap_{D_j \in Z(w)} \ker(D_j)
\end{align}

\item Let $n(w)$ be the codimension of $K(w)$, so that the dimension
of $K(w)$ is $d-n(w)$.  
\end{enumerate}
\end{definition}
Note that $w$ is one (non-zero) element of the subspace $K(w)$.

Since $P_{\Dset}$ is non-empty, there are vectors $v$ for which $D_j(v)>0$
for all $j$.  It follows that $K(w)$ cannot be the whole of $V$.  Hence
\begin{equation}
  n(w)\geq1 .
\end{equation}

\subsection{Decomposition of the boundary
  of \texorpdfstring{$P_{\Dset}$}{PD}} 
\label{sec:bdy.decomp}

In this section, we show that the boundary of $P_{\Dset}$ can be
decomposed into a hierarchy of disjoint flat segments.  On each of
these one set of $D_j$s is strictly positive and the others are zero. 

First, given $w \in \widetilde{\partial} P_{\Dset}$, we construct the
boundary segment of which it is part.  We define
\begin{align}
   B(w) \eqdef {}& \bigl\{ v \in K(w): \forall D_j \in \hat{Z}(w), D_j(v) > 0 \bigr\}
\nonumber\\
   ={}& \bigl\{ v \in V: \forall D_j \in Z(w), D_j(v) = 0;
\nonumber\\
   &\hspace*{1cm}
    \mbox{and}\ \forall D_j \in \hat{Z}(w), D_j(v) > 0 \bigr\}.
\end{align}
Note that $B(w)$ is a subset of $K(w)$.

The boundary segments have the following elementary properties
\begin{theorem}
\label{thm:B.props}
\begin{enumerate}[(a)]

\item $B(w)$ is convex.
\label{thm:part:B.convex}

\item Whenever $v\in B(w)$, so is $\lambda v$ for positive  $\lambda$.
\label{thm:part:B.scale}

\item $B(w)$ is a flat connected manifold of the same dimension as
$K(w)$, i.e., $d-n(w)$.
\label{thm:part:B.manif}

\item For every point $v\in B(w)$,
\label{thm:part:B.const.ZK}
\begin{equation}
\label{eq:B.ZK}
\begin{split}
    Z(v) = Z(w), \quad
    \hat{Z}(v) = \hat{Z}(w), \\
    K(v) = K(w), \quad
    n(v) = n(w).
\end{split}
\end{equation}

\item When $v \in B(w)$, we have $B(v)=B(w)$.
\label{thm:part:B.const.B}

\item For any $v$ and $w$ in $\widetilde{\partial} P_{\Dset}$,either $B(v)$ and
  $B(w)$ are non-intersecting or they are equal.  It immediately follows
  that the boundary $P_{\Dset}$ is decomposed into a set of disjoint flat
  segments.
\label{thm:part:B.overlap}

\end{enumerate}
\end{theorem}
\begin{proof}
Parts (\ref{thm:part:B.convex}) and (\ref{thm:part:B.scale}) follow by the
same method used to prove the corresponding properties for $P_{\Dset}$. 

It immediately follows that $B(w)$ is connected and flat.  As to the
dimension, first note that from its definition, $B(w)$ is contained in the
kernel space $K(w)$, so that its dimension is at most that of $K(w)$, i.e.,
$d-n(w)$.  Furthermore, let $\delta w$ be any element of $K(w)$.  For all small
enough $\delta w$, $w+\delta w$ is in $B(w)$.  This is because when $D_j\in Z(w)$
$Dj(w+\delta w)=D_j(w) + D_j(\delta w)=0$, and because when $D_j\in\hat{Z}(w)$ and $\delta w$
is small enough the value of $D_j(\delta w)$ cannot compensate the positive value
of $D_j(w)$.  Hence the dimension of $B(w)$ is at least $d-n(w)$.  

Property (\ref{thm:part:B.manif}) now follows.

Next suppose $v \in B(w)$.  By the definition of $B(w)$, $D_j(v) = 0$ for
every $D_j$ in $Z(w)$, and $D_j(v) > 0$ for every $D_j$ in $\hat{Z}(w)$.
Hence the set $Z(v)$ is the same as $Z(w)$, since it is the set of $D_j$
for which $D_j$ is zero at $v$.  From this follows all of (\ref{eq:B.ZK}).

It then follows from the definition of $B(v)$ that $B(v)=B(w)$ whenever
$v\in B(w)$, which thereby proves property (\ref{thm:part:B.const.B}).

Now consider $B(v)$ and $B(w)$ for two points $v$ and $w$.  Either
they do not intersect or they intersect.  In the second case, pick $k$
in the intersection.  From the previous result it follows that
$B(k)=B(v)$ and $B(k)=B(w)$, and hence that $B(v)=B(w)$.  This proves
property (\ref{thm:part:B.overlap}).
\end{proof}

Since the sets $Z(v)$, $\hat{Z}(v)$, and $K(v)$, and the number $n(v)$ are
constant on any given boundary segment $B$, we can say that $Z$ etc are
determined by the set of points $B$.  Thus we can write
\begin{align}
  Z(B) &\eqdef \{ D_j : D_j(v) = 0 \mbox{~for every~} v \in B \}
\nonumber\\
       &= Z(v) \mbox{ for every $v \in B$},
\end{align}
and similarly for $\hat{Z}(B)$, $K(B)$ and $n(B)$.

Notice that subspace $K(B)$ contains the common kernel subspace $K_{\Dset}$
of all the $D_j$.  For a non-trivial boundary segment $B$ this implies that
the subspace $K(B)$ is strictly larger than $K_{\Dset}$.  This is because
in this case there are points of $K(B)$ where at least one of $D_j$ is
non-zero; these points cannot be in $K_{\Dset}$.  Hence the codimensions
obey $n(B) \leq n_{\Dset}$, with $n_{\Dset}$ being the codimension of the
smallest (trivial) boundary segment, i.e., the common kernel of all the
$D_j$, and with equality only for the trivial boundary segment.

Observe that each boundary segment $B$ obeys all of the properties of
positive regions, but with respect to $K(B)$ instead of the whole
space $V$, and with respect to $\widehat{Z}(B)$ instead of $\Dset$.
In particular, it is an open and convex set in $K(B)$.  Moreover, the
same arguments as given above for $P_{\Dset}$ show that each $B$
itself has a boundary consisting of boundary segments, which are also
boundary segments of $P_{\Dset}$ itself, with all the associated
properties.

There is in fact a hierarchy of boundary segments, for which it is possible
to prove the following results:
\begin{enumerate}
\item The unique lowest dimension boundary segment is the subspace
  $K_{\Dset}$, of dimension $d-n_{\Dset}$.
\item There are boundary segments of every dimension between the minimum
  dimension $d-n_{\Dset}$ and the maximum dimension $d-1$, inclusive.
\item Each boundary segment of non-maximal dimension is a boundary segment
  of a boundary segment of one dimension higher.  If it has the maximal
  dimension $d-1$, it is a boundary segment only of $P_{\Dset}$ itself.
\item $P_{\Dset}$ and non-minimal boundary segments have one or more
  boundary segments of one dimension lower.
\end{enumerate}
In visualizable examples, the existence of this hierarchy and many of its
properties are quite obvious.  But the general case needs a proof, which is
non-trivial.  For the purposes of this paper, we will not need the whole
collection of properties of the hierarchy, so we will not make all the
proofs.

What we do need are the boundary segments of one dimension higher than the
minimal dimension, whose existence we will prove.  Projected onto a
subspace $V_{\perp \Dset}$ that gives a decomposition of the form in Eq.\
(\ref{eq:V.decomp}), the next-to-minimal boundary segments become line
segments. This leads us to the concept of edge vectors specifying the
directions of the next-to-minimal boundary segments.  The edge vectors play
a critical role in our later analysis.

\subsection{Edge vectors \texorpdfstring{$e_L$}{eL}}

Now we construct what we call the edge vectors $e_L$ of $P_{\Dset}$.  Each
edge vectors has a label $L$, whose meaning will be given below.  There are
two cases (with $P_{\Dset}$ being non-empty, as we are assuming): One is
where the subspace $V_{\perp \Dset}$ in Eq.\ (\ref{eq:V.decomp}) has dimension
$n_{\Dset}=1$ and the other is where it has a higher dimension.

\subsubsection{Case \texorpdfstring{$n_{\Dset}=1$}{n(D)=1}}
\label{sec:edge.1}

First is the case $n_{\Dset}=1$, i.e., that the subspace $V_{\perp \Dset}$
defined in Eq.\ (\ref{eq:V.decomp}) has dimension one. As observed below
that equation, there is a region of $V_{\perp \Dset}$ where all the $D_j$ are
positive.  We choose any vector in this region to be the single edge vector
$e$ for $P_{\Dset}$; no more will be needed.  For every $D_j$, $D_j(e)>0$.

Then every vector $v \in V$ is of the form $v=Ce + k$ for some $k\in K_{\Dset}$
and some real number $C$. Then
\begin{equation}
\label{eq:Dj.nD.1}
  D_j(v) = C D_j(e) + D_j(k) = C D_j(e).
\end{equation}
So the condition that $v \in P_{\Dset}$ is simply that $C>0$.  Then
\begin{equation}
\label{eq:std.decomp.1}
  P_{\Dset} = \left\{ Ce + k : C>0 \mbox{ and } k \in K_{\Dset} \right\}.
\end{equation}

Note that $e$ is non-unique, but only up to a scaling by a positive factor
and the addition of an element of $K_{\Dset}$.  Any single choice of $e$ is
sufficient for our purposes.

From Eq.\ (\ref{eq:Dj.nD.1}) it follows that
$D_j=\frac{D_j(e)}{D_1(e)}D_1$ and hence that all the $D_j$ are
proportional to each other, with positive coefficients.

\subsubsection{Case \texorpdfstring{$n_{\Dset}\geq2$}{nD.ge.2}}
\label{sec:edge.ge.2}

For all the higher co-dimension cases, we will see that $P_{\Dset}$
has non-trivial boundary segments, with lower dimension.  These in
turn have boundary segments, etc.  At each stage of taking boundaries,
one has a strictly lower dimension.

The minimum possible dimension for a non-trivial boundary segment is
$d-n_{\Dset} +1$.  Later, we will prove results about the existence and
properties such next-to-minimal boundary segments.  Here we will simply
provide a definition of corresponding edge vectors, i.e., a vector $e_L$
for each next-to-minimum dimension boundary segment $L$.

Let $L$ be one such boundary segment.  We apply to it the argument of Sec.\
\ref{sec:edge.1} but applied for $L$ with respect to $K(L)$ instead of
$P_{\Dset}$ with respect to $V$, and with the set $\hat{Z}(L)$ instead of
$\Dset$.  We then choose a corresponding vector $e_L$ in the boundary
segment.  A general $v$ in $K(L)$ is $\lambda e_L+k$ where $\lambda$ is real and $k \in
K_{\Dset}$.

We find the conditions for $v$ to be in $L$ as follows: For $D_j \in
Z(L)$, $D_j(e_L)=0$ by the construction of $e_L$, so $D_j(v)=0$.  For
$D_j \in \widehat{Z}(L)$, $D_j(v)=\lambda D_j(e_L)$.  Hence
\begin{equation}
  L = \left\{ \lambda e_L + k : \lambda>0 \mbox{ and } k \in K_{\Dset} \right\}.
\end{equation}

\subsubsection{Overall definition of set of edge vectors}
\label{sec:eL.def}

If $n_{\Dset}\geq2$, we define the set of edge vectors to be all the $e_L$
found in Sec.\ \ref{sec:edge.ge.2} for each boundary segment $L$ that obeys
$n(L)=n_{\Dset}-1$.

If $n_{\Dset}=1$, the set of edge vectors is simply the set consisting of
the one element $e$ constructed in Sec.\ \ref{sec:edge.1}.

The name ``edge vector'' is appropriate when $n_{\Dset} \geq 2$, since each
$e_L$ then corresponds to a projection of boundary segment $L$ onto a line
in $V_{\perp\Dset}$, a projection onto a segment of a line.  But ``edge
vector'' is a bit of a misnomer in the case that $V_{\perp\Dset}$ is
one-dimensional, i.e., $n_{\Dset}=1$.

\subsection{The main decomposition theorem}
\label{ref:decomp.thm}

We are now ready to prove the following theorem:
\begin{theorem}
\label{thm:PD.decomp}
Every element of $v$ of $P_{\Dset}$ can be written in the form
  \begin{equation}
  \label{eq:std.decomp}
     v = \sum_L C_L e_L + v_K,
  \end{equation}
  where all the $C_L$ are positive real numbers, $C_L>0$, and $v_K \in
  K_{\Dset}$, and where the set of $e_L$ is a set of edge vectors, as
  defined in Sec.\ \ref{sec:eL.def}.  Conversely, every $v$ of the form
  Eq.\ (\ref{eq:std.decomp}) with positive $C_L$ is in $P_{\Dset}$.

   Thus $P_{\Dset}$ is exactly the set of vectors of the form
   (\ref{eq:std.decomp}) with the stated restrictions.
\end{theorem}

Before proving the theorem, we make the following comments:
\begin{itemize}
\item The values $C_L$ need not be unique, since it may happen that
  the number of $e_L$s is larger than the dimension $n_{\Dset}$ of
  $V_{\perp \Dset}$. In that case, the $e_L$s are over-complete as a
  spanning set.  If we removed the extra $e_L$s compared with those
  needed to make a basis for $V_{\perp \Dset}$, we could still express $v$
  in the form (\ref{eq:std.decomp}), but some of the coefficients
  might need to be negative for some values of $v$.
\item The edge vectors $e_L$ are not actually in $P_{\Dset}$ except in the
  almost trivial case that $n_{\Dset}=1$.  In other cases, they are always
  on the boundary of $P_{\Dset}$, as we saw.
\end{itemize}

\subsubsection{Examples}
Before treating the general case, we examine examples with effective
dimension one and two, i.e., $n_{\Dset}=1$ and $n_{\Dset}=2$. Then the
derivation of the corresponding specializations of the theorem will be
elementary.  The trick for the general case is to find a way of
successively reducing the dimension of the problem by repeated application
of the two-dimensional version.

In setting up the examples in a fairly general context, it is useful
to recall the following theorem of linear algebra:
\begin{quote}
  Let $\mathcal{E}=(E_1,\dots,E_A)$ be dual vectors on a vector space
  $V$, and let $K_{\mathcal{E}}$ be the intersection of their kernels,
  as defined earlier.  Let $F$ be another dual vector.  Then $F$ is a
  linear combination of $E_1,\dots,E_A$ if and only if the kernel of
  $F$ contains $K_{\mathcal{E}}$, i.e., $\ker F \supseteq K_{\mathcal{E}}$.
\end{quote}

The example of $n_{\Dset}=1$ was already treated in Sec.\ \ref{sec:edge.1}.
Observe that the common kernel $K_{\Dset}$ of the $D_j$ has its maximum
possible dimension $d-1$, and is equal to the kernel of every $D_j$, and
that all the $D_j$ are all proportional to each other (with positive
coefficients so that $P_{\Dset}$ is non-empty).  We constructed an instance
of the single edge vector needed for the problem, and obtained the
decomposition Eq.\ (\ref{eq:std.decomp.1}).  Positivity constraints can be
obtained by examining values of $D_j$ on the space $V_{\perp \Dset}$, and the
results visualized because it is one-dimensional, as in Fig.\
\ref{fig:PD.1}.

\begin{figure}
  \centering
  \includegraphics[scale=0.7]{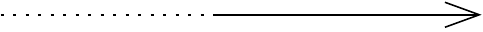}
  \caption{Positivity constraint in $V_{\perp \Dset}$ for the case that it is
    one-dimensional, i.e., $n_{\Dset}=1$.  that $n_{\Dset}=1$.  All the
    $D_j$ are necessarily proportional. The solid line is where $D_j(v)>0$,
    i.e., it is $P_{\Dset}$ projected onto $V_{\perp \Dset}$.  The dotted line
    is where $D_j(v)\leq0$.}
  \label{fig:PD.1}
\end{figure}

In the case $n_{\Dset}=2$, $V_{\perp \Dset}$ is a two-dimensional space
illustrated in Fig.\ \ref{fig:PD.2}.  Each of the $D_j$ has a positive
space delimited by its kernel.  Let us parameterize vectors in $V_{\perp
  \Dset}$ by polar coordinates $(r,\theta)$ with respect to some axes.  Then the
positive region for each $D_j$ is a range $r>0$ with $\theta$ in a continuous
range of size $\pi$.  The kernel of each $D_j$ is a line of fixed $\theta$. The
common positive region is of the form $\alpha<\theta<\beta$, where $0<\beta-\alpha<\pi$.  The most
limiting directions are given by two distinct $D_j$ whose kernels are the
lines of angles $\alpha$ and $\beta$; we use vectors in these directions for the
edge vectors $e_L$, and it is evident that the common positive region
$P_{\Dset}$ is the set of all linear combinations of the two $e_L$ with
positive coefficients.  In polar coordinates, the edge vectors can be
chosen as unit vectors with angles $\alpha$ and $\beta$; they are linearly
independent because $0<\beta-\alpha<\pi$.

\begin{figure}
  \centering
  \includegraphics[scale=0.7]{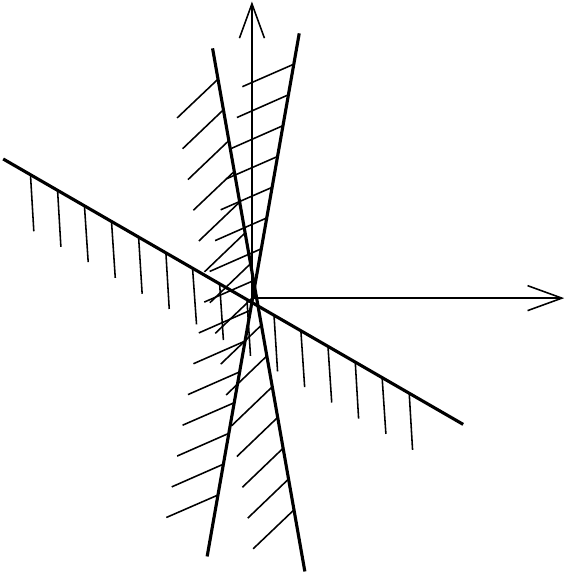}
  \caption{Positivity constraints in $V_{\perp \Dset}$ for a case where it is
    two-dimensional, i.e., $n_{\Dset}=2$.  The diagram depicts the case
    that there are three different $D_j$s involved.  The diagonal lines are
    the locations of the kernels of the $D_j$, and the shaded parts point
    to the negative regions of the $D_j$.}
  \label{fig:PD.2}
\end{figure}

If we had made a mistake in stating the situation, and in fact all the
$D_j$ were proportional to each other, then all the kernels would lie on
top of each other, and we would get the situation shown in Fig.\
\ref{fig:PD.2-1}.  Then the positive range is $\alpha<\theta<\beta$, but now with
$\beta-\alpha=\pi$, so that the would-be $e_L$ vectors from Fig.\ \ref{fig:PD.2}, at
angles $\alpha$ and $\beta$, are exactly opposite to each other, and are therefore
linearly dependent.  These now span one dimension of the kernel space
instead of the positive manifold.  The kernel space has its dimension
increased by one, and correspondingly $V_{\perp \Dset}$ has its dimension
reduced by one. To get an exemplar of the single edge vector that is
needed, we choose a vector pointing in a direction intermediate between
angles $\alpha$ and $\beta$.  To get the results in terms of $V_{\perp \Dset}$, we
simply project onto a one-dimensional space in the direction of the edge
vector, after which we recover a version of Fig.\ \ref{fig:PD.1}.

\begin{figure}
  \centering
  \includegraphics[scale=0.7]{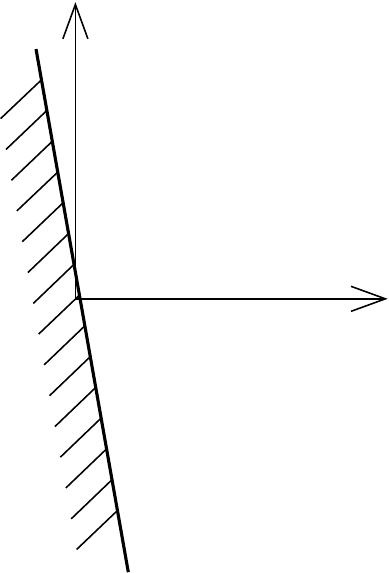}
  \caption{Like Fig.\ \ref{fig:PD.2}, but for the case that all the
    $D_j$ are linearly dependent. The diagram now depicts the space $V_{\perp
      \Dset}$ plus one dimension of the kernel space $K_{\perp \Dset}$.}
  \label{fig:PD.2-1}
\end{figure}

\subsubsection{General case}

If $n_{\Dset}=1$, we already proved the appropriate specialization of
Thm.\ \ref{thm:PD.decomp} in Sec.\ \ref{sec:edge.1}, with $L$ having
one value and the associated $e_L$ being the $e$ of that section.

We now provide a method to deal with all the remaining cases $n_{\Dset}\geq2$
(including the already treated case of $n_{\Dset}=2$).  Necessarily, at
least two of the $D_j$ are linearly independent.  Otherwise all of them
would be proportional to each other (with positive coefficients to allow
$P_{\Dset} \neq \emptyset$), and then the positive space is the positive space for one
$D_j$, so that we get $n_{\Dset}=1$.

Let $v$ be any vector in $P_{\Dset}$.  Then to prove that it is of the
form Eq.\ (\ref{eq:std.decomp}), we adopt the following recursive
strategy
\begin{enumerate}
\item Construct an expression for $v$ as a linear combination of two
  vectors on non-trivial boundary segments.  This we will do quite easily,
  by a simple generalization of the two-dimensional case that was
  illustrated in Fig.\ \ref{fig:PD.2}.
\item For each of these vectors:
  \begin{enumerate}
  \item Either its boundary segment is of the lowest possible dimension for
    a non-trivial boundary segment, i.e., $d-n_{\Dset}+1$, and we can write
    the vector as a positive coefficient times the chosen edge vector for
    the segment, plus a contribution from a vector in the kernel
    $K_{\Dset}$.
  \item Or the boundary segment has a higher dimension, in which case we
    repeat the procedure to express the vector in terms of vectors on
    non-trivial boundary segments of yet lower dimension.
  \end{enumerate}
\item All of this terminates when we get to the lowest dimension
  non-trivial boundary segments.  This gives the desired expansion.
\end{enumerate}

To implement this strategy, given a vector $v$ in $P_{\Dset}$, we first
pick two independent $D_j$ in $\Dset$, and call them $D_a$ and $D_b$.  Then
pick any vector $\delta v$ in the kernel of $D_b$ such that $D_a(\delta v)>0$, and
make it small enough that $v+\delta v$ is still in $P_{\Dset}$.  Then let $w=v+\delta
v$.  The geometry of this situation in the two-dimensional space spanned by
$v$ and $w$ is shown in Fig.\ \ref{fig:P-vw}.  The vectors $v$ and $w$ are
linearly independent, so that they do in fact span a two-dimensional
space.  

\begin{figure}
  \centering
  \includegraphics[scale=0.8]{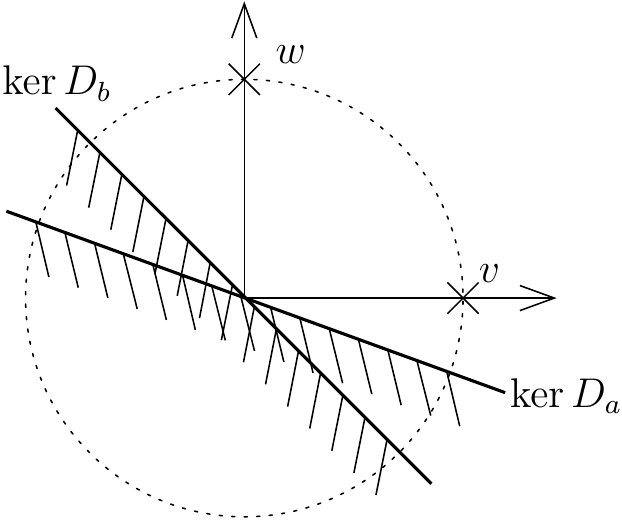}
  \caption{The dotted line is the circle explored to express $v$ in terms
    of boundary vectors, defined to be where the circle first hits the
    kernel of a $D_j$.  Here are seen the intersections of $\ker D_a$ and
    $\ker D_b$ with the two dimensional space spanned by $v$ and $w$.
    \emph{Note that there is not necessarily any metric specified on the
      space $V$, and even if there were there would be no guaranteed
      constraint on the angle between $v$ and $w$.  Nevertheless, it is
      always possible to change the coordinate system by applying a linear
      transformation.  One can do this to go from a situation where $v$ and
      $w$ are in general directions to one where they are drawn at right
      angles, as is the case here.  With this choice of coordinates, the
      loop of vectors in Eq.\ (\ref{eq:loop}) becomes a circle.}}
  \label{fig:P-vw}
\end{figure}

Now $D_j(v)$ and $D_j(w)$ are positive, for all $j$ including $j=a$
and $j=b$, and in addition
\begin{align}
  D_a(w) & = D_a(v) + D_a(\delta v) > D_a(v),
\\
  D_b(w) & = D_b(v) + D_b(\delta v) = D_b(v).
\end{align}
Let $r = D_a(\delta v)/D_a(v) > 0$, so that $D_a(w) = (1+r) D_a(v)$.
Then consider the following loop of vectors in the plane of $v$ and
$w$, parameterized by an angle $\theta$:
\begin{equation}
\label{eq:loop}
  u(\theta) \eqdef v \cos \theta + w \sin \theta,
\end{equation}
on which for a general $D_j$
\begin{equation}
  D_j(u(\theta)) = D_j(v) \cos \theta + D_j(w) \sin \theta.  
\end{equation}
Since both of $D_j(v)$ and $D_j(w)$ are positive, $D_j(u(\theta))$ is
positive in the range $0\leq\theta\leq\pi/2$, and also somewhat beyond
this range. Now define
\begin{equation}
  \theta_j \eqdef \arctan \frac{D_j(v)}{D_j(w)},
\end{equation}
which is in the range $0<\theta_j<\pi/2$.  The zeros of
$D_j(u(\theta))$ are at $\theta=-\theta_j$ and $\theta=\pi-\theta_j$,
so that $D_j(u(\theta))$ is positive when
$-\theta_j<\theta<\pi-\theta_j$.

For the specific cases of $D_a$ and $D_b$
\begin{align}
  D_a(u(\theta)) & = D_a(v) \left[ \cos \theta + (1+r) \sin \theta \right],
\\
  D_b(u(\theta)) & = D_b(v) \left[ \cos \theta + \sin \theta \right],
\end{align}
so that
\begin{equation}
\label{eq:theta.a.b}
  \theta_a = \arctan \frac{1}{1+r} < \frac{\pi}{4},
\qquad
  \theta_b = \frac{\pi}{4}.
\end{equation}

Now define the minimum and maximum values of the $\theta_j$:
\begin{equation}
  \alpha \eqdef \min_j \theta_j,
\qquad
  \beta \eqdef \max_j \theta_j.
\end{equation}
Then for $-\alpha<\theta<\pi-\beta$, all $D_j(u(\theta))$ are
positive, so $u(\theta)\in P_{\Dset}$.  But at each of
$\theta=-\alpha$ and $\theta=\pi-\beta$, at least one $D_j(u(\theta))$
is zero, so that $u(-\alpha)$ and $u(\pi-\beta)$ are on the boundary
of $P_{\Dset}$.  They are in fact on the non-trivial part of the
boundary of $P_{\Dset}$ and are linearly independent.  To see this, we
first observe that from Eq.\ (\ref{eq:theta.a.b}) and from
$0<\theta_j<\pi/2$, it follows that $0<\alpha\leq\theta_a<\pi/4$, while
$\pi/2>\beta\geq\theta_b=\pi/4$.  It follows that $D_b(u(-\alpha))$
and $D_a(u((\pi-\beta))$ are both positive, which puts the vectors
$u(-\alpha)$ and $u(\pi-\beta)$ on the non-trivial part of the
boundary of $P_{\Dset}$, where at least one $D_j$ is positive.
Furthermore, from the same bounds, it follows that $-\alpha$ and
$\pi-\beta$ are not opposite angles, and hence that $u(-\alpha)$ and
$u(\pi-\beta)$ are linearly independent.  See Fig.\ \ref{fig:P-vw-1}
for an illustration of how another $D_j$ can impose a more restrictive
bound on where $u_j(\theta) \in P_{\Dset}$ than is given by $D_a$ and
$D_b$ alone.

Since $u(-\alpha)$ and $u(\pi-\beta)$ are on the non-trivial part of
the boundary of $P_{\Dset}$, we have incidentally proved that for the
case we are treating, $n_{\Dset}\geq2$, there are in fact non-trivial
boundary segments.

\begin{figure}
  \centering
  \includegraphics[scale=0.8]{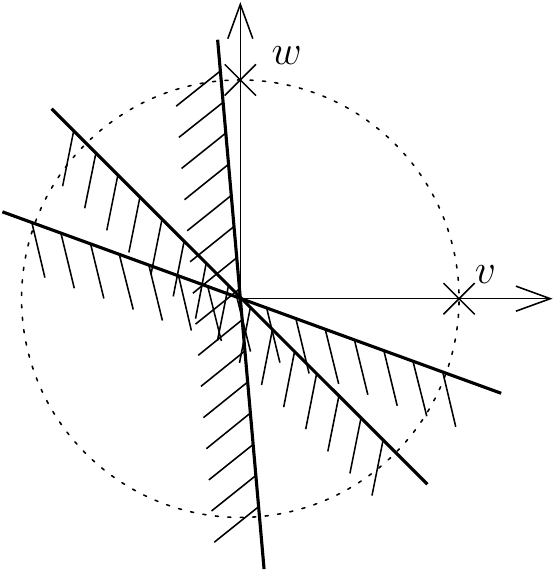}
  \caption{The same as Fig.\ \ref{fig:P-vw}, except that the position of
    the kernel of another $D_j$ is shown, in a situation where it provides
    a more restrictive region of positive $D_j(\theta)$ than is given by $D_a$
    and $D_b$ alone.}
  \label{fig:P-vw-1}
\end{figure}

We can now express $v$ in terms of non-trivial boundary vectors:
\begin{equation}
\label{eq:v.v1.v2}
  v = v_1 \frac{ \sin\beta }{ \sin(\beta-\alpha) }
      + v_2 \frac{ \sin\alpha }{ \sin(\beta-\alpha) }
\end{equation}
where
\begin{equation}
  v_1 = u(-\alpha), 
\qquad
  v_2 = u(\pi-\beta).
\end{equation}
The coefficients in Eq.\ (\ref{eq:v.v1.v2}) are positive, so we have
accomplished our aim of expressing $v$ in terms of vectors on the
non-trivial part of boundary of $P_{\Dset}$ with positive coefficients.
Let the boundary segments in which $v_1$ and $v_2$ lie be $B_1$ and $B_2$

First consider the case that $v_1$'s (non-trivial) boundary segment has the
minimum dimension $d-n(B_1)=d-n_{\Dset}+1$. Then there is an edge vector
for that segment, as defined in Sec.\ \ref{sec:edge.ge.2}, and $v_1$ can be
expressed in terms of the edge vector, with a positive coefficient, plus an
element of the common kernel $K_{\Dset}$.

The other case is that $v_1$'s boundary segment $B_1$ is of higher
dimension.  Then we apply the whole argument of this section to $v_1$, but
now instead of ${\Dset}$ and $V$, we apply the argument to the dual vectors
$\widehat{Z}(B_1)$ that are non-zero at $v_1$ and work in the space
$K(B_1)$.  The argument needs to be extended only by the observation that
all the vectors involved give zero for any $D_j \in Z(B_1)$, i.e., for any
$D_j$ that is zero at $v_1$ and hence on $B_1$.

The result is to express $v_1$ in terms of vectors in
yet lower dimension boundary segments.

The same argument applies equally to $v_2$.

Iterating the argument eventually stops when all the vectors obtained are
proportional to edge vectors (plus elements of $K_{\Dset}$), with positive
coefficients.  Thus any element $v\in P_{\Dset}$ is a linear combination of
edge vectors with positive coefficients, plus a vector in $K_{\Dset}$).

Hence any vector in $P_{\Dset}$ is of the form (\ref{eq:std.decomp}) given
in the statement of Thm.\ \ref{thm:PD.decomp}.  

To complete the proof of Thm.\ \ref{thm:PD.decomp}, we need to show that
any vector of the form (\ref{eq:std.decomp}) is in the positive manifold
$P_{\Dset}$, as opposed to being in its boundary.  So let $v$ be any vector
of the form (\ref{eq:std.decomp}).  For each $D_j$, at least one $D_j(e_L)$
is positive, and so $D_j(v)$ is strictly positive.  Hence $v\in P_{\Dset}$.

\section{The Landau theorem for dual vectors}
\label{sec:landau.proof}

Now we come to the already-stated Thm.\ \ref{thm:main.geom} relating the Landau
condition to the non-existence of good directions, i.e., to the
non-existence of a $v$ for which all of $D_j(v)$ are positive.  To prove
the theorem, we consider the cases that there is and that there is not a
good direction.

We already saw in Sec.\ \ref{sec:elem.parts} that if there is a good
direction, then there can be no Landau point.  It remains to show that if
there is no good direction, then a Landau point exists.  Given that there
fails to be a good direction for $(D_1,\dots,D_N)$, we will construct a set
of $\lambda_j$s that instantiates a Landau point.

It might be that one or more of the $D_j$s is zero.  In that case, let
$D_{j_0}$ be one of the zero dual vectors.  Then set $\lambda_{j_0}=1$
and set the remaining $\lambda_j$ to zero, and we have a Landau point.

So we only need further to consider the case that every $D_j$ is
non-zero.

Consider the following subsets of $D_j$s, where we start with $D_1$,
and successively add an extra $D_j$: $S_1=(D_1)$, $S_2=(D_1,D_2)$,
\dots, $S_N=(D_1,\dots,D_N)$.  Since $D_1\neq 0$, we can find a vector
$v\in V$ with $D_1(v)=1$, and so there exists a good direction
for $S_1$.  But by hypothesis there is no good direction for
$S_N$.  Therefore there is a last one in this sequence, $S_{n_0}$, for
which there is a good direction; for the next set, $S_{n_0+1}$,
there is no good direction.

In the following, two different vector spaces come into play.  One is the
space $V$ on which the $D_j$s act, with an important role played by its
submanifold where all the $D_j$s are positive.  The other space is a space
$\Lset$ of the coefficients $\3\lambda$ used in linear combinations of the form
$\sum_{j=1}^{n_0} \lambda_j D_j + D_{n_0+1}$, with its definition in Eq.\
(\ref{eq:Lset.def}) below.

\subsection{The positive hyperplane \texorpdfstring{$P$}{P}}

Let $P$ be the set of good directions for $S_{n_0}$, i.e., $P$ is
the positive space for the corresponding $D_j$s:
\begin{equation}
  P = P_{S_{n_0}}
    = \left\{
        v \in V :
        D_j(v)>0 \mbox{ whenever $1\leq j \leq n_0$}
      \right\}.
\end{equation}
Then the lack of a good direction for $S_{n_0+1}$ immediately
shows that $D_{n_0+1}(v) \leq 0$ for every $v \in P$.  In fact, strict
inequality holds:
\begin{lemma}
  \begin{equation}
    \label{eq:n0+1.P}
    D_{n_0+1}(v) < 0 \mbox{ for every $v \in P$}.
  \end{equation}
\end{lemma}
\begin{proof}
  We use the fact, following from Thm.\ \ref{thm:P.basic}, that $P$ is an
  open set.  Suppose that the strict inequality did not hold.  Then there
  would be a $v \in P$ for which $D_{n_0+1}(v) = 0$.  Since $D_{n_0+1}$ is
  non-zero, we can find a $w \in V$ for which $D_{n_0+1}(w)=1$.  Then for
  every $\kappa > 0$, $D_{n_0+1}(v+\kappa w) = \kappa >0$.  Since $P$ is an open set, $v+\kappa
  w \in P$ for all small enough $\kappa$, and we would therefore find a vector in
  $P$ on which $D_{n_0+1}$ is positive.  That is, we would find a good
  direction for the set $S_{n_0+1}$.  This is contrary to hypothesis, so we
  need the strict inequality (\ref{eq:n0+1.P}).
\end{proof}

\subsection{Kernels of \texorpdfstring{$D_j$ $(0\leq j \leq
    n_0+1)$}{Dj}}
\label{sec:ker.Dj}

Next, let $K$ be the intersection of the kernels of $D_1$, \dots,
$D_{n_0}$:
\begin{equation}
\label{eq:K.def}
  K \eqdef \left\{ v \in V : D_j(v) = 0
              \mbox{ for $1 \leq j \leq n_0$}
      \right\}.
\end{equation}
It is a vector subspace of $V$.  

Now the other $D_j$ we consider, i.e., $D_{n_0+1}$, is also zero on
$K$.  To see this, suppose otherwise, and we will prove a
contradiction.  Thus, suppose that there is a $u \in K$ such that
$D_{n_0+1}(u) \neq 0$.  By scaling $u$, we can arrange
$D_{n_0+1}(u)=1$, while maintaining $D_j(u) = 0$ for the other $D_j$.
Pick any $v \in P$, so that $D_j(v)>0$ for every $1 \leq j \leq n_0$.
Then for every positive real number $\kappa>0$
\begin{equation}
  D_{n_0+1}(\kappa u + v) = \kappa + D_{n_0+1}(v),
\end{equation}
while
\begin{multline}
  D_j(\kappa u + v) = \kappa D_j(u) + D_j(v) = D_j(v) > 0
\\
  \mbox{for $1 \leq j \leq n_0$}.
\end{multline}
It follows that $\kappa u+v$ is also in $P$.  But by choosing $\kappa$
large enough, we can make $D_{n_0+1}(\kappa u + v)$ positive, which
would give us a good direction for the set $S_{n_0+1}$.  We only
avoid this by having $D_{n_0+1}(u)=0$ for every element $u\in K$.

Thus the kernel of $D_{n_0+1}$ contains the intersection of the kernel of
the other $D_j$s.  It follows, by a standard theorem of linear algebra,
that $D_{n_0+1}$ is a linear combination of the other $D_j$s.  But we do
not need to use this.  In fact, we will prove a stronger result that a
linear combination can be found where all the coefficients are negative or
zero.

\subsection{Spanning vectors of \texorpdfstring{$P$}{P}}

We now recall results from Sec.\ \ref{sec:pos.sets}, but applied with
$\Dset$ set equal to $S_{n_0}=(D_1,\dots,D_{n_0})$ instead of the original
set of dual vectors.  The space $V$ can be decomposed as a direct sum $V=K\oplus
V_\perp$.  Then there is a set of non-zero edge vectors $e_L$ that give
one-dimensional edges for $P\cap V_\perp$, and the general form for a vector $v \in
P$ is
\begin{equation}
  \label{eq:P.decomp}
  v = k + \sum_L C_L e_L,
\end{equation}
where $k \in K$ and all the real-valued coefficients $C_L$ are
strictly positive: $C_L>0$.  The vectors $e_L$ span $V_\perp$, but
they could be an over-complete set; the extra elements are needed to
maintain the positivity property on the $C_L$s for every $v \in P$.

All the edge vectors obey $D_je_L \geq 0$ for $1 \leq j \leq n_0$ and any $L$.  For
every $j$ in the range $1\leq j \leq n_0$, there is at least one value of $L$ for
which $D_je_L$ is strictly positive.  Similarly for every $L$ there is at
least one value of $j$ in the range $1\leq j \leq n_0$ for which $D_je_L$ is
strictly positive.

From the properties that $D_{n_0+1}(v) < 0$ for every $v \in P$ and that
$D_{n_0+1}(v) = 0$ for every $v \in K$, it follows that
\begin{lemma}
  \begin{equation}
    \mbox{For all $L$}, D_{n_0+1}(e_L) \leq 0,
  \end{equation}
  and at least one $D_{n_0+1}(e_L)$ is strictly negative.
\end{lemma}

\subsection{Linear combinations of \texorpdfstring{$D_j$s}{Dj}; the
  regions \texorpdfstring{$\Lset$}{Lambda}, \texorpdfstring{$M$}{M},
  and \texorpdfstring{$\notM$}{M-tilde}}
\label{ref:sec.M}

Our aim is to find a set of $\lambda_j$ for which $\sum_{j=1}^{n_0+1}\lambda_jD_j=0$,
with all $\lambda_j\geq0$, and with at least one non-zero (positive) $\lambda_j$.  To
obtain this, it is necessary that the last $\lambda$ is non-zero, i.e.,
$\lambda_{n_0+1}>0$.  This is because if it were zero, we would have the
Landau point for $S_{n_0}=(D_1,\dots,D_{n_0})$, i.e., we would have
$\sum_{j=1}^{n_0}\lambda_jD_j=0$ (with the sum up to $j=n_0$).  But the
definition of $n_0$ is that there is a good direction for $S_{n_0}$,
which implies that there is no Landau point for $S_{n_0}$.

Therefore, to avoid a contradiction, any Landau point for $S_{n_0+1}$ must
have $\lambda_{n_0+1}$ strictly greater than zero.  We can now scale all the
$\lambda_j$'s to make $\lambda_{n_0+1}=1$, and still have a Landau point.  So we will
work with
\begin{equation}
  D(\3\lambda) \eqdef \sum_{j=1}^{n_0} \lambda_j D_j + D_{n_0+1},
\end{equation}
where we use boldface notation $\3\lambda = (\lambda_1,\dots,\lambda_{n_0})$ to denote a
vector of only the first $n_0$ values, and we simply require allowed
values to obey $\lambda_j\geq0$.  Then our aim is to find an allowed $\3\lambda$ for
which $D(\3\lambda)=0$.

Define $\Lset$ to be the set of allowed $\3\lambda$:
\begin{equation}
\label{eq:Lset.def}
  \Lset \eqdef \left\{ \3\lambda : \lambda_j \geq 0 \mbox{ for all $j$} \right\},
\end{equation}
and define the following subset of $\Lset$:
\begin{equation}
\label{eq:M.def}
  M \eqdef \left\{ \3\lambda \in \Lset : D(\3\lambda)(v) \leq 0 
     \mbox{ for all $v \in P$} \right\},
\end{equation}
i.e., $M$ is the set of all $\3\lambda$ in $\Lset$ for which
$D(\3\lambda)$ is negative or zero for every vector that makes all of
$D_1$, \dots, $D_{n_0}$ positive.  Its complement in $\Lset$ is the
set $\3\lambda$ for which we have a positive value for $D(\3\lambda)$
somewhere in $P$:
\begin{align}
  \notM
  \eqdef \Lset \backslash M
  = {}& \bigl\{ \3\lambda \in \Lset : 
       \exists v \in V 
       \mbox{ such that }
       D(\3\lambda)(v) > 0,
\nonumber\\ & ~
       \mbox{ and }
       D_j(v) > 0 \mbox{ for $1 \leq j \leq n_0$}
    \bigl\}.
\end{align}
We will find a Landau point at a certain corner or edge of the set
$M$.

To visualize the kind of set that $M$ is, it is useful to refer to the
simple example given in App.\ \ref{sec:M.simple} below.  It results in
a region for $M$ that is illustrated in Fig.\ \ref{fig:ex.M}.  Notice
that $M$ is convex and is a closed set.  The boundaries are segments
of straight lines.  The value of $\3\lambda$ giving a Landau point is
at the upper right-hand corner.

\begin{figure}
  \centering
  \includegraphics[scale=0.9]{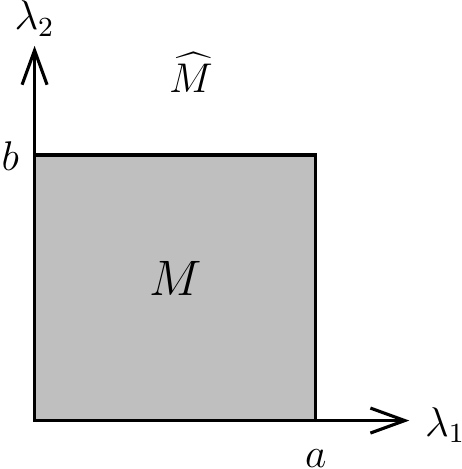}
  \caption{Set $M$ for the example given in App.\ \ref{sec:M.simple}.}
  \label{fig:ex.M}
\end{figure}

\subsection{Properties of \texorpdfstring{$M$}{M}}

We first derive some elementary properties of $M$ and $\notM$ for
the general case:
\begin{enumerate}

\item The zero vector $\3{0}$ is in $M$, so that $M$ is non-empty.
  This is simply because $D(\3{0})=D_{n_0+1}$, and $D_{n_0+1}(v)$ is
  negative for all vectors in $P$, Eq.\ (\ref{eq:n0+1.P}).

\item $M$ is convex.  Suppose that $\3\lambda_a$ and $\3\lambda_b$ are
  any 2 elements of $M$ and that $t$ is any real number obeying $0 \leq
  t \leq 1$.  Then for every $v \in P$
  \begin{multline}
     D\xleft( t\3\lambda_a + (1-t)\3\lambda_b \right)(v)
  \\
     = t D(\3\lambda_a)(v)
       + (1-t) D(\3\lambda_b)(v),
   \end{multline}
   which is zero or negative because each term is.  It follows that
   $t\3\lambda_a + (1-t)\3\lambda_b \in M$.  Hence $M$ is convex.

 \item $M$ is a closed set.  Let $\3\lambda_\alpha$ be any sequence of
   elements of $M$ that converges to some element $\3\lambda$ of
   $\Lset$.  To show that $M$ is closed, we need to show that the
   limit point $\3\lambda$ is actually in $M$.  To do this, we observe
   that for every $v \in P$, all the $\3\lambda_\alpha$ obey
   $D(\3\lambda_\alpha)(v) \leq 0$, by the definition of $M$.  Hence
   \begin{equation}
     D(\3\lambda)(v) = \lim_{\alpha\to\infty} D(\3\lambda_\alpha)(v) \leq 0,
   \end{equation}
   by the continuity of linear functions.  Hence $\3\lambda\in M$.

 \item Consider an arbitrary non-zero $\3\lambda \in \Lset$, and consider
   an arbitrarily scaled value $\kappa \3\lambda$, where $\kappa$ is a
   positive real number. Then for large enough $\kappa$, $D(\kappa
   \3\lambda) \in \notM$, but not $M$

   \emph{Proof}: For any $v \in P$
   \begin{equation}
   \label{eq:D.kappa.v}
     D(\kappa \3\lambda)(v)
     = \kappa \3\lambda \cdot \3D(v) + D_{n_0+1}(v).
   \end{equation}
   Since $v \in P$, at least one $\lambda_j>0$, and the others are non-negative,
   $\3\lambda \cdot \3D(v)$ is positive, so for large enough $\kappa$, the quantity in
   (\ref{eq:D.kappa.v}) is positive, and hence $D(\kappa \3\lambda) \in \notM$, but not
   $M$.  The line of $\kappa \3\lambda$ intersects the boundary of $M$ at some point,
   which may in degenerate situations be at $\3{0}$.

\end{enumerate}
It is easily checked that these properties are obeyed in the example shown
in Fig.\ \ref{fig:ex.M}.

\subsection{Characterization of \texorpdfstring{$M$}{M} in terms of
  properties of edge vectors \texorpdfstring{$e_L$}{eL}}

We have seen in Eq.\ (\ref{eq:P.decomp}) that any vector in $P$ can be
written as a sum of edge vectors with strictly positive coefficients plus
an element of the kernel $K$, i.e., $v=k+\sum_LC_Le_L$.  Since $D_j(k)=0$ for
$1 \leq j \leq n_0+1$, it follows that
\begin{equation}
  D(\3\lambda)(v)= \sum_L C_L f(L,\3\lambda),
\end{equation}
where
\begin{equation}
  f(L,\3{\lambda}) \eqdef \sum_{j=1}^{n_0} \lambda_jD_j(e_L)  + D_{n_0+1}(e_L).
\end{equation}

Therefore for any given $\3\lambda\in \Lset$, we can characterize
whether $\3\lambda$ is in $M$ or $\notM$, and whether it is in the
boundary of $\notM$ by the following exclusive criteria:
\begin{enumerate}

\item \emph{Either at least one $f(L,\3{\lambda})$ is strictly positive.} In this
  case $\3\lambda \in \notM$.

  The last statement is proved by letting $L_0$ be one of the cases for
  which $f(L_0,\3{\lambda})>0$, and we set $C_L = \delta_{L,L_0} + \kappa$, with $\kappa>0$.
  Then $v =\sum_L C_Le_L \in P$ and
  \begin{equation}
    D(\3\lambda)(v)
    = f(L_0,\3{\lambda}) + \kappa \sum_L f(L,\3{\lambda}).
  \end{equation}
  By making $\kappa$ small enough (but non-zero), we can make this
  positive.  Hence $\3\lambda \in \notM$.

\item \emph{Or all of $f(L,\3{\lambda})$ are strictly negative.}  Then $\3\lambda \in M$
  and $\3\lambda$ is in the interior of $M$, not on its boundary with $\notM$.

  First, we observe that the negativity of $f(L,\3{\lambda})$ implies
  that $D(\3\lambda)(v)$ is negative for all $v \in P$, so that
  $\3\lambda \in M$.  Then we consider a nearby point $\3\lambda_a =
  \3\lambda + \delta\3\lambda$ that is still in $\Lset$ (i.e., the
  components obey $\lambda_{a,j}\geq 0$), and we let a general element
  of $P$ be $v = k + \sum_L C_L e_L$.  We let
  \begin{align}
    -F &= \max_{L} f(L,\3\lambda) < 0,
  \\
    G &= \max_{j,L} D_j(e_L) > 0.
  \end{align}
  Then we can bound
  \begin{align}
    D(\3\lambda+\delta\3\lambda)(v)
    & = \sum_L C_L \biggl[ f(L,\3\lambda)
                         + \sum_{j=1}^{n_0} \delta\lambda_j D_j(e_L)
                   \biggr]
  \nonumber\\
    & \leq \sum_L C_L \biggl[ -F  + \sum_{j=1}^{n_0} |\delta\lambda_j| G \biggr].
  \end{align}
  Now take $\sum_{j=1}^{n_0} |\delta\lambda_j| < F/G$.  Then
  $D(\3\lambda+\delta\3\lambda)$ is negative on the whole of $P$, and
  so $\3\lambda+\delta\3\lambda$ is in $M$.  Hence all points
  sufficient close to $\3\lambda$ are themselves in $M$, and not in
  $\notM$.  Therefore $\3\lambda$ is not on the boundary with $\notM$.

\item \emph{Or for all $L$, $f(L,\3{\lambda}) \leq 0$, and at least one is zero.}
  Then $\3\lambda \in M$ and $\3\lambda$ is on its boundary with $\notM$.

  Given that none of $f(L,\3{\lambda})$ is positive, $\3\lambda$ must be in $M$,
  not $\notM$.  It remains to show that it is on the boundary.

  So pick $L_0$ such that $f(L_0,\3{\lambda}) = 0$, and pick $\delta\3\lambda$ such that all
  the $\delta\lambda_j$ are strictly positive, $\delta\lambda_j>0$, for $1 \leq j \leq n_0$.  We will
  show that $\3\lambda+\delta\3\lambda$ is in $\notM$ no matter how small $\delta\3\lambda$ is.  First,
  $\3\lambda+\delta\3\lambda$ is in $\Lset$, because each component of the vector is
  non-negative.  Let
  \begin{equation}
     -A = \min_{L} f(L,\3\lambda) \leq 0,
  \end{equation}
  and choose an element of $P$ by $v = \kappa e_{L_0} + \sum_{L}e_L$,
  with $\kappa>0$.  Then
  \begin{equation}
  \label{eq:A.v}
    D(\3\lambda+\delta\3\lambda)(v)
     \geq - A \#(L)
       + \kappa \sum_{j} \delta\lambda_jD_j(e_{L_0}),
  \end{equation}
  with $\#(L)$ being the number of $e_L$ vectors.  Hence, by choosing $\kappa$
  large enough, we make $D(\3\lambda+\delta\3\lambda)(v)$ positive.  Therefore $D(\3\lambda+\delta\3\lambda)$
  is in $\widehat{M}$ no matter how small the non-zero $\delta\3\lambda$ is, and so
  $\3\lambda$ is on the boundary between $M$ and $\widehat{M}$. 
\end{enumerate}

\subsection{Moving along boundary between \texorpdfstring{$M$}{M}
  and \texorpdfstring{$\notM$}{M-tilde}}
\label{sec:notM.M}

Now consider a point $\3{\lambda}_0$ on the boundary between $M$ and $\notM$.
Such a point exists.  At it, $f(L,\3{\lambda}_0)$ is zero for some number of
edges $L$ of the positive region $P$, and negative for any others.  We will
now show that we can move from $\3{\lambda}_0$ along the boundary of $M$ in such
a way that we find a place where there is an increase in the number of $L$
for which $f(L,\3{\lambda})=0$.  We keep going, repeating this process, which
terminates only when all are zero.  At that point $D(\3{\lambda})$ is itself zero
and we have a Landau point, as we aimed to find.

Let\footnote{Note that $Z$ is now used with a different meaning and type of
  argument than before.} $Z(\3{\lambda}_0)$ and $\widehat{Z}(\3{\lambda}_0)$ be the set
of $L$ for which $f(L,\3{\lambda}_0)$ is zero and non-zero (necessarily
negative):
\begin{align}
  Z(\3{\lambda}_0) & = 
   \left\{ L : f(L,\3{\lambda}_0) = 0 \right\},
\\
  \widehat{Z}(\3{\lambda}_0) & = 
   \left\{ L : f(L,\3{\lambda}_0) < 0 \right\}.
\end{align}
This is a partition of the set of all edges of $P$.

Let $v$ be a general element of $P$, decomposed as in Eq.\
(\ref{eq:P.decomp}). Then
\begin{equation}
\label{eq:D.lambda.0}
  D(\3{\lambda}_0)(v)
  =
  \sum_{\text{all }L} C_L f(L,\3{\lambda}_0)
  =
  \sum_{L \in \widehat{Z}(\3{\lambda}_0)} C_L f(L,\3{\lambda}_0).
\end{equation}

It is possible that all the $F(L,\3{\lambda})$ are zero, so that $D(\3\lambda_0)(v)=0$
for every $v\in P$.  Since $P$ is a manifold of the same dimension as the
whole space $V$, it follows that $D(\3\lambda_0)$ itself is zero, so that we have
a Landau point, and we need go no further.

Otherwise at least one $F(L,\3{\lambda})$ is nonzero and negative.  To deal with
this case, our method will be to first prove that the set of $e_L$ with $L
\in Z(\3{\lambda}_0)$ spans a boundary segment of $P$, rather than the whole of
$P$, and then that there is a $j_0$ for which
\begin{equation}
\label{eq:on.bound.M}
  D_{j_0}(e_L) = 0 \mbox{ for all } L \in Z(\3{\lambda}_0). 
\end{equation}
We will use this to provide a direction $\delta\3\lambda$ in which to
move while staying on the boundary of $M$ and then eventually find a
point where yet another $f(L,\3{\lambda})$ is zero.

The following simple results are useful in the sequel:
\begin{lemma}
\label{lem:neg.D.P}
  $D(\3\lambda_0)(v)$ is negative for every $v$ in $P$.
\end{lemma}
\begin{proof}
  The $F(L,\lambda_0)$ in (\ref{eq:D.lambda.0}) are negative or zero, and at
  least is nonzero.  Hence $D(\3\lambda_0)(v)$ is negative for every $v$ in $P$,
  since all the $C_L$ are strictly positive.
\end{proof}

\begin{lemma}
\label{lem:null.D.P}
  There is a null intersection between the kernel of 
  $D(\3\lambda_0)$ and the positive region $P$:
  \begin{equation}
  \label{eq:K.D.intersect}
    \ker D(\3\lambda_0) \cap P = \emptyset.
  \end{equation}
\end{lemma}
\begin{proof}
  This follows from Lemma \ref{lem:neg.D.P}, since $D(\3\lambda_0)$ is nonzero on
  the whole of $P$.
\end{proof}

\begin{lemma}
\label{lem:P.and.bound}
   For every $v$ in both $P$ and its boundary, $D(\lambda_0)(v) \leq 0$.
\end{lemma}
\begin{proof}
  We know that $D(\3\lambda_0)$ is negative for every element of $P$.  A boundary
  point is obtained by taking a limit of points in $P$.  Therefore
  $D(\3\lambda_0)$ is negative or zero on the boundary of $P$.
\end{proof}

\begin{lemma}
  For every $k \in K$ and any $\3\lambda$,
  \begin{equation}
  \label{eq:D.lambda.K}
    D(\3\lambda)(k) = 0.
  \end{equation}
\end{lemma}
\begin{proof}
  Since $k\in K$, every $D_j(k)=0$, for $1 \leq j \leq n_0$.  We have
  also seen in Sec.\ \ref{sec:ker.Dj} that $\ker D_{n_0+1} \supseteq
  K$, so $D_{n_0+1}(k)=0$. Equation (\ref{eq:D.lambda.K}) follows.
\end{proof}

Now consider vectors of the form
\begin{equation}
\label{eq:v.Z}
  w = k + \sum_{L \in Z(\3{\lambda}_0)} C_L e_L,
\end{equation}
where $C_L>0$, $k \in K$, and we have only used the subset of $e_L$ for which
$f(L,\3{\lambda}_0)=0$.  The vector $w$ is in the kernel of $D(\3\lambda_0)$, i.e.,
$D(\3\lambda_0)(w)=0$.  So by Lemma \ref{lem:null.D.P} it cannot be in $P$.  But
by adding a term $\kappa \sum_{L \in \hat{Z}(\3{\lambda}_0)} e_L$, with $\kappa$ non-zero and
positive, but arbitrarily small, we get a vector in $P$ itself.  Hence the
given $w$ is in a non-trivial boundary segment $B$ of $P$.

Now from Thm.\ \ref{thm:bdy.char}, applied to $P$ instead of $P_{\Dset}$,
we know that on the boundary segment $B$, there is a value $j_0$ for which
$D_{j_0}$ is zero on $B$, and hence $D_{j_0}(w)=0$.
 
Now, in Eq.\ (\ref{eq:v.Z}), for all $j$ in the range $1\leq j \leq n_0$,
$D_j(k)=0$ and for all $L$ $D_j(e_L)\geq0$.  Hence from the zero value of
$D_{j_0}(w)$ it follows that $D_{j_0}(e_L)$ is zero for all $L \in
Z(\3{\lambda}_0)$, i.e., for all those $L$ for which $D(\3\lambda_0)(e_L)$ is zero
rather than negative.  (At least one such $L$ exists, since $\3\lambda_0$ is on
the boundary of $M$.)

\medskip

Now let us make an increment $\delta \3\lambda$ to $\3\lambda_0$ defined by
\begin{equation}
  \delta\lambda_j = \delta_{jj_0}\kappa,
\end{equation}
with $\kappa \geq 0$.  All its components are non-negative, so 
$\3\lambda_0 +\delta \3\lambda$ remains in $\Lset$.  To determine
its location with regards to $M$ and $\notM$, we calculate
\begin{equation}
  f(L,\3{\lambda}_0 + \3\delta\lambda)
 = 
  f(L,\3{\lambda}_0) + \kappa D_{j_0}(e_L).
\end{equation}
When $L\in Z(\3\lambda_0)$, this is zero.  When, instead, $L\in \widehat{Z}(\3\lambda_0)$,
this starts out negative and either stays at the same value or increases
with $\kappa$, depending on whether $D_{j_0}(e_L)$ is zero or not.  (Recall that
every $D_j(e_L)$ is positive or zero.)

Thus for small enough $\kappa$, $\3{\lambda}_0 + \3\delta\lambda$ is still on the boundary of
$M$.  But given $j_0$, there is at least one $e_L$ for which $D_{j_0}(e_L)$
is strictly positive; this $e_L$ is necessarily one of those corresponding
to $\widehat{Z}(\3\lambda_0)$.  So at least one of the initially negative
$f(L,\3{\lambda}_0 + \3\delta\lambda)$ increases.  There is a least $\kappa$ for which one (or
more) of these reaches zero.  Let $\3\lambda_1$ be the resulting position on the
boundary of $M$.

At this point, we go back to the start of this Sec.\ \ref{sec:notM.M}, and
replace $\3\lambda_0$ by $\3\lambda_1$.  We keep iterating this procedure, getting a
sequence of boundary points $\3\lambda_i$, with at each stage getting an
increased number of $L$ for which $f(L,\3{\lambda}_i)=0$.  
Eventually this procedure has to stop because we run out of
values of $L$, and the only way this happens in the argument is that there
are no values of $L$ for which $F(L,\3{\lambda}_{i_{\rm max}})$ is negative,
i.e., that all the $F(L,\3{\lambda}_{i_{\rm max}})$ are zero, it follows that
$D(\3\lambda_{i_{\rm max}})=0$, i.e., we have a Landau point.

This completes our proof of Thm.\ \ref{thm:main.geom}.

\section{Contour deformations and pinches}
\label{sec:deal.with.zero.first.order}

We now return to the determination of the conditions for the existence of a
pinch at a particular value $w_S$ of the integration variable in an
integral of the form given in Eq.\ (\ref{eq:integral}).  We will complete
the proof of Thm.\ \ref{thm:main.contour}, that there is a pinch if and
only if the corresponding Landau condition holds.  To do this, we need to
prove Thm.\ \ref{thm:pinch.to.1st.order.shift} relating a pinch to the
non-existence of an allowed deformation with positive first-order shifts in
the imaginary parts of the relevant denominators. Once that is proved, the
already-proved geometric Thm.\ \ref{thm:main.geom} gives Thm.\
\ref{thm:main.contour} as an immediate consequence.

We saw in Sec.\ \ref{sec:elem.parts} that if an allowed deformation exists
with positive first-order shifts in the relevant denominators, then the
integration is not trapped.  So it remains to show that if there is no such
deformation, then the integral is trapped.  Now in the one-dimensional
case, this is easy to show, because a contour deformation avoids a
singularity due to a zero of a denominator if and only if the first order
shift in the denominator is positive --- see App.\ \ref{sec:1D.first.order}
for explicit details. But in higher dimensions, the example in App.\
\ref{sec:2D.first.order} shows that a singularity can be avoided while
having a first-order shift that is zero, i.e., the direction of contour
deformation can be tangent to the singularity surface.  Hence a more
detailed argument is needed for the general case, and it will turn out to
be annoyingly difficult for such an apparently elementary result.

\subsection{Elementary results}

First we prove some elementary results that strongly restrict the kinds of
contour deformation that do or do not avoid singularities.

Given a denominator $A_j(w_R)$ that is zero at $w_R=w_S$, define its
derivative by $D_j = \partial A_j(w_S)$.  Then consider a candidate contour
deformation specified by $v(w_R)$, and let $v_S=v(w_S)$, the direction of
deformation at $w_S$.  We classify what happens by the sign of $D_j(v_S) =
v_S\cdot \partial A_j$, and codify the results in some theorems.

It is useful to define $\delta w=w_R-w_S$ and $\delta v(\delta w) = v(w_S+\delta w)-v_S$, i.e., the
deviations from values at $w_S$.

First, for positive $D_j(v_S)$:
\begin{theorem}
  \label{thm:v.D.pos}
  Suppose, with the notation and conditions just stated, that the
  deformation is allowed and that $D_j(v_S)>0$. Then the deformation avoids
  the singularity at $w_S$ associated with the zero of $A_j$.
\end{theorem}
\begin{proof}
  Although this result is elementary,we will give a detailed argument,
  since this will introduce techniques to be used in more difficult
  situations.

  The denominator is $A_j(w_S+\delta w+i\lambda (v_S + \delta v))+i\epsilon$.  Now we always
  require that $A_j(w)$ is analytic and that it is real when $w$ is real.
  Therefore all the Taylor coefficients for an expansion about $w_S$ are
  real.  We expand the denominator $A_j+i\epsilon$ in powers of $\delta w$ and $\lambda$.

  The imaginary part comes solely from the odd terms in $\lambda$, and hence
  \begin{equation}
    \label{eq:im.aj}
    \Im(A_j+i\epsilon) = \epsilon + \lambda \left[ D_j(v_S) + O(\lambda^2) + O(\delta w) \right].
  \end{equation}
  For small enough $\delta w$ and $\lambda$, the correction terms are smaller in size
  than the positive $D_j(v_S)>0$ term, and we therefore have a sum of two
  non-negative terms.  Therefore the imaginary part of $A_j$ is zero only
  if both $\epsilon$ and $\lambda$ are zero.

  Now a zero of the denominator occurs when both its real and imaginary
  parts are zero. Hence in a neighborhood of $w_S$, the denominator is
  nonzero when $\lambda$ is small and positive, and therefore the singularity due
  to $A_j(w_S)=0$ is avoided by the contour deformation.
\end{proof}

Next, if  $D_j(v_S)$ is negative, the candidate deformation is not even
allowed:
\begin{theorem}
  \label{thm:v.D.neg}
  Suppose that $v(w_R)$ specifies a candidate deformation, and that
  $D_j(v_s) < 0$.  Then the deformation is not allowed.
\end{theorem}
\begin{proof}
  We need to show that if $D_j(v_S)$ is negative then we have a situation
  like that shown in Fig.\ \ref{fig:eps.lam.sing}(a). That is, for all
  small $\lambda$, there is a zero of $A_j\!\bigl(w_S+\delta w+i\lambda (v_S+\delta v(\delta
  w))\bigr)+i\epsilon$ for some small positive $\epsilon$ and some value(s) of $\delta w$.
  Furthermore (at least one of) these values of $\epsilon$ and $\delta w$ approach zero
  as $\lambda\to0+$.  This corresponds to the negation of the definition of an
  allowed deformation given as given by Defns.\
  \ref{def:compat.deform.point} and \ref{def:allowed.deform}.

  We start with $\epsilon$ slightly positive, increase $\lambda$ from zero, and then
  decrease $\epsilon$ to zero.  We encounter a situation where the imaginary part
  of $A_j+i\epsilon$ is zero.  This can be seen from Eq.\ (\ref{eq:im.aj}) given
  that $D_j(v_S)$ is negative.  The zero of the imaginary part occurs both
  when $\delta w=0$ and for all nearby values of $\delta w$, and it occurs for all
  small positive $\lambda$.

  But a zero in the imaginary part of the denominator does not itself show
  that the deformation encounters a singularity from a zero in $A_j+i\epsilon$,
  because to get a zero we also need the real part to be zero, and we need
  to show that such a zero occurs independently of any higher order terms
  in the Taylor expansion of $A_j$ in powers of small quantities $\lambda$, $\delta w$
  and $\delta v$.

  Choose $w_R=w_S+xv_S$, i.e., $\delta w=xv_S$.  Thus $x$ parameterizes a
  particular line in the space of $w_R$.  Then $\delta v=O(x)$ as $x\to0$.  Hence
  applying a Taylor expansion of $A_j(w)$ about $w=w_S$ gives
  \begin{multline}
    \label{eq:case.linear}
    A_j(w_S+xv_S + i\lambda(v_S+\delta v)) +i\epsilon
  \\
    = i\epsilon + (x+i\lambda) D_j(v_S) + O\xleft( |x|^2, \lambda^2, |x|\lambda \right) .
  \end{multline}
  Recall that $v_S\cdot D_j$ is real, and is negative in the situation that we
  are currently considering.  Define a complex variable
  \begin{equation}
    \zeta = x+i\lambda,
  \end{equation}
  and consider values $\zeta=re^{i\theta}$ for positive $r$ and for $0\leq\theta\leq\pi$, so that
  both $x$ and $\lambda$ are at most of size $r$, with $x=r\cos\theta$ and $\lambda=r\sin\theta$.
  As we increase $\theta$ from 0 to $\pi$, $(x+i\lambda) v_S \cdot D_j$ traces out the
  semicircle in the lower half plane shown in Fig.\ \ref{fig:tour1}. It
  necessarily crosses the negative imaginary axis. The value of $A_j$
  differs from $(x+i\lambda) v_S \cdot D_j$ by terms of order $r^2$, so for small
  enough $r$, they only slightly modify the path in Fig.\ \ref{fig:tour1}.
  It still starts on the negative real axis and ends on the positive real
  axis, and crosses the negative imaginary axis.  But it crosses the
  imaginary axis with a value of $x$ that is order $r^2$ (and hence of
  order $\lambda^2$), instead of exactly zero.  We therefore get a zero of
  $A_j+i\epsilon$ for any small $\lambda$ for some small $\epsilon$, and have not avoided a
  singularity.  This gives the situation shown in Fig.\
  \ref{fig:eps.lam.sing}(a), which shows where, as we take $\epsilon$ and $\lambda$
  through the values used to try to get a successful deformation, we first
  encounter a singularity.

  Hence the deformation is not allowed.

\begin{figure}
  \centering
  \includegraphics[scale=0.7]{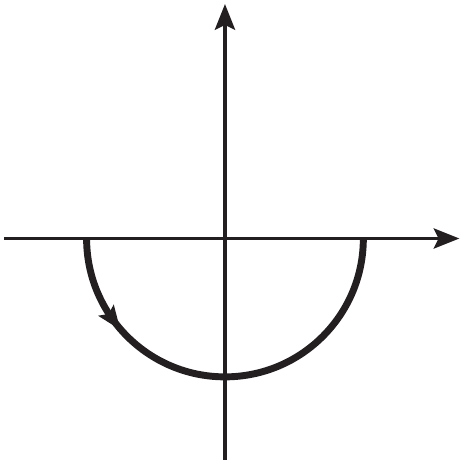}
  \caption{Value of $A_j$ takes approximately this tour in complex
    plane, in the case of (\ref{eq:case.linear}), with
    $\zeta=x+i\lambda=re^{i\theta}$, over $\theta$ from $0$ to $\pi$.}
  \label{fig:tour1}
\end{figure}
\end{proof}

From Thm.\ \ref{thm:v.D.neg}, the following property of allowed
deformations immediately follows:
\begin{theorem}
  \label{thm:v.D.neg.bis}
  Suppose that $v(w_R)$ specifies an allowed deformation. Then $v(w_S)\cdot D_j
  \geq 0$, whenever $A_j(w_S)=0$, for every $j$ and every real $w_S$ in the
  integration range.
\end{theorem}

The remaining case for $D_j(v_S)$ is that it is zero.  One possibility is
that $v_S$ itself is zero.  In that case $A_j(w_S+i\lambda v(w_S))=A_j(w_S)=0$,
Hence
\begin{theorem}
  \label{thm:v.zero}
  Suppose that $v_s=0$.  Then the deformation does not avoid the
  singularity due to the zero of $A_j$ at $w_S$.
\end{theorem}

This leaves one situation to treat, that $D_j(v_S)=0$ but $v_S\neq0$, which we
defer to Sec.\ \ref{sec:v.D.zero}.

The difficulties in its analysis concern the possibility of a non-constant
dependence of $v(w_R)$ on $w_R$.  So it is useful to prove the simple
results that obtain if $v(w_R)$ is independent of $w_R$, at least in a
neighborhood of $w_S$.  
\begin{theorem}
  \label{thm:v.D.zero.nonzero.const}
  Suppose that $v(w_R)$ is a candidate deformation that has no dependence
  on $w_R$ near $w_S$, and that $D_j(v(w_S)) = 0$ but $A_j(w_S+i\lambda v_S)$ is
  non-zero for some (non-zero) $\lambda$.  Then the deformation is not allowed.
\end{theorem}
\begin{proof}
  The non-zero value of $A_j(w_S+i\lambda v_S)$ implies that on the deformed
  contour we no longer need encounter a zero of $A_j+i\epsilon$ when we restrict
  attention to $w_R=w_S$.  To see this, first observe that the analyticity
  of $A_j(w_S+i\lambda v_S)$ as a function of $\lambda$ and its nonzero value for some
  value of $\lambda$ imply that the zero of $A_j$ at $\lambda=0$ is isolated. Then we
  can get a situation where no zero of $A_j(w_S+i\lambda v_S)+i\epsilon$ is encountered,
  for all small enough non-zero $\lambda$,

  But there are, in fact, zeros at nearby values of $w_R$, and these
  obstruct the deformation, as we now show.  Consider values $w_R =
  w_S+xv_s$, with $x$ real, so that on the deformed contour we have
  $A_j(w_S+(x+i\lambda)v_s)$.  This is an analytic function of $x+i\lambda$. The
  function is zero when $x=\lambda=0$, and by the hypothesis of the theorem is
  not zero for some values of $x+i\lambda$.

  Therefore there is a first non-zero term in the Taylor expansion:
  \begin{equation}
  \label{eq:v.D.zero.Taylor}
    A_j(w_S+(x+i\lambda)v_s) = C (x+i\lambda)^n + O\xleft( |x+i\lambda|^{n+1} \right),
  \end{equation}
  with $n\geq2$ since $v_S \cdot D_j=0$. Since $A_j$ is real for real values of
  its argument, so is $C$.  Set $x+i\lambda = r e^{i\theta}$, with $r$ positive.
  Allowed values have $0\leq\theta\leq\pi$, and any small $r$ is possible.  The value of
  $A_j$ is
  \begin{equation}
    Cr^n e^{i\theta n} + O(r^{n+1}).
  \end{equation}
  The value is real when $\theta$ is $0$ or $\pi$.  As $\theta$ is increased from $0$
  to $\pi$, the value $A_j$ must go round the origin $n/2$ times, i.e., at
  least once.  So it crosses the negative imaginary axis.  The order
  $r^{n+1}$ term from higher terms in the Taylor expansion can affect the
  position of this crossing, but do not affect its existence, at least when
  $r$ is small enough.

  Hence for small $\lambda$ we find a zero of $A_j+i\epsilon$ during the contour
  deformation, which therefore encounters a singularity, as in Fig.\
  \ref{fig:eps.lam.sing}(a).  Hence the deformation was not allowed,
  contrary to hypothesis.
\end{proof}

\begin{theorem}
  \label{thm:v.D.zero.const}
  Suppose that $v(w_R)$ is also required to be an \emph{allowed}
  deformation as well as having no dependence on $w_R$ near $w_S$, and that
  $v(w_S)\cdot D_j = 0$.  Then $A(w_S+i\lambda v_S)=0$ for all $\lambda$, and so the
  singularity due to the zero in $A_j$ is not avoided.
\end{theorem}
\begin{proof}
  This is an immediate consequence of Thm.\
  \ref{thm:v.D.zero.nonzero.const}:
\end{proof}

\begin{theorem}
  \label{thm:v.D.avoid.const}
  Suppose that $v(w_R)$ is required to be an \emph{allowed} deformation as
  well as having no dependence on $w_R$ near $w_S$, and that $A_j(w_S)=0$
  for one or more $A_j$, where $w_S$ is real.  Then the singularity due to
  the zero in $A_j$ is avoided if and only if $v(w_S)\cdot D_j$ is strictly
  positive, i.e., $D_j(v(w_S)) > 0$, for every one of the zero
  denominators. 
\end{theorem}
\begin{proof}
  This follows directly from the application of the last few theorems
  proved so far to multiple denominators, together with the results of
  Sec.\ \ref{sec:elem.parts}.
\end{proof}

This last theorem is Thm.\ \ref{thm:pinch.to.1st.order.shift} with a
restriction on the $w_R$ dependence of $v(w_R)$, but without any of the
extra restrictions on the denominator that appear in the statement of Thm.\
\ref{thm:main.contour}.
  
Combined with Thm.\ \ref{thm:main.geom}, it gives our primary Theorem
\ref{thm:main.contour} under the same conditions.

\subsection{Analysis of neighborhood of singularity of integrand}
\label{sec:v.D.zero}

Consider a point $w_S$ in an integral where some denominators are zero.

First consider the case that there is a vector $v_S$ such that $D_j(v_S)>0$
for the derivatives of all the zero denominators.  Then we choose a contour
deformation function which at $w_S$ is equal to $v_S$ (or proportional to
it with a positive coefficient).  Then we saw in Sec.\ \ref{sec:elem.parts}
that the deformation avoids the singularity at $w_S$.  This gives part of
the result stated in Thm.\ \ref{thm:pinch.to.1st.order.shift}.

To complete the proof of Thm.\ \ref{thm:pinch.to.1st.order.shift} (and
hence of Thm.\ \ref{thm:main.contour}) we now show that if no such $v_S$
exists, then the integration is trapped at $w_S$, i.e., that no allowed
deformation avoids the integrand's singularity at $w_S$.

We start by assuming that we have an allowed deformation, given by
$v(w_R)$, and that there is no $v_S$ such that $D_j(v_S)$ is strictly
greater than zero for all those $D_j$ that correspond to zero denominators.
We will obtain constraints that $v(w_R)$ must obey, and hence show that in
all cases the deformation does not avoid the singularity, thereby
completing the proof of the theorem.  As given in the statement of the
theorem, we will restrict attention to denominators that are at most
quadratic in the integration variable, and the reason for the remaining
restriction in the statement of the theorem will emerge in the course of
making the proof.

To simplify the notation, we shift the integration variable so that
$w_S=0$.  We define $v_0=v(0)$, the deformation at the singular point being
examined.

Since we cannot make all the relevant $D_j(v_0)$ positive, Thm.\
\ref{thm:main.geom} shows that the array of $D_j$s has a Landau point, i.e.,
there are values $\alpha_j$ such that
\begin{equation}
\label{eq:Landau.pt.deriv}
  \sum_j \alpha_j D_j = 0,
\end{equation}
with all the $\alpha_j$s being non-negative and at least one being positive. The
denominators for which $\alpha_j=0$ will play no role in the proof, and so we
now focus attention on only those values of $j$ with nonzero $\alpha_j$.  With
this focus, the denominators $A_j(w)$ in the retained set are zero at
$w=0$ and have strictly positive $\alpha_j$ in Eq.\
(\ref{eq:Landau.pt.deriv}). 

Now for an allowed deformation, $D_j(v_0)\geq0$.  From Eq.\
(\ref{eq:Landau.pt.deriv}), $\sum_j \alpha_j D_j(v_0) = 0$.  So strict positivity
of the $\alpha_j$ implies that each $D_j(v_0)$ is actually zero.

We expand each denominator in powers of $w$:
\begin{align}
\label{eq:denom.series}
  A_j(w) &= \sum_a D_{j,a} w^a + \frac12 \sum_{a,b} w^a E_{j,ab} w^b
\nonumber\\
       &= D_j \cdot w + \frac12 w \cdot E_j \cdot w,
\end{align}
given that the $A_j$ are at most quadratic in the integration variables.

On the deformed contour $w=w_R+i\lambda v(w_R)$, as usual.  
For the particular case of $w_R=0$, a denominator on the deformed contour
is 
\begin{equation}
\label{eq:denom.0}
  A_j(i\lambda v_0) = i\lambda D_j \cdot v_0 - \frac{\lambda^2}{2} v_0 \cdot E_j \cdot v_0 = -\frac{\lambda^2}{2} v_0 \cdot E_j \cdot v_0.
\end{equation}
If $v_0 \cdot E_j \cdot v_0$ is zero for at least one $j$, then we have a zero of
$A_j$ on the deformed contour, so that the deformation has not avoided the
singularity.  Then we need go no further for proving the target result.

So we now restrict attention to the case that $v_0\cdot E_j \cdot v_0$ is nonzero
for all of the attended denominators.  We now examine the denominators near
the origin, to search for possible zeros and to determine whether or not
they are avoided.  Define $\delta v$ by
\begin{equation}
  v(w_R) = v_0 + \delta v(w_R),
\end{equation}
so that $\delta v(w_R)$ goes to zero as $w_R$ goes to zero, i.e., $\delta(w_R) =
o(1)$ in this limit.  Then
\begin{align}
  A_j(w_R+i\lambda v(w_R)) 
\hspace*{-20mm}&
\nonumber\\
     = {}& D_j \cdot w_R + \frac12 w_R \cdot E_j \cdot w_R - \frac{\lambda^2}{2} v_0 \cdot E_j \cdot v_0
\nonumber\\ &
                       - \lambda^2 v_0 \cdot E_j \cdot \delta v
                       - \frac{\lambda^2}{2} \delta v \cdot E_j \cdot \delta v
\nonumber\\ &
      + i\lambda \left[ D_j \cdot\delta v + w_R \cdot E_j \cdot v_0 + w_R \cdot E_j \cdot \delta v  \right],
\end{align}
where the first two lines give the real part and the last line gives the
imaginary part.

We now search for possible obstructions to the contour deformation, i.e.,
zeros of $A_j+i\epsilon$ for small $w_R$, $\lambda$ and $\epsilon$.  Avoiding these will give
constraints on the functional form of $\delta v(w_R)$.  We do this by choosing a
small value of $w_R$, finding a value of $\lambda$ for which the real part of
$A_j+i\epsilon$ is zero, and then investigating the imaginary part.  A zero of the
real part is obtained by setting $\lambda=\lambda(w_R)$, where
\begin{equation}
  \lambda(w_R)
 =
 \sqrt{
 \frac{ 2 D_j \cdot w_R + w_R \cdot E_j \cdot w_R }
      { v_0 \cdot E_j \cdot v_0 + 2 v_0 \cdot E_j \cdot \delta v + \delta v \cdot E_j \cdot \delta v }
  },
\end{equation}
provided that the argument of the square root is positive.

Now let us consider the particular case that $w_R$ is in the direction
$v_0$ and set $w_R = xv_0$. Then
\begin{align}
  \lambda(xv_0)
 &=
 \sqrt{
 \frac{ x^2 v_0 \cdot E_j \cdot v_0 }
      { v_0 \cdot E_j \cdot v_0 + 2 v_0 \cdot E_j \cdot \delta v + \delta v \cdot E_j \cdot \delta v }
  }
\nonumber\\
 & =
 |x| ~
 \left( 1 + \frac{ 2 v_0 \cdot E_j \cdot \delta v + \delta v \cdot E_j \cdot \delta v }{ v_0 \cdot E_j \cdot v_0}  \right)^{-1/2}.
\end{align}
Then there is a zero of the real part of $A_j$ for all small enough $x$,
both positive and negative, and the solution has $\lambda(xv_0)\simeq |x|$. There the
value of $A_j$ only arises from its imaginary part
\begin{align}
   A_j(x v_0+i\lambda(x v_0) \, v(x v_0)) 
\hspace*{-35mm}&
\nonumber\\
      &= i \lambda(x v_0) \left[ D_j \cdot\delta v + xv_0 \cdot E_j \cdot v_0 + xv_0 \cdot E_j \cdot \delta v  \right]
\nonumber\\
      &= \lambda(x v_0) \left[ D_j\cdot \delta v + x v_0 \cdot E_j \cdot v_0 + o(x)  \right].
\end{align}
If at any point we were to get a negative imaginary part for all small $x$,
then a zero of $A_j+i\epsilon$ would be encountered in the contour deformation, so
that the deformation would not be allowed.  We therefore ask what
constraints apply to $\delta v(w_R)$ to avoid such a negative imaginary part.

For the deformation to be allowed, we must have
\begin{equation}
\label{eq:Im.Aj}
  D_j \cdot\delta v(x v_0) + x v_0 \cdot E_j \cdot v_0 + x v_0 \cdot E_j \cdot \delta v(x v_0) \geq 0
\end{equation}
for all small $x$.

First notice that $x$ can have either sign, and that when $x$ has the
opposite sign to $v_0 \cdot E_j \cdot v_0$, the term $x v_0 \cdot E_j \cdot v_0$ is
negative.  The $o(x)$ term is strictly smaller (in the limit $x\to0$). 

There are two cases to consider, according to whether $D_j$ itself is zero
or not.

If $D_j$ is zero, then $D_j\cdot v=0$, and the negative term $x v_0 \cdot E_j \cdot
v_0$ dominates; we have a negative imaginary part, and the contour
deformation is not allowed, contrary to our initial assumption.

Therefore $D_j$ must be nonzero. Then it is conceivable that the $D_j\delta v$
term compensates the negativity of $x v_0 \cdot E_j \cdot v_0$.

As announced in the statement of the theorem, we now restrict\footnote{It
  would be desirable to make a proof without this restriction, but it would
  require a harder proof beyond the scope of this paper.}  attention only
to cases with the property that all non-zero $v_0 \cdot E_j \cdot v_0$ have the
same sign.  Thus
\begin{align}
  \label{eq:restrict}
  \mbox{Either for all ``relevant'' $j$, $v_0 \cdot E_j \cdot v_0 \geq 0$;}
  \nonumber\\
  \mbox{or for all ``relevant'' $j$, $v_0 \cdot E_j \cdot v_0 \leq 0$},
\end{align}
where a ``relevant'' $j$ is one for which $\alpha_j$ is non-zero in Eq.\
(\ref{eq:Landau.pt.deriv}), and for which $A_j$ is zero, $D_j$ is nonzero,
and $D_j\cdot v_0=0$, both at the point of integration space under
consideration.  As already mentioned, if $v_0 \cdot E_j \cdot v_0$ is zero for at
least one relevant $j$, then the contour is definitely trapped, and we only
now examine the case where all the $v_0 \cdot E_j \cdot v_0$ are nonzero.

The restriction is obeyed for standard applications to Feynman graphs and
certain generalizations. To see this, observe that the standard Feynman
denominator for a line of a Feynman graph has the form $A_j=k^2-m^2$, where
$k$ is the line's momentum.  It is zero when $k^2=m^2$.  Let the projection
of an allowed deformation onto the momentum of the line be $\hat{v}_0$.  We
then have $D_j\cdot v_0 = 2k\cdot \hat{v}_0$ and $\frac12 v_0\cdot E_j \cdot v_0 =
\hat{v}_0 \cdot \hat{v}_0$.  For a massive line (i.e., $m\neq0$), all deformations
that obey $D_j\cdot v_0=0$ must have a space-like (or zero) $\hat{v}_0$, and
hence $v_0\cdot E_j \cdot v_0\leq0$.  In the massless case with $k^2=0$ and $k$
nonzero, $\hat{v}_0$ is either space-like or null (or zero), and again
$v_0\cdot E_j \cdot v_0\leq0$.  In the massless case with $k=0$, i.e., a soft line,
$D_j=2k=0$, so the denominator is not one of the relevant ones in Eq.\
(\ref{eq:restrict}). Another important case is of a Wilson line, for which
the denominator is simply linear: $A_j=k\cdot n$ for some vector $n$, and hence
$E_j$ itself is zero.  A non-relativistic propagator, with denominator
$E-\3p^2/2m$ has the for the quadratic term as in the massive relativistic
case.  One other case that can be met in QCD is an approximation where a
longitudinal light-front component of momentum is set to zero, but
transverse momenta are preserved.  Then the quadratic terms involve only
transverse momentum, and the quadratic terms obey the same sign condition
as for an unapproximated denominator.

Hence in all of these cases, the restriction (\ref{eq:restrict}) is
obeyed. 

Given this restriction (independently now of which sign occurs), the second
term in (\ref{eq:Im.Aj}) is negative when we give $x$ the opposite sign to
$v_0\cdot E_j \cdot v_0$.  Most importantly, the same value of $x$ can be used for
all the relevant denominators.  The third term in the imaginary part is
always smaller when the size of $x$ is small enough.  So the only hope for
getting a non-negative imaginary part is for the first term, $D_j \cdot\delta v(x
v_0)$, to compensate by being sufficiently positive.  Since $\delta v$ is zero
when $x$ is zero, this compensation relies on $x$-dependence in $v(x v_0)$,
and hence on $w_R$ dependence in $v(w_R)$.

In the present case, there is a Landau point, so that $\sum_j \alpha_j D_j \cdot\delta
v(xv_0)=0$. Hence at least one $D_j \cdot\delta v(xv_0)$ is not positive.  Hence,
for at least one $j$, the first term in (\ref{eq:Im.Aj}) cannot compensate
the negative value of the sum of the second and third terms.  Then there is
a zero in $A_j$ that causes an obstruction to the contour deformation, and
the proposed deformation would not be allowed.  

We have now covered all the cases, so that given the existence of a Landau
point, we have shown that all allowed contour deformations fail to avoid
the singularity.  This completes the proof of Thm.\
\ref{thm:pinch.to.1st.order.shift} and hence of Thm.\
\ref{thm:main.contour}.  Notice how we used the existence of a Landau
point, which was shown by a use of the geometrical Thm.\
\ref{thm:main.geom}.

We now revisit the rationale for extra restriction (\ref{eq:restrict}).  If
the restriction were not obeyed, then there would be at least one positive
and one negative $v_0 \cdot E_j \cdot v_0$.  The negative values of the second term
in in (\ref{eq:Im.Aj}) would occur for opposite signs of $x$ for different
denominators. Hence an appeal to $\sum_j \alpha_j D_j \cdot \delta v(xv_0) = 0$ would not be
sufficient to rule out a compensation of the negative terms by some choice
of $\delta v(xv_0)$.  A better argument would be needed, but I have not found
one that is watertight.

\subsection{Anomalous deformations}

All but the very last part of the derivation in the previous subsection
gives a strategy for finding examples like that in App.\
\ref{sec:2D.first.order}, where a singularity is avoided by a deformation
that has zero first-order shifts at the singular point(s).  Let us call
such a deformation an ``anomalous deformation'', formally defined by:
\begin{definition}
  \label{def:anom.def}
  An \emph{anomalous deformation} means an allowed deformation that avoids
  the singularity due to a zero of one or more denominators $A_j$, but
  where the first-order imaginary part is zero.
\end{definition}

What we did in the previous section, was to exclude the possibility that
when the Landau condition is obeyed an anomalous deformation could exist
that avoids singularities due to all of the denominators.  But the proof
relied on the extra restrictions on the denominators stated in Thm.\
\ref{thm:main.contour}.

If, in contrast, there is no Landau point, then we can find a vector that
gives positive first-order shifts in the denominators.  Hence, in this
situation of no Landau point, given the existence of an anomalous
deformation we can find another that is not anomalous and is still
singularity-avoiding.

\subsection{Patching local deformations to global}

The arguments in the preceding sections as to whether or not an integrand's
singularity can be avoided by a contour deformation were local.  That is,
the arguments were applied at each position $w_S$ where there a
singularity, and they involved determining (a) which directions of
deformation at $w_S$ are compatible with the integrand's singularities, and
most importantly (b) which directions avoid a singularity.

The question now arises as to whether such locally determined directions
can be globally patched together consistently, so as to give a contour
deformation $v(w_R)$ for all $w_R$ that has one of the determined
directions at each point of singularity of the integrand, and that can be
implemented without some kind of discontinuity.

As an indication of possible issues, the example of determining normal
directions to a M\"obius strip comes to mind.  This is a situation in which
the global topology of a surface prevents global patching of locally
determined vectors.  But the present situation is different.  At each point
on the initial integration contour, properties of the denominators
determine a manifolds of directions of singularity avoiding deformations
(and similarly for allowed deformations that don't avoid singularities).
The boundary of the manifold of possible directions depends continuously on
position in the manifold, and the denominators are single-valued functions
of position.  So we can steer the deformation to stay within the allowed
manifold.  Of course, at some parts of the original contour we may find
that no deformation is avoids singularities.

If we take a tour of the initial integration contour going from some
initial point back to the same point, then we have the same restrictions on
the direction of deformation at the start and end, and no inconsistency.

This is an extremely simple-minded argument, and undoubtedly too weak to be
fully persuasive.  An improved argument would be useful.

Of course, if we changed our integral to one in which a denominator $A_j$
had a branch cut on the initial integration contour, then the situation
would be different.  But that is not the case for the integrals that we
consider in this paper.  Now there are allowed to be non-integer exponents
in Eq.\ (\ref{eq:integral}), so that the integrand itself can have branch
points and cut(s) that are on the initial integration contour. But that
does not affect the possible directions of deformation, which are all
determined by the denominator functions themselves, $A_j$, which we require
always to be analytic and single valued.

\subsection{Case of one denominator}
\label{sec:one-denom}

We now examine the situation when there is only one denominator.  This is a
common special case, because it occurs when Feynman parameters are used.
Its analysis has some special features compared with the case of multiple
denominators, so it is useful to treat this case specially.  In particular,
we will understand explicitly why Coleman and Norton \cite{Coleman:1965xm}
needed to put the restriction on their proof, that the matrix of second
derivatives of the denominator has no zero eigenvalues.

Let the denominator be $A(w)$.  The Landau criterion for a putative pinch
at some point $w_S$ is simply that the denominator and its first derivative
$D(w)\eqdef \diff{A(w)}/\diff{w}$ are zero at $w_S$.  The aim is to show,
if possible, that the contour of integration is trapped at that point.  Of
course, given the zero derivative, any deformation that avoids the
singularity has a zero first-order shift in the denominator, and hence is
anomalous.  If we don't succeed in excluding the possibility of an
anomalous deformation, then at least we can strongly constraint its
properties and those of the denominator.  Of course, if $A(w_s)=0$ but the
derivative were non-zero, then we can certainly avoid the singularity at
$w_S$ by a deformation $w\mapsto w_R + i\lambda v(w_S)$ with $D(w_S)\cdot v(w_S)>0$.  So
the case of a zero derivative is the only one to examine further.

If the denominator is quadratic in the integration variables, then the
restrictions in Thm.\ \ref{thm:main.contour} are obeyed, and the derivation
in previous sections is valid.  But the denominator from applying the
Feynman parameter method to a standard Feynman graph is cubic if the
momentum integrals are not performed.

If the momentum integrals are performed, as can be done analytically for
standard Feynman graphs, then the denominator is a polynomial of order one
plus the number of loops.  It can therefore be of arbitrarily high order.
But as we have already observed, a pinch in momentum space does not always
entail a pinch in parameter space, so this case isn't so generally useful.

Given the significance of the Feynman parameter representation of a Feynman
graph before the momentum integrals are performed, we will restrict
attention to the case that the denominator is at most cubic in the
integration variables. As before, given a point $w_S$ where the denominator
and its derivative are zero, we simplify the notation by shifting variables
so that $w_S=0$.  Then the denominator has the form:
\begin{align}
\label{eq:one.denom}
  A(w) &= \frac12 \sum_{ab} w^a E_{ab} w^b + \frac16 \sum_{abc} w^aw^bw^c F_{abc}
\nonumber\\
       &= \frac12 w \cdot E \cdot w + \frac16 w w w \cdot F.
\end{align}
where each of the arrays $E$ and $F$ is symmetric in its indices.
At certain points it will be useful to follow Coleman and Norton, and
diagonalize $E$ by a change of variable, to write
\begin{equation}
  \label{eq:E.diag}
  w \cdot E \cdot w = \sum_j c_j \eta_j^2,
\end{equation}
with each $\eta_j$ being a linear combination of $w^a$s.  By rescaling the
$\eta_a$, we can arrange that each non-zero $c_a$ has absolute value unity.
Thus without loss of generality, we can arrange that each $c_a$ is either
$+1$, $-1$ or $0$.

There are several different cases to consider, so it is convenient to
encapsulate in a lemma each of the separate cases, as well as several
subsidiary results.  The first lemma is elementary:
\begin{lemma}
  For a contour deformation $w = w_R +i\lambda v(w_R)$ to avoid the singularity at
  caused by the denominator (\ref{eq:one.denom}) at $w=0$, it is necessary
  (but not sufficient, as we will see), for $v_0\cdot E\cdot v_0$ or $v_0v_0v_0\cdot F$
  (or both) to be nonzero.  Here $v_0=v(0)$.

  Conversely, if both of $v_0\cdot E\cdot v_0$ and $v_0v_0v_0\cdot F$ are zero, the
  singularity is not avoided.
\end{lemma}
\begin{proof}
  The trivial proof is to observe that $A(w_R+i\lambda v(w_R)$ needs to be nonzero
  at $w_R=0$ if the singularity is to be avoided.
\end{proof}

\begin{lemma}
  For an allowed deformation, $v_0\cdot E\cdot v_0$ must be zero. 
\end{lemma}
\begin{proof}
  The proof is a minor modification of the argument leading to Eq.\
  (\ref{eq:Im.Aj}).  If $v_0\cdot E\cdot v_0$ were nonzero, then we would have a
  zero of the real part of $A$ with $\lambda$ close to $|x|$.  The higher-than
  quadratic terms that we now have do not affect that result.  But in Eq.\
  (\ref{eq:Im.Aj}) we now longer have the single-derivative term, which is
  therefore not available to compensate the negative value of $xv_0\cdot E\cdot
  v_0$ that occurs when $x$ has the opposite sign to $v_0\cdot E\cdot v_0$.  The
  cubic term does not affect that for small $\lambda$ (and $x$).  So the
  constraint Eq.\ (\ref{eq:Im.Aj}) for an allowed deformation cannot be
  obeyed. 

  This leaves only the case of zero $v_0\cdot E\cdot v_0$ for an allowed deformation.
\end{proof}

\begin{lemma}
  For an allowed deformation, $v_0$ must be an eigenvector of $E$ with
  eigenvalue 0.
\end{lemma}
\begin{proof}
  We now use the change of variables that gives the diagonalized form for
  the quadratic term, Eq.\ (\ref{eq:E.diag}), with each $c_a$ being either
  $+1$, $-1$ or $0$.  Let $h$ be the result of applying the change of
  variables to $v_0$.  Then
  \begin{equation}
  \label{eq:vEv.diag}
    v_0\cdot E\cdot v_0 = \sum_j c_j h_j^2 = 0.
  \end{equation}

  There are two cases to consider.  One is where the only nonzero values of
  $h_j$ are with $c_j=0$.  Then $v_0$ is an eigenvector of $E$ with
  eigenvalue zero, so we are done.

  The other case is where there is at least one $j$ with both of $h_j$ and
  $c_j$ nonzero.  To get the zero value in Eq.\ (\ref{eq:vEv.diag}), there
  must be at least one positive term and one negative term.  Permute the
  labels so that the $j=1$ term is positive and the $j=2$ term is
  negative: $v_0\cdot E\cdot v_0 = h_1^2-h_2^2 + \mbox{terms from other $j$}$, with
  both of $h_1$ and $h_2$ nonzero.

  We now find a zero of the denominator that gives deformation-obstructing
  singularity in the integrand.  Choose $w_R$ to correspond to
  \begin{equation}
    \eta_j = x \delta_{j1} + y \delta_{j2}.
  \end{equation}
  \begin{widetext}
  Then the denominator is
  \begin{align}
    A(w_R(x,y) + i \lambda (v_0+\delta v))
    ={}& \frac12 x^2 -  \frac12 y^2 + O(\lambda^2\delta v^2) + O(\lambda^3)
         +i\lambda \left[ xh_1 -yh_2 + O(|x|\delta v) + O(|y|\delta v) + O(\lambda^2) \right].
  \end{align}
  \end{widetext}
  If the first two terms in the real part had no corrections, then it would
  be zero whenever $|x|=|y|$, with all combinations of signs allowed.  By
  taking $x$ to have the opposite sign to $h_1$ and $y$ to have the same
  sign as $h_2$, we get a negative value for the first two terms in the
  imaginary part.  We now choose $x$ and $y$ to be of order $\lambda$, and take
  $\lambda$ to zero.  Then the correction terms in the real part are smaller than
  the first two terms, and cause the position of the zero to move slightly,
  with the fractional change decreasing to zero as $\lambda\to0$.  The correction
  terms in the imaginary part are similarly small than the first two terms,
  and leave the imaginary part negative.  

  Hence when $\lambda$ and $\epsilon$ are decreased to zero, a zero of
  $A_j(w_R+i\lambda v(w_R))+i\epsilon$ is always encountered somewhere on the integration
  contour, and hence the deformation is obstructed by a singularity and is
  not allowed.

  Hence the case that $v_0$ is not an eigenvector of $E$ of eigenvalue zero
  is ruled out, and the lemma is established.
\end{proof}

\begin{lemma}
  \label{lemma:PSS.tangent}
  Suppose $w=0$ is part of a manifold $M$ of points satisfying the Landau
  condition, i.e., $A=0$ and $D=0$.  Then (a) any tangent vector $t$ to the
  manifold (at $w=0$) is an eigenvector of $E$ of eigenvalue zero; (b) $t t
  t \cdot F=0$; (c) hence a contour deformation whose $v_0$ is tangent to $M$
  at $w=0$ does not avoid the singularity.
\end{lemma}
\begin{proof}
  Consider a path within the manifold $M$, starting at the origin and with
  initial direction $t$.  Let the path be parameterized by $P(x)$ where $x$
  is real, $P(0)=0$, and $P'(0)=t$, with $w'(x)$ denoting $\diff{P(x)}/
  \diff{x}$.  Then for all $x$ for which $P(t)\in M$,
  \begin{subequations}
  \begin{align}
    \label{eq:A.path.0}
    0 & = A(w(x)) = \frac12 P(x) \cdot E \cdot P(x) + \frac16 P(x)P(x)P(x) \cdot F,
  \\
    \label{eq:D.path.0}
    0 & = D(w(x)) = E \cdot P(x) + \frac12 P(x)P(x) \cdot F.
  \end{align}
  \end{subequations}
  with the second equation meaning
  \begin{align}
    0 & = D_a(w(x)) = \sum_b E_{ab} P^b(x) + \frac12 \sum_{bc} F_{abc} P^b(x)P^c(x).
  \end{align}
  
  Differentiate Eq.\ (\ref{eq:D.path.0}) with respect to $x$ to get
  \begin{equation}
    0 = \frac{\diff{D(w(x)}}{\diff{x}}
       = E \cdot P'(x) + P'(x)P(x) \cdot F.
  \end{equation}
  Set $x=0$ to get $E\cdot t=0$, i.e., $t$ is an eigenvector of $E$ of
  eigenvalue zero. 

  Now differentiate Eq.\ (\ref{eq:D.path.0}) twice with respect to $x$ to
  get
  \begin{equation}
    0 = E \cdot P''(x) + P'(x)P'(x) \cdot F + P(x)P''(x) \cdot F.
  \end{equation}
  Setting $x=0$, using $P(0)=0$ and $P'(0)=t$, and contracting with $t$
  gives
  \begin{equation}
    t t t \cdot F = 0,
  \end{equation}
  which gives item (b) in the lemma.

  Then because both $t\cdot E\cdot t$ and $ttt \cdot F$ are zero, a contour deformation
  with $v_0\propto t$ gives a zero of $A$ at $w_R=0$ on the deformed contour, so
  that the singularity is not avoided.  This proves item (c).
\end{proof}

\begin{lemma}
  For any contour deformation that avoids the singularity (necessarily an
  anomalous deformation), $v_0$ has eigenvalue zero, but is not tangent to
  any manifold such as $M$ in the previous lemma.
\end{lemma}
\begin{proof}
  This is immediate from the previous two lemmas.
\end{proof}

\begin{lemma}
  \label{lemma:PSS.anomalous}
  For there to exist a contour deformation that avoids the singularity, $E$
  must have an eigenvector of eigenvalue zero that is not in the space of
  directions of manifolds of the form of $M$.
\end{lemma}
\begin{proof}
  Immediate from the previous lemma.
\end{proof}

We now see why Coleman and Norton needed their restriction on the
eigenvalues of $E$.  However, they did not explain why, and the argument in
this section appears to show that the derivation is non-trivial.  They
refer to Ref.\ \cite{Eden:1961} for situations when the zero eigenvalue
problem arises.

There is one common case of zero eigenvalues in a massless theory, and that
is when there is a collinear region.  In that case the corresponding
pinch-singular-surface is not simply a point, being parameterized by
longitudinal momentum fraction(s).  But Lemmas \ref{lemma:PSS.tangent} and
\ref{lemma:PSS.anomalous} show that if the only zero eigenvalues are for
tangents to the surface, the singularity is not avoided.  There is perhaps
an obscure reference to this in the third and fourth lines of p.\ 441 of
Ref.\ \cite{Coleman:1965xm}.

Undoubtedly, it is possible to examine in more detail the case with a zero
eigenvalue, and to find further constraints on allowed deformations.  If
the contour deformation direction $v(w_R)$ were required to be independent
of $w_R$, or sufficiently slowing varying, then Thms.\
\ref{thm:v.D.zero.nonzero.const} and \ref{thm:v.D.zero.const} show that in
full generality it is not possible to avoid the singularity due to a zero
of the denominator and its derivative.  The elementary proof simply uses
the first non-zero term in the Taylor expansion of $A(zv_0)$ in powers of
$z$.  In our case it would be the cubic term that is relevant.

However, the deformation direction can depend on $w_R$.  The
non-trivial problem is there can then arise a nonzero contribution to the
quadratic term involving $\lambda\delta v$, and this is potentially capable of
compensating the part of the cubic term that would otherwise result in an
unavoidable singularity in the contour deformation.

\section{Conclusions and implications}

In this paper, I have provided a complete proof of the necessity and
sufficiency of the Landau condition for a pinch in the kind of integral
typified by Feynman graphs in the physical region.  The proof overcomes a
number of deficiencies in existing work, and it can be applied directly to
Feynman graphs in momentum space (unlike many previous proofs).  The
analysis of pinch singularities is foundational to perturbative QCD, so it
is important not only to have a full explicit proof, but to have one whose
domain of application, as here, includes Feynman graphs with massless
propagators as well as massive ones, and also modified propagators such as
the Wilson-line denominators that are common in QCD applications.

The methods and intermediate results in the proof have further
implications, beyond simply determining where pinches occur.  For example,
an analysis of coordinate-space behavior can be made by deforming a contour
of integration as much as possible to convert rapidly oscillating
exponential factors into strongly decaying exponentials.  Dominant regions
are determined by the locations where such a deformation cannot be made.
Allowed directions of deformation are constrained not only by the need to
avoid singularities of the integrand, but also to avoid making the
exponentials rapidly growing. Dominant regions of the integration variables
are where the constraints cannot be satisfied, and it is useful to have an
analysis that works at all orders of perturbation theory.

Another possible application, especially of the geometric results in Sec.\
\ref{sec:overall}--\ref{sec:landau.proof}, is to improve algorithmic
methods for deforming contours in numerical calculations of Feynman
graphs, as in Refs.\ \cite{Gong:2008ww,Becker:2012nk,Becker:2012bi}. 

In constructing the proof, some interesting subsidiary results were found.
Some of these simply resulted from a close analysis of treatments in the
classic literature (which give a strong inspiration to treatments in
textbooks).  Particular problems and even a demonstrably false assumption
were found.  Awareness of such issues is important to provide sound and
properly persuasive pedagogical treatments.

Another notable case was to recognize the possibility of avoiding a
singularity in the integrand by a contour deformation in a direction that
is tangent to the singularity surface.  Such a deformation I termed
``anomalous''.  With such a deformation, the first order shift of a
denominator due to the contour deformation is zero.  This contrasts with
the natural intuition (engendered by experience in one dimensional cases)
that, in order to avoid a singularity of the integrand, the contour
deformation must give a positive first-order shift to the imaginary part of
the denominator, matching the sign of the $i\epsilon$.  An example of an anomalous
deformation was found.  But this was a case where the contour is not
trapped, in which case there are also a non-anomalous contour deformations
that avoid the singularity.

As regards the proof given here, considerable complications were
encountered in excluding the possibility that one can have a situation
where a Landau condition is obeyed, but the contour is \emph{not} trapped;
that is, it was necessary to rule out the possibility of an anomalous
deformation when a Landau condition is obeyed.  A proof was found only when
the denominators in the integral obeyed certain conditions.  Luckily, these
conditions are indeed obeyed for Feynman graphs --- see the statement of
Thm.\ \ref{thm:main.contour} and Eq.\ (\ref{eq:restrict}).  A more general
proof (or counterexample) would be obviously be useful.

Here are some possible directions for further work.
\begin{enumerate}
\item It would be useful to apply the methods to give a fully systematic
  and general account in coordinate space of the large-$Q$ behavior of
  amplitudes, such as appear in QCD factorization.  This would extend, for
  example, the work of Erdo\u{g}an and Sterman
  \cite{Erdogan:2014gha,Erdogan:2016ylj,Erdogan:2017gyf}.

\item Another direction is to determine from the geometrical considerations
  given in this paper the possible directions for allowed contour
  deformations at a pinch.  Given the existence of a pinch at a particular
  point or at a points on some manifold, there is a certain set of
  denominator(s) whose corresponding singularity/ies of the integrand
  cannot be avoided.  These effectively are the denominators that actually
  cause the pinch.  But it is possible that other denominators are zero,
  but that the corresponding singularities can be avoided by a contour
  deformation that respects the constraints given by the pinching
  denominators.  It would be useful to have a determination of the range of
  allowed directions.

  Such issues were not important in the original application of the Landau
  analysis to determine singularities of an integral as a function of
  \emph{external} parameters.  But in QCD applications, the focus is rather
  on the momentum configurations at a pinch and their neighborhoods.  The
  exact pinches of relevance in standard pQCD applications are in a
  massless theory, whereas the true theory is not massless; the massless
  version is simply a useful tool for locating relevant regions in the
  space of loop momenta.  Moreover, a subtracted hard scattering
  coefficient, calculated in the massless limit as usual, is singular at
  zero mass, but the singularity is not strong enough to make the hard
  scattering actually divergent there, given the subtractions.  (The same
  is not true of the derivatives of sufficiently high order with respect to
  mass at zero mass.)

\item Consider the Coleman-Norton result that locations where a Landau
  condition is obeyed correspond to possible classical processes.  Their
  result is very useful for readily determining the well-known results on
  regions involved in asymptotic large $Q$ behavior, notably the
  classification into hard, collinear, and soft subgraphs.  Coleman and
  Norton's proof works in the massive case, but it becomes singular in the
  massless case and doesn't fully capture \cite{Ma:2019hjq} what is
  actually needed in QCD applications.  It would be useful to remedy this
  problem, perhaps in conjunction with a systematic treatment in coordinate
  space.

\end{enumerate}

\section*{Acknowledgments}

I thank Marko Berghoff, Yao Ma, Maximilian M\"uhlbauer, Dave Soper, and
George Sterman for useful conversations.


\appendix

\section{Contour trapping without singularity}
\label{sec:trap.no.sing}

Consider the following dimensionless function:
\begin{equation}
\label{eq:trap}
  I(Q/m) = 
    m^2Q \int_{-\infty}^{\infty} 
         \frac{\diff{k}}
              { (k-m+i\epsilon) (k-i\epsilon)^2 (k+Q-i\epsilon) } .
\end{equation}
This is intended to be a simple analog of a QCD Feynman graph with a large
momentum scale $Q$, a mass scale $m$, and with a certain numerator factor.
The singularities of the integrand as a function of $k$ are shown in Fig.\
\ref{fig:trap.no.sing}. There is a double pole at $k=0$ just above the
contour of integration, a pole at $k=m$ below the contour and a pole at
$k=-Q$.  When $m\to0$, the contour is evidently trapped at $k=0$.

\begin{figure}
  \centering
  \includegraphics[scale=0.65]{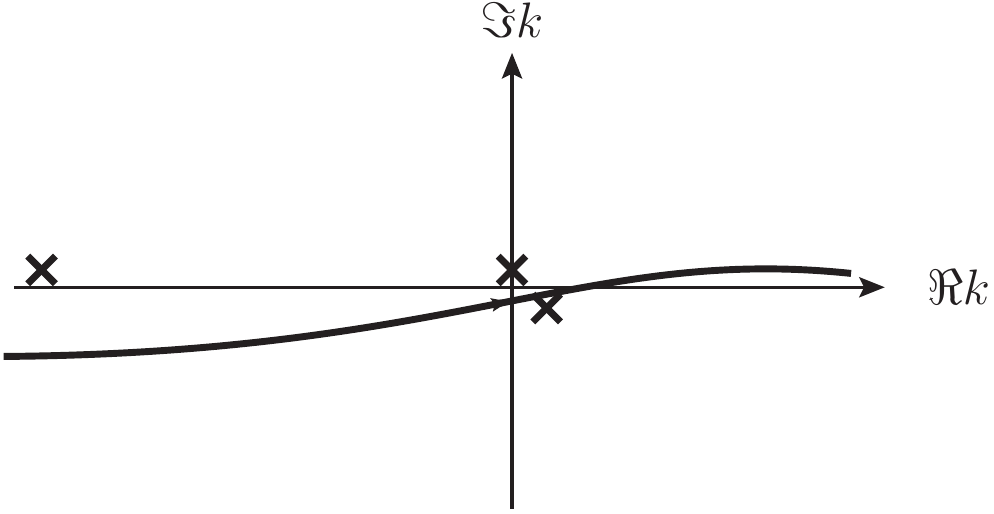}
  \caption{Singularities in $k$-plane for integrand in Eq.\
    (\ref{eq:trap}), with a trapping of the contour of integration at
    $k$ of order $m$.} 
  \label{fig:trap.no.sing}
\end{figure}

For the integral itself, without the $m^2Q$ prefactor, elementary contour
integration using the residue theorem shows that the value is $-2\pi
i/[m^2(m+Q)]$.  But with the explicit prefactor $m^2Q$, we find that the
function $I(Q/m)$ equals $-2\pi i/(1+m/Q)$, which has no singularity at
$m=0$.

The reason is for the lack of singularity is entirely trivial: The would-be
power singularity is canceled by an explicit numerator factor $m^2$.

Since numerator factors can occur in Feynman graphs, their presence
allows a potential violation of an absolute connection between the
Landau condition and actual singularities of a Feynman graph as a
function of external momenta or masses. Note that most singularities
we treat in QCD are logarithmic, and hence are not so easily removed.

However, even though there is no singularity in the function $I(Q/m)$,
there is a pinch in the integral in Eq.\ (\ref{eq:trap}), and the prefactor
$m^2$ does not remove this pinch.  Thus when $m/Q\to0$, the computation of
the integral is dominated by small values of $k$, of order $m$.  The
contribution of much larger values of $k$ is power suppressed because of
the large number of denominator factors in the integrand.

If it were possible to deform the contour away from small values of
$k$, i.e., values of order $m$, then $|k|$ would be of order $Q$
instead of being sometimes much smaller on the contour of integration.
This could happen with a different choice of $i\epsilon$ prescriptions.  In
such a case, the result for $I$ have been of order $m^2/Q^2$ instead
of order unity, as $Q/m\to\infty$.  Thus the order unity result for $I$, in
conjunction with power counting for the ``ultra-violet'' region of
large $k$, is a symptom that the integration is trapped at small $k$.

The importance of this result in QCD is given by considering the statement
by Libby and Sterman at the beginning of their paper \cite{Libby:1978bx}.
They say that quantities in QCD with a large external scale $Q$ can be
effectively computed provided that there are no mass divergences.  Taken
literally, this statement is falsified by examples like Eq.\
(\ref{eq:trap}).  But Libby and Sterman's statement becomes correct if the
no-singularity property is replaced by a no-pinch property.  In that case,
the Landau condition is both necessary and sufficient. Moreover, it is, in
fact, the presence or absence of pinches that is relevant for the QCD
applications.

\section{Singularity avoidance without a first-order shift in the denominator}
\label{sec:zero.first.order}

\subsection{Example of singularity avoidance with zero first-order shift in
  denominator}
\label{sec:2D.first.order}

Consider an integral of the following form
\begin{equation}
  I = \int \diff{E} \diff{p} \frac{i}{E-p^2/(2m)+i\epsilon}
                          f(E,p).
\end{equation}
The first factor has a singularity at $E=p^2/(2m)$, and we will
consider contour deformations to avoid it.  The other factor $f(E,p)$
generally has singularities.  But for the purposes of constructing an
example, we will assume they are far enough away not to concern us.
We could choose a function like $f= 1/(E^2+p^2+Q^2)^2$, with $Q\gg m$;
this factor has no singularities for real $E$ and $p$ and gives good
convergence of the integral in the ultra-violet (i.e., at large $p$
and $E$).

The first factor is of the form of the propagator for a non-relativistic
particle, indicating that this example is directly relevant to physics.
Similar treatments to the one in this section can be applied in the
relativistic case, but with more complication.

The derivative of the propagator with respect to the two-dimensional
integration variable is
\begin{equation}
  D = (1, -p/m).
\end{equation}

Let a contour deformation be made:
\begin{equation}
  (E,p) = (E_R,p_R) + i\lambda(\eta,\xi),
\end{equation}
where $E_R$, $p_R$, $\eta$ and $\xi$ are real, and $\eta$ and $\xi$ are
functions of $E_R$ and $p_R$.  As usual, $0\leq\lambda\leq1$.

A natural and obvious candidate for a contour deformation to avoid the
pole simply has $\eta$ positive and $\xi=0$, so that
\begin{equation}
  (\eta,0) \cdot D = \eta > 0.
\end{equation}
for then the singular factor is
\begin{equation}
  \frac{i}{ E_R - \frac{p_R^2}{2m} + i\lambda\eta + i\epsilon }
\end{equation}
and we avoid encountering a singularity as we deform the contour.

We now construct two examples of contour deformation that avoid the pole,
but which are anomalous, i.e., they obey
\begin{equation}
\label{eq:lin}
  \bigl( \eta(E_R,p_R), \xi(E_R,p_R) \bigr) \cdot D = 0
\end{equation}
in one or more situations where there is a singularity before
deformation, i.e., where $E_R=p_R^2/(2m)$.  The first example will
obey this condition at $E_R=p_R=0$.  The second will obey it at all
values of $p_R$.  In both cases, the deformation direction is along a
tangent to the surface of singularity.

To see how to construct such an example, consider the denominator on
the deformed contour, in the case of a general deformation:
\begin{align}
\label{eq:den}
  E - \frac{p^2}{2m} + i\epsilon
  ={}& E_R -\frac{p_R^2}{2m}
   + \lambda^2\frac{\xi^2}{2m}
\nonumber\\&
    + i\lambda(\eta-\xi p_Rm) +i\epsilon.
\end{align}
For a deformation to be allowed and to avoid the pole, we must arrange
$\eta(E_R,p_R)$ and $\xi(E_R,p_R)$, such that the only zeros of the
denominator occur at $\epsilon=\lambda=0$.  

In the first example, we choose $\eta(0,0)=0$ and $\xi(0,0)=m$, so that the
first order shift in the denominator, $\left( \eta(E_R,p_R), \xi(E_R,p_R)
\right) \cdot D$, is zero at $E_R=p_R=0$. Because of the non-zero $\lambda^2$ term in
(\ref{eq:den}), we no longer encounter a singularity on the deformed
contour, for the case that $E_R=p_R=0$.  Now, if $\eta$ and $\xi$ were given no
dependence on $E_R$ and $p_R$, the denominator (\ref{eq:den}) would have a
negative imaginary $\lambda$ term when $p_R$ is positive.  We could make the real
part of the denominator zero by choice of $E_R$, and then the whole
denominator becomes zero at some point as we reduce $\epsilon$ to zero.  In this
case, the contour deformation crosses a singularity somewhere, and is not
allowed.

But by choosing the $E_R$ and $p_R$ dependence of $\eta$ and $\xi$
appropriately, we can compensate the negative imaginary part.

We first make the following choice:
\begin{equation}
  (\eta,\xi) = ( |E_R|+2|p_R|, m ),
\end{equation}
so that the denominator is
\begin{align}
  E - \frac{p^2}{2m} + i\epsilon
  ={}& E_R - \frac{p_R^2}{2m} + \frac{\lambda^2m}{2}
\nonumber\\&
    + i\lambda \left( |E_R| + 2|p_R| - p_R \right) +i\epsilon.
\end{align}
We ask when this is zero.  On the deformed contour, i.e., with $\lambda$
positive, the term linear in $\lambda$ is positive and non-zero except for
being zero at $E_R=p_R=0$.  Because $\lambda$ is non-zero, the real part is
non-zero here.  So in deforming the contour we encounter no poles; the
deformation is allowed, despite the zero first-order shift.

Notice how the deformation has a acquired a component in the direction
of the natural deformation.  At a general values of the integration
variables, the first order shift is $i\lambda(|E_R|+2|p_R|-p_R)$, which is
always positive, except at $E_R=p_R=0$.  So the deformation is only
anomalous at this one point.

The second example of a deformation is more striking because it avoids
the singularity for all $p_R$, but is anomalous everywhere:
\begin{equation}
  (\eta,\xi) = \left( p_R+ \frac{p_R^2}{2m} - E_R, m \right),
\end{equation}
so that the denominator is
\begin{multline}
\label{eq:anom.deform2}
  E - \frac{p^2}{2m} + i\epsilon
  = E_R - \frac{p_R^2}{2m} + \frac{\lambda^2m}{2}
    + i\lambda \left( \frac{p_R^2}{2m} - E_R \right) +i\epsilon.
\end{multline}
This looks more dangerous because the imaginary part in the $i\lambda$ term
is not always positive.  However, to get a zero of the whole
denominator, both the real and imaginary parts must be zero.  A zero
real part gives $E_R = \frac{p_R^2}{2m} - \frac{\lambda^2m}{2}$, and then the imaginary part
is
\begin{equation}
\label{eq:ex2.Im}
    i \frac{\lambda^3m}{2} + i\epsilon.
\end{equation}
This is zero if and only if $\lambda=\epsilon=0$, and again we have avoided the pole.
Notice how the positive $\lambda$-dependent imaginary part in (\ref{eq:ex2.Im})
is of order $\lambda^3$ instead of its usual size $\lambda$; the careful choice of
deformation has canceled bigger terms.

The $i\lambda$ term in (\ref{eq:anom.deform2}) does go negative, but only where
the real part of the denominator is definitely non-zero.

\subsection{In one dimension, singularity avoiding deformation requires
  positive first-order shift in denominator}
\label{sec:1D.first.order}

To obtain a singularity-avoiding deformation with a zero first-order shift,
the deformation vector $v(w_R)$ at a singular point needed to be non-zero
and tangent to the singularity surface.  But in one dimension,
singularities are at points, and there is no surface to which a tangent can
be constructed.  So we expect that a minimum example of an anomalous
deformation must be in two dimensions.  In this section we analyze the
one-dimensional case in more detail.

The denominators are $A_j(z)+i\epsilon$, with $z$ being an ordinary complex
number.  As usual for this paper, $A_j$ is real when $z$ is real and is an
analytic function of its argument.  Suppose that a particular $A_j(z)$ has
a zero at a real value $z=x_S$.  To simplify the notation, shift the
integration variable so that the zero is at $z=0$.  We now investigate the
conditions under which the corresponding singularity is or is not avoided
by a contour deformation $x\mapsto x+i\lambda v(x) = x+i\lambda(v_0+\delta v(x))$.  Here $v_0=v(0)$,
the direction of contour deformation at $x=0$, so that $\delta v(0)$ = 0.

Let the Taylor expansion of $A_j$ be
\begin{equation}
  A_j(z) = \sum_{n=1}^\infty a_n z^n,
\end{equation}
with all the $a_n$ real.  Then $D_j = A_j'(0)=a_1$.

If $D_jv_0>0$, then the deformed contour avoids singularities due to $A_j$
in a neighborhood of $x=0$, since the imaginary part of $A_j$ on the
deformed contour for all $x$ near zero.

If $D_jv_0<0$, then the denominator is negative imaginary at $x=0$.  Then
the deformation crosses a singularity, and hence is not allowed.

If $v_0=0$, then $A_j(x+i\lambda v(x))$ is zero at $x=0$; the singularity is not
avoided.

The above cases (and the trivial proofs) are no different than in the
multi-dimensional case.

The remaining case is $D_j v_0=0$ with $v_0$ nonzero.  This entails
$a_1=0$.  Here is the first difference between the one-dimensional case and
higher dimensions.  In higher dimensions $D_jv_0$ can be zero while having
both $D_j$ and $v_0$ be nonzero.

Let the lowest order non-zero $a_n$ be for $n=n_0$ with $n_0\geq2$. Then
\begin{equation}
  A_j(z) = \sum_{n=n_0}^\infty a_n z^n,
\end{equation}
Redefine the real value $x$ to be $\hat{x}v_0$, so that
\begin{multline}
  A_j(x+i\lambda v(x)) = a_{n_0}v_0^{n_0} (\hat{x}+i\lambda)^{n_0}
                 + O(|x+i\lambda|^{n_0}\delta v(x))
\\
                 + O(|x+i\lambda|^{n_0+1}).
\end{multline}
Set $\hat{x}+i\lambda=re^{i\theta}$.  Then
\begin{equation}
\label{eq:A.one-dim}
  A_j(x+i\lambda v(x)) = a_{n_0}(rv_0)^{n_0} e^{in_0\theta}
                 + o(r^{n_0}),
\end{equation}
where the $o(r^{n_0})$ notation means that this (``correction'') term
divided by $r^{n_0}$ goes to zero as $r\to0$.  When $v_0$ is positive,
allowed (non-negative) values of $\lambda$ correspond to $0\leq\theta\leq\pi$.  If is $v_0$ is
negative, the allowed values correspond to $0\geq\theta\geq-\pi$. Because $n_0\geq2$, the
first term on the right of Eq.\ (\ref{eq:A.one-dim}) goes in a circle round
the origin at least once when $\theta$ goes through its allowed values.
Therefore, there is at least one allowed value of $\theta$ where the real part
is zero and the imaginary part is negative.  The correction term modifies
this position of this situation by only a small amount if $r$ is small
enough, but doesn't affect its occurrence. Hence as the contour is deformed
from $\lambda=0$ and as $\epsilon$ is taken to zero, a zero of $A_j+i\epsilon$ is always
encountered, and therefore the deformation is not allowed.

Combining all these cases shows that the contour deformation in one
dimension avoids the singularity if and only if $D_j v_0$ is strictly
positive.  The singularity is not avoided if $D_jv_0=0$.  It follows that
anomalous deformations only exist in two or more dimensions.

\section{Illustrative examples}

\subsection{Elementary example of Landau singularity}
\label{sec:self-energy}

\begin{figure}
  \centering
  \includegraphics[scale=0.7]{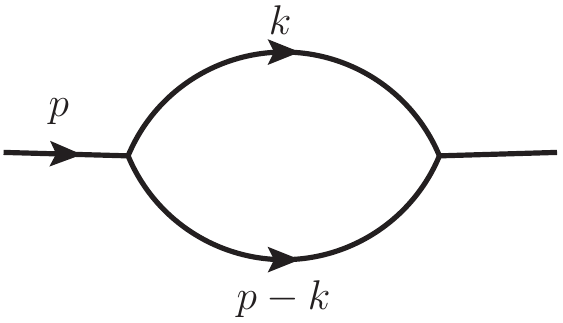}
  \caption{Self-energy graph.}
  \label{fig:self.energy}
\end{figure}

To provide elementary example of a Landau singularity, and to examine it in
the light of the approach used in this paper, and also to be able to later
pinpoint differences between the massive and massless cases, we consider
the self-energy graph of Fig.\ \ref{fig:self.energy}.  The integral for it,
with couplings, external propagators and symmetry factor omitted, is
\begin{equation}
  \label{eq:self.energy}
  \Pi(p,m) = \int \frac{\diff[n]{k}}{(2\pi)^n}
             \frac{1}{ \left[k^2-m^2+i\epsilon\right] \left[(p-k)^2-m^2+i\epsilon\right]},
\end{equation}
where we work in an $n$-dimensional space-time, and where $p$ is the
external momentum and $m$ the mass for each line, which we will assume to
be non-zero for the moment.  Since the graph is ultra-violet (UV) divergent
when $n\geq4$, we will use $n<4$.  The modification to (\ref{eq:self.energy})
to implement renormalization of UV divergences is in fact
irrelevant\footnote{See Ref.\ \cite[p.\ 5]{Hwa.Teplitz.1966} and an
  unpublished paper by Hepp cited there.} for the pinch/singularity
analysis, but it is easier not to have to bother with the issue.

Since our analysis concerns pinches and singularities in the
physical-region, $p$ is real.  It is well-known that the only
physical-region singularity in $\Pi(p,m)$ is the normal threshold singularity
at $p^2=4m^2$. In the Coleman-Norton analysis, this corresponds to a pinch
where $k=p/2$, so that both lines are on-shell, and have equal momenta.
The corresponding classical process corresponds to two particles of
momentum $p/2$ starting at the same point.  The particles propagate for an
arbitrary time, with both of them having the same trajectory, so that they
can then recombine again.

To derive this result explicitly, we first observe that the integrand is
only singular at values of $k$ where one or both lines is on-shell.  There
can be no pinch at other values of $k$.  

There are two cases. The first is that only one line is on-shell, for which
we choose the first propagator.  Let $k=k_1$ be the position of one zero of
the denominator.  Then we can avoid the corresponding singularity by a
contour deformation $k\to k_R +i\lambda v(k_R)$ such that $v(k_1)\cdot k_1 > 0$.  It is
always possible to find such a $v(k_1)$ for a non-zero value of $k_1$,
e.g., $v(k_1)=k_1$ for a massive on-shell momentum.  Then the imaginary
part of the denominator is positive and the singularity is avoided in a
neighborhood of $k_R=k_1$.

The second case is when both propagators are on-shell, i.e., at the
position $k=k_1$, so that $k_1^2=(p-k_1)^2=m^2$.  The imaginary parts of
the denominators are $2\lambda v(k_1)\cdot k_1$ and $2\lambda v(k_1)\cdot(k_1-p)$.  Here $2k_1$
and $2(k_1-p)$ are the (non-zero) derivatives of the denominators.  If the
two derivatives are linearly independent, then $v(k_1)$ can be chosen to
make both imaginary parts positive, and hence the singularity is avoided,
i.e., there is no pinch.  

This cannot necessarily be done if the two vectors are proportional to each
other, i.e., $k_1 = c(k_1-p)$ for some non-zero number $c$.  If $c$ is
positive, then by a choice of $v(k_1)$, e.g., $v(k_1)=k_1$, we can make
both imaginary parts positive.  But if $c$ is negative, then the signs of
the imaginary parts of the denominators are always opposite.  If one is
positive, the other is negative and the singularities are not avoided.

It is conceivable that a cunningly chosen $k$-dependent deformation obeying
$v(k_1)\cdot k_1=0$ could avoid both singularities, i.e., an anomalous
deformation as treated in App.\ \ref{sec:2D.first.order} for an unpinched
situation.  But when $c$ is negative, this possibility is ruled out in the
course of our proof in Sec.\ \ref{sec:deal.with.zero.first.order}.  

Therefore the contour is pinched if and only if both denominators are zero
and $k_1 = c(k_1-p)$ for negative $c$.  We can rewrite this in the standard
form of a Landau condition as $k_1+a(k_1-p)=0$ with $a=-c>0$.  

Solving this equation and the on-shell conditions gives $a=1$ and
$k_1=p-k_1=p/2$.  Hence $p^2=4m^2$ which is the standard normal
threshold. 

Observe that given $p$, there is exactly a single pinch point. 

Applying the Feynman parameter method gives
\begin{multline}
  \Pi(p,m) = \int_0^1 \diff{\alpha} \int \frac{\diff[n]{k}}{(2\pi)^n}
\\
             \frac{1}{ \left( k^2-2k\cdot p \alpha + p^2\alpha-m^2+i\epsilon \right)^2 }.
\end{multline}
With a single denominator, the Landau condition for a pinch is just that
the denominator and its first derivative are zero. This gives $p^2=4m^2$,
$\alpha=1/2$, and $k = p-k = p/2$, corresponding exactly to the previous
determination.

Observe that there are endpoint zeros of the denominator, and hence
endpoint singularities of the integrand. These occur when $(\alpha=0,k^2=m^2)$
and when $(\alpha=1,(p-k)=m^2)$.  While a contour deformation cannot change the
value of $\alpha$ at an endpoint, a deformation on $k$ suffices to avoid the
endpoint singularity of the denominator.

Finally, we can perform the $k$ integral, using a standard formula, to get 
\begin{multline}
  \Pi(p,m) = \frac{i\Gamma(2-n/2)}{(4\pi)^{n/2}}
\times\\\times
    \int_0^1 \diff{\alpha} \frac{1}{ \left( p^2\alpha(1-\alpha)-m^2+i\epsilon \right)^{2-n/2} }.
\end{multline}
Again there is a pinch when the single denominator and its first derivative
are zero.  This gives $p^2=4m^2$ and $\alpha=1/2$, but without any direct
indication of a suitable value of $k$.  Away from that case, the
singularity of the integrand can be avoided by a contour deformation. 

The lack of direct information on a value of $k$ that corresponds to the
pinch in the pure parameter integral indicates that the use of a pure
parameter representation gives less information on the momenta concerned at
a pinch compared with a representation with a momentum integral.  This is
important for many QCD applications, where the concern is not so much with
singularities of an integral as a function of external parameter(s), but
with the regions where the integration is trapped in a region of low
virtuality for some lines.

\subsection{Simple example for \texorpdfstring{$M$}{M}}
\label{sec:M.simple}

To understand and visualize the main ideas, in Sec.\ \ref{ref:sec.M}
and subsequent sections, let $V$ be a space of 2-dimensional real
column vectors, and define
\begin{align}
  D_1 &= \begin{pmatrix} 1, & 0 \end{pmatrix},\\
  D_2 &= \begin{pmatrix} 0, & 1 \end{pmatrix},\\
  D_3 &= \begin{pmatrix} -a, & -b \end{pmatrix},
\end{align}
where $a$ and $b$ are any chosen positive numbers.  Of course, these
obey
\begin{equation}
  D(a,b) = a D_1 + b D_2 + D_3 = 0
\end{equation}
i.e., the Landau condition is obeyed with $\lambda_1 = a$ and
$\lambda_2 = b$. Here, we construct for this example the main objects
used in our general proof that there is a Landau point given the
knowledge that no ``good direction'' exists.  These objects can
be used as illustrations of the steps in the general proof.

The positive space of $D_1$ and $D_2$ is
\begin{equation}
  P = \left\{ 
         \begin{pmatrix} \alpha \\ \beta \end{pmatrix}
         \mbox{ with }
         \alpha > 0, \beta > 0
      \right\}.
\end{equation}
Now with $v = (\alpha,\beta)^{\rm T}$, and $D(\3\lambda)=\lambda_1D_1
+ \lambda_2D_2 +D_3$,
\begin{equation}
  D(\3\lambda)(v) = \alpha (\lambda_1-a) + \beta (\lambda_2-b).
\end{equation}
In Eq.\ (\ref{eq:M.def}), we defined the set $M$ to be set of
$\3\lambda$ for which $D(\3\lambda)$ is negative or zero on $P$.  This
is 
\begin{equation}
\label{eq:ex.M}
  M = \left\{
         (\lambda_1, \lambda_2)
         :
         0 \leq \lambda_1 \leq a
         \mbox{ and }
         0 \leq \lambda_2 \leq b
      \right\},
\end{equation}
as illustrated in Fig.\ \ref{fig:ex.M}.

The Landau point is at the corner $(\lambda_1,\lambda_2)=(a,b)$. The construction
in Sec.\ \ref{sec:notM.M} starts from a point on the boundary between
$M$ and $\widehat{M}$, and moves along the boundary in a direction
where one of the $\lambda_j$ increases until no further increases are
possible.  We showed in general that the resulting extreme point is a
Landau point.

The various results used in this process can be illustrated by using as a
starting point $\3\lambda=(a,\lambda_2)$, with $0\leq \lambda_2<b$.  There $D(\3\lambda)=(0,\lambda_2-b)$.

\section{Difficulties with the use of Feynman parameters}

In standard derivations (e.g., Refs.\ \cite{Landau:1959fi,ELOP}) of the
Landau criterion, a common technique is the use of Feynman parameters, as
in Eqs.\ (\ref{eq:F.graph.param.mom}) and (\ref{eq:F.graph.param.only}).
This converts an integral with multiple denominators to an integral with
one denominator.  As regards non-endpoint\footnote{The generalization to
  endpoint singularities requires modifications \cite{ELOP} to the analysis
  that are straightforward. } singularities/pinches, the Landau condition
for a pinch becomes simply the condition that the single denominator and
its derivative with respect to every integration variable are zero.  In
this case, the justification of the condition as being both necessary and
sufficient for a pinch no longer has any need for the complicated
geometrical argument of Secs.\ \ref{sec:overall}--\ref{sec:landau.proof}.
The analysis of the single-denominator case can then be shortened to the
treatment in Sec.\ \ref{sec:one-denom}.

So it might be supposed that the use of Feynman parameters is a panacea for
many of the difficulties exposed in this paper.  However, this is not the
case.  In the first place, it is often useful to be able to analyze
directly what happens in the original integral.  For example, this applies
to the region analysis so pervasive in QCD, and its further elaboration in
treatments of Glauber-type regions; it also applies to the problem of
finding a good algorithm for contour deformation in a numerical integral.
Furthermore, methods for the extension to a coordinate-space analysis, such
as is summarized in Sec.\ \ref{sec:coord}, do not readily lend themselves
to the use of Feynman parameters.

In this section, I will present some simple examples that show that certain
further issues are a serious obstacle to a general-purpose use of Feynman
parameters.

One issue is that it is not all clear that a pinch in a parametric form of
an integral necessarily entails a pinch in the original integral, although
this is implicitly assumed in essentially all the standard treatments.  In
the next Sec.\ \ref{sec:F.param.massless}, I will show a counterexample
where the assumption is actually wrong; the example is simply a one-loop
massless self-energy graph.

A second issue is that for more general cases than standard relativistic
Feynman graphs, it is not always the case that integrals involving Feynman
parameters are sufficiently well-behaved to be treated by normal contour
integration.  In contrast, the Feynman parameter method was designed to be
very useful for the standard quadratic denominators in normal relativistic
Feynman graphs.  For example, the momentum integrals can be calculated
analytically, leaving an integral only over the Feynman parameters, with
rules for the denominator being found in, for example, Sec.\ 1.5 of Ref.\
\cite{ELOP}.

Other kinds of denominator do appear even in QCD, e.g., Wilson lines with
their linear denominators.  The general analysis in the present paper has
no problems in the presence of linear denominators; indeed some parts
become easier.  But the example given in Sec.\ \ref{sec:F.param.linear}
below shows that the Feynman parameter method can become pathological in
the presence of linear denominators.

\subsection{Massless self-energy graph}
\label{sec:F.param.massless}

Consider the massless version of the self energy graph that was treated in
App.\ \ref{sec:self-energy} for the massive case:
\begin{equation}
\label{eq:massless.self.energy}
  \Pi(p,0) = \int \frac{\diff[n]{k}}{(2\pi)^n}
             \frac{1}{ \left[k^2+i\epsilon\right] \left[(p-k)^2+i\epsilon\right]},
\end{equation}
and we search for conditions for a pinch.

\subsubsection{Pure momentum representation}

If both denominators in Eq.\ (\ref{eq:massless.self.energy}) are non-zero,
then there is no pinch, as before.

Next suppose the first denominator is zero, i.e., $k^2=0$, then either $k$
is non-zero and light-like, or it is zero.  If it is light-like, then we
can find a vector $v(k)$ such that $v(k)\cdot k>0$, so that we can avoid the
singularity by a contour deformation.  

But if $k$ is zero, then both the denominator and its first derivative is
zero, so we have a pinch there, independently of the value of $p$.
Similarly the other denominator gives a pinch at $k=p$.  Even though we
have a pinch, there is no singularity of the value of the integral $\Pi(p)$
unless also $p^2=0$, a well-known property.  The property of having a pinch
on a massless line in a Feynman graph when the line's momentum is zero
evidently applies to all graphs with massless lines.

Finally, if both denominators are zero, i.e., $k^2=(p-k)^2=0$, then the
same approach as in App.\ \ref{sec:self-energy} shows that when $p^2=0$,
there is a line of collinear pinches, with
\begin{equation}
  k = \alpha p \quad \mbox{with $0\leq\alpha\leq1$}.
\end{equation}
This also immediately follows from the Landau condition for the graph.  The
single-denominator pinches are at the endpoints of the collinear pinch
line.

\subsubsection{Mixed momentum-parameter  representation}

The mixed momentum-parameter form is
\begin{equation}
  \Pi(p,m) = \int_0^1 \diff{\alpha} \int \frac{\diff[n]{k}}{(2\pi)^n}
             \frac{1}{ \left( k^2-2k\cdot p \alpha + p^2\alpha+i\epsilon \right)^2 }.
\end{equation}
First, there are zeros of the denominator at the endpoints in $\alpha$, at
$(\alpha=0,k=0)$ and $(\alpha=1,k=p)$. Unlike the massive case, there is a zero
derivative of a denominator with respect to $k$ at these points.  So the
endpoint singularity can no longer be avoided, and we have a pinch.  This
reproduces the first two configurations seen in the momentum
representation.

For a non-endpoint value of $\alpha$, there is again a pinch when the
denominator and its first derivative are zero, i.e., when
\begin{equation}
  k^2-2k\cdot p \alpha + p^2\alpha=0, ~ k-p\alpha = 0, ~\mbox{and} ~ -2k\cdot p + p^2 = 0.
\end{equation}
These have non-endpoint solutions only when $p^2=0$, and then the solution
exists for all $\alpha$, with $k=p\alpha$, again reproducing the results in the
momentum representation.

\subsubsection{Pure parameter representation}

The pure-parameter form is obtained by using a standard formula for the
momentum integral: 
\begin{multline}
  \Pi(p,0) 
\\
  =  \frac{i\Gamma\xleft(2-\frac{n}{2}\right) (p^2)^{n/2-2}}{(4\pi)^{n/2}} 
 \int_0^1 \diff{\alpha} \frac{1}{ \left[ \alpha(1-\alpha)+i\epsilon \right]^{2-n/2} }. 
\end{multline}
The momentum integration has given an overall factor that is a power of
$p^2$ and that is singular at $p^2=0$ only.  But there is no pinch at all
in the $\alpha$ integral (which can in fact be performed analytically to give
the textbook result
\begin{equation}
  \Pi(p,0)
  =  \frac{i\Gamma(2-n/2)[\Gamma(n/2-1)]^2(p^2)^{n/2-2}}{(4\pi)^{n/2} \Gamma(n-2)}
  (p^2)^{n/2-2} .
\end{equation}

\subsubsection{Results}

We now see several differences compared with the massive case
\begin{itemize}
\item There is a pinch in the momentum integral for all $p$, even though
  $\Pi$ is nonsingular if $p^2\neq0$.  This is another case beyond the rather
  trivial example in App.\ \ref{sec:trap.no.sing} where the existence of a
  pinch does not entail a singularity of the integral as a function of
  external parameters.

\item In the massive case, there was a pinch at a single point of the
  integration variables. This applied in all three representations.  But in
  the massless case, there is a whole line of collinear pinches, and this
  is visible in both the momentum and momentum-parameter representations.

\item But in the pure parameter representation, there is no pinch
  correspond to the collinear singularity.  Instead the singularity of the
  integral at $p^2=0$ is in the prefactor only.

\end{itemize}

It follows that the existence of a pinch in the momentum representation
does not entail a pinch in the pure parameter representation, contrary to
what is assumed as obvious in the standard literature.  The example is not
at all exotic; it gives the simplest possible example of a collinear pinch
in a massless theory.

\subsection{Denominators linear in some momentum components}
\label{sec:F.param.linear}

In this section, I show explicitly that the Feynman-parameter method does
not readily apply to situations with denominators that depend linearly on
some or all components, e.g., with non-relativistic theories or with Wilson
lines. 

It suffices to consider one example, a non-relativistic analog of a
self-energy graph in a 2-dimensional space-time:
\begin{widetext}
\begin{equation}
  \Gamma(E,p) = \int \diff{\omega} \diff{k}
  \frac{1}
  { \left( \omega - \frac{k^2}{2m} + i\epsilon \right) 
    \left( E-\omega - \frac{(p-k)^2}{2m} + i\epsilon \right) }.
\end{equation}
Applying the Feynman parameter method and exchanging the order of the
parameter and momentum integrals gives
\begin{equation}
  \Gamma(E,p) = \int_0^1 \diff{\alpha}~ \left\{ \int \diff{k} \diff{\omega}
  \frac{1}
       { \left[ (1-2\alpha)\omega + E\alpha - \frac{1}{2m} \left( k^2 - 2\alpha pk +\alpha p^2\right) + i\epsilon \right]^2 }.
  \right\}.
\end{equation}
By standard contour-integration methods, the integral over $\omega$ is zero if
$\alpha$ is not equal to $1/2$.  To get a non-zero value, the result of
integration over $\omega$ has to be a non-trivial generalized function
(distribution) localized at $\alpha=1/2$.  In fact, using the methods of Yan
\cite{Yan:1973qf} gives the integral over $\omega$:
\begin{equation}
  \int_{-\infty}^\infty \diff{\omega} 
  \frac{1}
       { \left[ (1-2\alpha)\omega + E\alpha - \frac{1}{2m} \left( k^2 - 2\alpha pk +\alpha p^2\right) + i\epsilon \right]^2 }
  = \frac{-4i\pi m}{ 2Em - 2k^2 + 2pk - p^2 + i\epsilon } 
    \, \delta\xleft(\alpha-\frac12\right).
\end{equation}
\end{widetext}
which implies that the result of performing the momentum integrals is not
of the form of a normal integral such as (\ref{eq:F.graph.param.only}).
(One could also say that if one tried restricting the integrals to
conventional ones, then the exchange of order of integration is not
allowed, contrary to the almost universally assumed situation for
relativistic graphs.)  In a sense, the existence of a $\delta(\alpha-\frac12)$
implies that the pure parameter integral (after performing the $\omega$ integral
and possibly the $k$ integral) always has a pinch at $\alpha=\frac12$
independently of whether there is a pinch in the original momentum
integral. 

In contrast, if the analysis in the present paper is applied directly to
the momentum space integral, one finds that there is a physical region
pinch (and singularity) if and only if $E=p^2/(4m)$, i.e., the external
energy-momentum corresponds to a particle of double the mass of that for
the individual lines.  This is easily verified by performing the momentum
integral analytically.  The pinch is at $(\omega,k) = \frac12(E,p)$.


\bibliography{jcc}

\providecommand{\noopsort}[1]{}
\begin{thebibliography}{22}%
\makeatletter
\providecommand \@ifxundefined [1]{%
 \@ifx{#1\undefined}
}%
\providecommand \@ifnum [1]{%
 \ifnum #1\expandafter \@firstoftwo
 \else \expandafter \@secondoftwo
 \fi
}%
\providecommand \@ifx [1]{%
 \ifx #1\expandafter \@firstoftwo
 \else \expandafter \@secondoftwo
 \fi
}%
\providecommand \natexlab [1]{#1}%
\providecommand \enquote  [1]{``#1''}%
\providecommand \bibnamefont  [1]{#1}%
\providecommand \bibfnamefont [1]{#1}%
\providecommand \citenamefont [1]{#1}%
\providecommand \href@noop [0]{\@secondoftwo}%
\providecommand \href [0]{\begingroup \@sanitize@url \@href}%
\providecommand \@href[1]{\@@startlink{#1}\@@href}%
\providecommand \@@href[1]{\endgroup#1\@@endlink}%
\providecommand \@sanitize@url [0]{\catcode `\\12\catcode `\$12\catcode
  `\&12\catcode `\#12\catcode `\^12\catcode `\_12\catcode `\%12\relax}%
\providecommand \@@startlink[1]{}%
\providecommand \@@endlink[0]{}%
\providecommand \url  [0]{\begingroup\@sanitize@url \@url }%
\providecommand \@url [1]{\endgroup\@href {#1}{\urlprefix }}%
\providecommand \urlprefix  [0]{URL }%
\providecommand \Eprint [0]{\href }%
\providecommand \doibase [0]{https://doi.org/}%
\providecommand \selectlanguage [0]{\@gobble}%
\providecommand \bibinfo  [0]{\@secondoftwo}%
\providecommand \bibfield  [0]{\@secondoftwo}%
\providecommand \translation [1]{[#1]}%
\providecommand \BibitemOpen [0]{}%
\providecommand \bibitemStop [0]{}%
\providecommand \bibitemNoStop [0]{.\EOS\space}%
\providecommand \EOS [0]{\spacefactor3000\relax}%
\providecommand \BibitemShut  [1]{\csname bibitem#1\endcsname}%
\let\auto@bib@innerbib\@empty
\bibitem [{\citenamefont {Landau}(1959)}]{Landau:1959fi}%
  \BibitemOpen
  \bibfield  {author} {\bibinfo {author} {\bibfnamefont {L.}~\bibnamefont
  {Landau}},\ }\bibfield  {title} {\bibinfo {title} {On analytic properties of
  vertex parts in quantum field theory},\ }\href
  {https://doi.org/10.1016/0029-5582(59)90154-3} {\bibfield  {journal}
  {\bibinfo  {journal} {Nucl. Phys.}\ }\textbf {\bibinfo {volume} {13}},\
  \bibinfo {pages} {181} (\bibinfo {year} {1959})}\BibitemShut {NoStop}%
\bibitem [{\citenamefont {Libby}\ and\ \citenamefont
  {Sterman}(1978)}]{Libby:1978bx}%
  \BibitemOpen
  \bibfield  {author} {\bibinfo {author} {\bibfnamefont {S.~B.}\ \bibnamefont
  {Libby}}\ and\ \bibinfo {author} {\bibfnamefont {G.}~\bibnamefont
  {Sterman}},\ }\bibfield  {title} {\bibinfo {title} {Mass divergences in
  two-particle inelastic scattering},\ }\href
  {https://doi.org/10.1103/PhysRevD.18.4737} {\bibfield  {journal} {\bibinfo
  {journal} {Phys. Rev.}\ }\textbf {\bibinfo {volume} {D18}},\ \bibinfo {pages}
  {4737} (\bibinfo {year} {1978})}\BibitemShut {NoStop}%
\bibitem [{\citenamefont {Collins}(2011)}]{Collins:2011qcdbook}%
  \BibitemOpen
  \bibfield  {author} {\bibinfo {author} {\bibfnamefont {J.~C.}\ \bibnamefont
  {Collins}},\ }\href {https://doi.org/10.1017/CBO9780511975592} {\emph
  {\bibinfo {title} {Foundations of Perturbative QCD}}}\ (\bibinfo  {publisher}
  {Cambridge University Press},\ \bibinfo {address} {Cambridge},\ \bibinfo
  {year} {2011})\BibitemShut {NoStop}%
\bibitem [{\citenamefont {Coleman}\ and\ \citenamefont
  {Norton}(1965)}]{Coleman:1965xm}%
  \BibitemOpen
  \bibfield  {author} {\bibinfo {author} {\bibfnamefont {S.}~\bibnamefont
  {Coleman}}\ and\ \bibinfo {author} {\bibfnamefont {R.~E.}\ \bibnamefont
  {Norton}},\ }\bibfield  {title} {\bibinfo {title} {Singularities in the
  physical region},\ }\href@noop {} {\bibfield  {journal} {\bibinfo  {journal}
  {Nuovo Cim.}\ }\textbf {\bibinfo {volume} {38}},\ \bibinfo {pages} {438}
  (\bibinfo {year} {1965})}\BibitemShut {NoStop}%
\bibitem [{\citenamefont {Ma}(2020)}]{Ma:2019hjq}%
  \BibitemOpen
  \bibfield  {author} {\bibinfo {author} {\bibfnamefont {Y.}~\bibnamefont
  {Ma}},\ }\bibfield  {title} {\bibinfo {title} {{A Forest Formula to Subtract
  Infrared Singularities in Amplitudes for Wide-angle Scattering}},\ }\href
  {https://doi.org/10.1007/JHEP05(2020)012} {\bibfield  {journal} {\bibinfo
  {journal} {JHEP}\ }\textbf {\bibinfo {volume} {05}},\ \bibinfo {pages}
  {012}},\ \Eprint {https://arxiv.org/abs/1910.11304} {arXiv:1910.11304
  [hep-ph]} \BibitemShut {NoStop}%
\bibitem [{\citenamefont {Eden}\ \emph {et~al.}(1966)\citenamefont {Eden},
  \citenamefont {Landshoff}, \citenamefont {Olive},\ and\ \citenamefont
  {Polkinghorne}}]{ELOP}%
  \BibitemOpen
  \bibfield  {author} {\bibinfo {author} {\bibfnamefont {R.~J.}\ \bibnamefont
  {Eden}}, \bibinfo {author} {\bibfnamefont {P.~V.}\ \bibnamefont {Landshoff}},
  \bibinfo {author} {\bibfnamefont {D.~I.}\ \bibnamefont {Olive}},\ and\
  \bibinfo {author} {\bibfnamefont {J.~C.}\ \bibnamefont {Polkinghorne}},\
  }\href@noop {} {\emph {\bibinfo {title} {The Analytic {S}-matrix}}}\
  (\bibinfo  {publisher} {Cambridge University Press},\ \bibinfo {address}
  {Cambridge},\ \bibinfo {year} {1966})\BibitemShut {NoStop}%
\bibitem [{\citenamefont {Fotiadi}\ \emph {et~al.}(1965)\citenamefont
  {Fotiadi}, \citenamefont {Froissart}, \citenamefont {Lascoux},\ and\
  \citenamefont {Pham}}]{Fotiadi:1965}%
  \BibitemOpen
  \bibfield  {author} {\bibinfo {author} {\bibfnamefont {D.}~\bibnamefont
  {Fotiadi}}, \bibinfo {author} {\bibfnamefont {M.}~\bibnamefont {Froissart}},
  \bibinfo {author} {\bibfnamefont {J.}~\bibnamefont {Lascoux}},\ and\ \bibinfo
  {author} {\bibfnamefont {F.}~\bibnamefont {Pham}},\ }\bibfield  {title}
  {\bibinfo {title} {Applications of an isotopy theorem},\ }\href@noop {}
  {\bibfield  {journal} {\bibinfo  {journal} {Topology}\ }\textbf {\bibinfo
  {volume} {4}},\ \bibinfo {pages} {159} (\bibinfo {year} {1965})}\BibitemShut
  {NoStop}%
\bibitem [{\citenamefont {Hwa}\ and\ \citenamefont
  {Teplitz}(1966)}]{Hwa.Teplitz.1966}%
  \BibitemOpen
  \bibfield  {author} {\bibinfo {author} {\bibfnamefont {R.~C.}\ \bibnamefont
  {Hwa}}\ and\ \bibinfo {author} {\bibfnamefont {V.~L.}\ \bibnamefont
  {Teplitz}},\ }\href@noop {} {\emph {\bibinfo {title} {Homology and {F}eynman
  integrals}}}\ (\bibinfo  {publisher} {Benjamin},\ \bibinfo {year}
  {1966})\BibitemShut {NoStop}%
\bibitem [{\citenamefont {Pham}(2011)}]{Pham:2011}%
  \BibitemOpen
  \bibfield  {author} {\bibinfo {author} {\bibfnamefont {F.}~\bibnamefont
  {Pham}},\ }\href@noop {} {\emph {\bibinfo {title} {Singularities of
  integrals: {H}omology, hyperfunctions and microlocal analysis}}}\ (\bibinfo
  {publisher} {Springer-Verlag London},\ \bibinfo {year} {2011})\BibitemShut
  {NoStop}%
\bibitem [{\citenamefont {Gallier}(2011)}]{Gallier:2008}%
  \BibitemOpen
  \bibfield  {author} {\bibinfo {author} {\bibfnamefont {J.}~\bibnamefont
  {Gallier}},\ }\bibfield  {title} {\bibinfo {title} {Notes on convex sets,
  polytopes, polyhedra, combinatorial topology, {V}oronoi diagrams and
  {D}elaunay triangulations},\ }\Eprint {https://arxiv.org/abs/0805.0292}
  {arXiv:0805.0292 [math.GM]} \BibitemShut {NoStop}%
\bibitem [{\citenamefont {Gallier}(2000)}]{Gallier:2011}%
  \BibitemOpen
  \bibfield  {author} {\bibinfo {author} {\bibfnamefont {J.~H.}\ \bibnamefont
  {Gallier}},\ }\href@noop {} {\emph {\bibinfo {title} {Geometric Methods and
  Applications, for Computer Science and Engineering}}}\ (\bibinfo  {publisher}
  {Springer},\ \bibinfo {year} {2000})\BibitemShut {NoStop}%
\bibitem [{\citenamefont {Brodsky}\ \emph {et~al.}(2019)\citenamefont
  {Brodsky}, \citenamefont {Schmidt},\ and\ \citenamefont
  {Liuti}}]{Brodsky:2019jla}%
  \BibitemOpen
  \bibfield  {author} {\bibinfo {author} {\bibfnamefont {S.~J.}\ \bibnamefont
  {Brodsky}}, \bibinfo {author} {\bibfnamefont {I.}~\bibnamefont {Schmidt}},\
  and\ \bibinfo {author} {\bibfnamefont {S.}~\bibnamefont {Liuti}},\ }\bibfield
   {title} {\bibinfo {title} {{Is the Momentum Sum Rule Valid for Nuclear
  Structure Functions?}},\ }\Eprint {https://arxiv.org/abs/1908.06317}
  {arXiv:1908.06317 [hep-ph]} \BibitemShut {NoStop}%
\bibitem [{\citenamefont {Gong}\ \emph {et~al.}(2009)\citenamefont {Gong},
  \citenamefont {Nagy},\ and\ \citenamefont {Soper}}]{Gong:2008ww}%
  \BibitemOpen
  \bibfield  {author} {\bibinfo {author} {\bibfnamefont {W.}~\bibnamefont
  {Gong}}, \bibinfo {author} {\bibfnamefont {Z.}~\bibnamefont {Nagy}},\ and\
  \bibinfo {author} {\bibfnamefont {D.~E.}\ \bibnamefont {Soper}},\ }\bibfield
  {title} {\bibinfo {title} {{Direct numerical integration of one-loop
  {F}eynman diagrams for {$N$}-photon amplitudes}},\ }\href
  {https://doi.org/10.1103/PhysRevD.79.033005} {\bibfield  {journal} {\bibinfo
  {journal} {Phys. Rev.}\ }\textbf {\bibinfo {volume} {D79}},\ \bibinfo {pages}
  {033005} (\bibinfo {year} {2009})},\ \Eprint
  {https://arxiv.org/abs/0812.3686} {arXiv:0812.3686 [hep-ph]} \BibitemShut
  {NoStop}%
\bibitem [{\citenamefont {Becker}\ and\ \citenamefont
  {Weinzierl}(2012)}]{Becker:2012nk}%
  \BibitemOpen
  \bibfield  {author} {\bibinfo {author} {\bibfnamefont {S.}~\bibnamefont
  {Becker}}\ and\ \bibinfo {author} {\bibfnamefont {S.}~\bibnamefont
  {Weinzierl}},\ }\bibfield  {title} {\bibinfo {title} {{Direct contour
  deformation with arbitrary masses in the loop}},\ }\href
  {https://doi.org/10.1103/PhysRevD.86.074009} {\bibfield  {journal} {\bibinfo
  {journal} {Phys. Rev.}\ }\textbf {\bibinfo {volume} {D86}},\ \bibinfo {pages}
  {074009} (\bibinfo {year} {2012})},\ \Eprint
  {https://arxiv.org/abs/1208.4088} {arXiv:1208.4088 [hep-ph]} \BibitemShut
  {NoStop}%
\bibitem [{\citenamefont {Becker}\ and\ \citenamefont
  {Weinzierl}(2013)}]{Becker:2012bi}%
  \BibitemOpen
  \bibfield  {author} {\bibinfo {author} {\bibfnamefont {S.}~\bibnamefont
  {Becker}}\ and\ \bibinfo {author} {\bibfnamefont {S.}~\bibnamefont
  {Weinzierl}},\ }\bibfield  {title} {\bibinfo {title} {{Direct numerical
  integration for multi-loop integrals}},\ }\href
  {https://doi.org/10.1140/epjc/s10052-013-2321-1} {\bibfield  {journal}
  {\bibinfo  {journal} {Eur. Phys. J.}\ }\textbf {\bibinfo {volume} {C73}},\
  \bibinfo {pages} {2321} (\bibinfo {year} {2013})},\ \Eprint
  {https://arxiv.org/abs/1211.0509} {arXiv:1211.0509 [hep-ph]} \BibitemShut
  {NoStop}%
\bibitem [{\citenamefont {Soper}(2000)}]{Soper:1999xk}%
  \BibitemOpen
  \bibfield  {author} {\bibinfo {author} {\bibfnamefont {D.~E.}\ \bibnamefont
  {Soper}},\ }\bibfield  {title} {\bibinfo {title} {{Techniques for {QCD}
  calculations by numerical integration}},\ }\href
  {https://doi.org/10.1103/PhysRevD.62.014009} {\bibfield  {journal} {\bibinfo
  {journal} {Phys. Rev.}\ }\textbf {\bibinfo {volume} {D62}},\ \bibinfo {pages}
  {014009} (\bibinfo {year} {2000})},\ \Eprint
  {https://arxiv.org/abs/hep-ph/9910292} {arXiv:hep-ph/9910292 [hep-ph]}
  \BibitemShut {NoStop}%
\bibitem [{\citenamefont {Collins}\ and\ \citenamefont
  {Qiu}(2007)}]{Collins:2007nk}%
  \BibitemOpen
  \bibfield  {author} {\bibinfo {author} {\bibfnamefont {J.}~\bibnamefont
  {Collins}}\ and\ \bibinfo {author} {\bibfnamefont {J.-W.}\ \bibnamefont
  {Qiu}},\ }\bibfield  {title} {\bibinfo {title} {{$k_{T}$} factorization is
  violated in production of high-transverse-momentum particles in hadron-hadron
  collisions},\ }\href {https://doi.org/10.1103/PhysRevD.75.114014} {\bibfield
  {journal} {\bibinfo  {journal} {Phys. Rev.}\ }\textbf {\bibinfo {volume}
  {D75}},\ \bibinfo {pages} {114014} (\bibinfo {year} {2007})},\ \Eprint
  {https://arxiv.org/abs/0705.2141} {arXiv:0705.2141 [hep-ph]} \BibitemShut
  {NoStop}%
\bibitem [{\citenamefont {Erdo\u{g}an}\ and\ \citenamefont
  {Sterman}(2015)}]{Erdogan:2014gha}%
  \BibitemOpen
  \bibfield  {author} {\bibinfo {author} {\bibfnamefont {O.}~\bibnamefont
  {Erdo\u{g}an}}\ and\ \bibinfo {author} {\bibfnamefont {G.}~\bibnamefont
  {Sterman}},\ }\bibfield  {title} {\bibinfo {title} {Ultraviolet divergences
  and factorization for coordinate-space amplitudes},\ }\href
  {https://doi.org/10.1103/PhysRevD.91.065033} {\bibfield  {journal} {\bibinfo
  {journal} {Phys. Rev.}\ }\textbf {\bibinfo {volume} {D91}},\ \bibinfo {pages}
  {065033} (\bibinfo {year} {2015})},\ \Eprint
  {https://arxiv.org/abs/1411.4588} {arXiv:1411.4588 [hep-ph]} \BibitemShut
  {NoStop}%
\bibitem [{\citenamefont {Erdo\u{g}an}\ and\ \citenamefont
  {Sterman}(2016)}]{Erdogan:2016ylj}%
  \BibitemOpen
  \bibfield  {author} {\bibinfo {author} {\bibfnamefont {O.}~\bibnamefont
  {Erdo\u{g}an}}\ and\ \bibinfo {author} {\bibfnamefont {G.}~\bibnamefont
  {Sterman}},\ }\bibfield  {title} {\bibinfo {title} {A coordinate description
  of partonic processes},\ }in\ \href
  {http://inspirehep.net/record/1419082/files/arXiv:1602.00943.pdf} {\emph
  {\bibinfo {booktitle} {{Proceedings, 12th International Symposium on
  Radiative Corrections (Radcor 2015) and LoopFest XIV (Radiative Corrections
  for the LHC and Future Colliders): Los Angeles, CA, USA, June 15-19,
  2015}}}}\ (\bibinfo {year} {2016})\ \Eprint
  {https://arxiv.org/abs/1602.00943} {arXiv:1602.00943 [hep-ph]} \BibitemShut
  {NoStop}%
\bibitem [{\citenamefont {Erdo\u{g}an}\ and\ \citenamefont
  {Sterman}(2017)}]{Erdogan:2017gyf}%
  \BibitemOpen
  \bibfield  {author} {\bibinfo {author} {\bibfnamefont {O.}~\bibnamefont
  {Erdo\u{g}an}}\ and\ \bibinfo {author} {\bibfnamefont {G.}~\bibnamefont
  {Sterman}},\ }\bibfield  {title} {\bibinfo {title} {{Path description of
  coordinate-space amplitudes}},\ }\href
  {https://doi.org/10.1103/PhysRevD.95.116015} {\bibfield  {journal} {\bibinfo
  {journal} {Phys. Rev. D}\ }\textbf {\bibinfo {volume} {95}},\ \bibinfo
  {pages} {116015} (\bibinfo {year} {2017})},\ \Eprint
  {https://arxiv.org/abs/1705.04539} {arXiv:1705.04539 [hep-th]} \BibitemShut
  {NoStop}%
\bibitem [{\citenamefont {Eden}\ \emph {et~al.}(1961)\citenamefont {Eden},
  \citenamefont {Landshoff}, \citenamefont {Polkinghorne},\ and\ \citenamefont
  {Taylor}}]{Eden:1961}%
  \BibitemOpen
  \bibfield  {author} {\bibinfo {author} {\bibfnamefont {R.~J.}\ \bibnamefont
  {Eden}}, \bibinfo {author} {\bibfnamefont {P.~V.}\ \bibnamefont {Landshoff}},
  \bibinfo {author} {\bibfnamefont {J.~C.}\ \bibnamefont {Polkinghorne}},\ and\
  \bibinfo {author} {\bibfnamefont {J.~C.}\ \bibnamefont {Taylor}},\
  }\href@noop {} {\bibfield  {journal} {\bibinfo  {journal} {J. Math. Phys.}\
  }\textbf {\bibinfo {volume} {2}},\ \bibinfo {pages} {656} (\bibinfo {year}
  {1961})}\BibitemShut {NoStop}%
\bibitem [{\citenamefont {Yan}(1973)}]{Yan:1973qf}%
  \BibitemOpen
  \bibfield  {author} {\bibinfo {author} {\bibfnamefont {T.-M.}\ \bibnamefont
  {Yan}},\ }\bibfield  {title} {\bibinfo {title} {{Quantum field theories in
  the infinite momentum frame 3. {Q}uantization of coupled spin one fields}},\
  }\href {https://doi.org/10.1103/PhysRevD.7.1760} {\bibfield  {journal}
  {\bibinfo  {journal} {Phys. Rev.}\ }\textbf {\bibinfo {volume} {D7}},\
  \bibinfo {pages} {1760} (\bibinfo {year} {1973})}\BibitemShut {NoStop}%
\end{thebibliography}%


\end{document}